\documentclass[aps,superscriptaddress,onecolumn,10pt,prx]{revtex4-2}

\usepackage{silence}
\usepackage{amsmath,mathtools,amsthm,amssymb}
\usepackage{tabularx}
\usepackage{tabularray}
\usepackage{tikz}
\usepackage{tabularx}
\usepackage{tabularray}

\usepackage[a4paper, total={6.7in, 9.7in}]{geometry}

\usepackage{graphicx}

\usepackage{color}
\usepackage{bbold}
\usepackage{enumitem}
\usepackage{centernot}
\usepackage{complexity}
\usepackage{physics}
\usepackage{comment}
\usepackage{array}
\usepackage{graphics}
\usepackage{wrapfig}
\usepackage{fontsize}
\graphicspath{{figures/}}
\usepackage{biolinum}
\usepackage{bbm}
\usepackage{dsfont}
\usepackage{mathrsfs}
\usepackage{mathdots}
\usepackage{enumitem}
\usepackage[table]{xcolor}
\usepackage{booktabs}
\usepackage{amssymb}

\definecolor{lightred}{RGB}{243,229,231}
\definecolor{lightgreen}{RGB}{241,255,239}
\definecolor{lightblue}{RGB}{232,240,244}
\definecolor{RoyalBlue}{RGB}{65,105,225}
\definecolor{ForestGreen}{RGB}{34,139,34}   
\definecolor{Maroon}{RGB}{135,0,0}
\definecolor{myrefcolor}{rgb}{0.067,0.5,0.5}
\definecolor{myurlcolor}{rgb}{0.1,0,0.9}

\usepackage[
    breaklinks,
    pdftex,
    colorlinks=true,
    linkcolor=myrefcolor,
    citecolor=myrefcolor,
    urlcolor=myrefcolor
]{hyperref}
\usepackage[capitalize]{cleveref}

\usepackage{hyperref}

\newtheorem{theorem}{Theorem}
\newtheorem{proposition}{Proposition}
\newtheorem{lemma}{Lemma}
\newtheorem{corollary}{Corollary}
\newtheorem{definition}{Definition}

\newtheorem{problem}{Problem}

\newcommand{\nocontentsline}[3]{}
\let\origcontentsline\addcontentsline
\newcommand\stoptoc{\let\addcontentsline\nocontentsline}
\newcommand\resumetoc{\let\addcontentsline\origcontentsline}
\begin{document}


\title{The unbearable hardness of deciding about magic}

\author{Lorenzo Leone}
\thanks{\href{mailto:loleone@unisa.it}{loleone@unisa.it}}


\affiliation{Dipartimento di Ingegneria Industriale, Universit\`a degli Studi di Salerno, Via Giovanni Paolo II, 132, 84084 Fisciano (SA), Italy}
\affiliation{Dahlem Center for Complex Quantum Systems, Freie Universit\"at Berlin, 14195 Berlin, Germany}
\author{Jens Eisert}
\thanks{\href{mailto:jense@zedat.fu-berlin.de}{jense@zedat.fu-berlin.de}}
\affiliation{Dahlem Center for Complex Quantum Systems, Freie Universit\"at Berlin, 14195 Berlin, Germany}

\author{Salvatore F.~E.~Oliviero}
\thanks{\href{mailto:s.oliviero@sns.it}{s.oliviero@fu-berlin.de}}
\affiliation{Dahlem Center for Complex Quantum Systems, Freie Universit\"at Berlin, 14195 Berlin, Germany}

\begin{abstract}
Identifying the boundary between classical and quantum computation is a central challenge in quantum information. In multi-qubit systems, entanglement and magic are the key resources underlying genuinely quantum behaviour. While entanglement is well understood, magic---essential for universal quantum computation---remains relatively poorly characterised. Here we show that determining membership in the stabilizer polytope, which defines the free states of magic-state resource theory, requires super-exponential time $\class{exp} ( n^2)$ in the number of qubits $n$, even approximately. We give a reduction from solving a $3$-\class{SAT} instance on $n^2$ variables and, by invoking the exponential time hypothesis, the result follows. As a consequence, both quantifying and certifying magic are fundamentally intractable: any magic monotone for general states must be super-exponentially hard to compute, and deciding whether an operator is a valid magic witness is equally difficult. As a corollary, we establish the robustness of magic as computationally optimal among monotones. This barrier extends even to classically simulable regimes: deciding whether a state lies in the convex hull of states generated by a logarithmic number of non-Clifford gates is also super-exponentially hard. Together, these results reveal intrinsic computational limits on assessing classical simulability, distilling pathological magic states, and ultimately probing and exploiting magic as a quantum resource.
\end{abstract}

\maketitle
\stoptoc
\section{Introduction}

Quantum computers promise the efficient solution of some computational problems that are out of reach for efficient classical computers \cite{Shor-1994,AshleyOverview}. For a long time, this idea was mainly seen as an inspiring theoretical possibility. In recent years, however—largely driven by major efforts in the quantum industry—it has started to become an experimental reality. The quantum computers available today are still limited in size and remain noisy, with devices operating at more than 1000 qubits. Even so, machines of this scale, with two-qubit error rates close to a tenth of a percent, would have been hard to imagine not long ago. This rapid progress has created an exciting situation for the field \cite{MindTheGaps,Myths,GrandChallenge}. The devices we have today, and those expected in the near future, are commonly referred to as \emph{noisy intermediate-scale quantum} (NISQ) devices \cite{preskill_quantum_2018}. These devices may enable useful quantum applications, and this question is currently the focus of intense research. At the same time, many quantum protocols implemented on such hardware can still be simulated efficiently on classical computers using modern methods. This process, often called ``dequantization'' \cite{AntonioPauliPropagation,Nonunital,ManuelPauliPropagation,zhou_what_2020-1}, removes the expected quantum advantage. One of the central questions in the field today is what can truly be achieved with these devices, and to what extent they still operate in a regime that is effectively classical. At a more conceptual level, the key issue is to understand where the boundary lies between the quantum regime—where one may expect a genuine quantum advantage \cite{MindTheGaps,Myths,GrandChallenge}—and the classical regime, which still allows for efficient classical description. Identifying this boundary is one of the main challenges of the field.

A particularly clean and rigorous way to formalise this boundary is provided by quantum resource theories~\cite{chitambarQuantumResourceTheories2019,gourResourcesQuantumWorld2024,horodeckiQuantumEntanglement2009}. In the context of multi-qubit quantum computation, two key resources emerge as necessary for a state to lie on the quantum side of the elusive classical–quantum boundary: entanglement and magic~\cite{horodeckiQuantumEntanglement2009, veitchResourceTheoryStabilizer2014}. However, while the resource theory of entanglement \cite{horodeckiQuantumEntanglement2009} can be regarded as being extremely well understood,
the same cannot quite be said for the resource theory of magic, largely because it has only recently attracted serious attention~\cite{veitchResourceTheoryStabilizer2014,liuManybodyQuantumMagic2022,leoneStabilizerRenyiEntropy2022}.
Magic-state resource theory formalises the observation that Clifford operations and stabilizer states, although genuinely quantum, are efficiently classically simulable~\cite{gottesmanHeisenbergRepresentationQuantum1998,PhysRevA.70.052328}, using the rigorous language of quantum resource theory. Within this framework, states confined to the stabilizer polytope---the region of Hilbert space reachable by classical mixtures of stabilizer states---constitute the {\em free states}, while quantum operations that map the stabilizer polytope onto itself define the {\em free operations} of the theory~\cite{veitchResourceTheoryStabilizer2014,liuManybodyQuantumMagic2022}. Exiting the stabilizer polytope requires \emph{magic}, the ingredient that pushes quantum states outside the classically simulable regime \cite{PhysRevA.70.052328,StabilizerPolytope,PhysRevA.95.022316,howardApplicationResourceTheory2017,regulaConvexGeometryQuantum2018,seddonQuantifyingMagicMultiqubit2019,Wigner}. Magic is, in a sense, the conceptual feature that makes quantum systems ``quantum.'' In more technological terms, it is the ingredient that uplifts schemes based on transversal Clifford circuits to fully universal fault-tolerant quantum computing schemes \cite{Wigner}. 

In this work, we bring this line of research in a striking new direction: we show that deciding whether a quantum state contains magic necessarily requires algorithms whose runtime grows super-exponentially with the number of qubits $n$. More precisely, we formulate a mathematically well-defined decision problem —namely, a weak membership problem for the stabilizer polytope— and reduce it from a $3$-\class{SAT} instance on $n^2$ variables. Invoking the exponential time hypothesis, this reduction implies that solving the problem requires $\class{exp}(n^2)$ time. 
While the technical argument is involved, the upshot is that we can relate membership in the stabilizer polytope to testing properties of \emph{graph states} \cite{Graphs,Hein06}. 

This result opens an exciting new direction for this line of research. It is no surprise that identifying the subtle boundary of the stabilizer polytope is difficult: it is a provably computationally hard task that adheres to the rigour of theoretical computer science. While this insight is primarily of conceptual value, we will later show that it also has important practical implications for magic-state resource theory, including quantifying and detecting magic, classical simulation of noisy quantum circuits, as well as resource distillation in modern quantum error correction.

\begin{figure}[tbh]
  \centering
  \includegraphics[width=\linewidth]{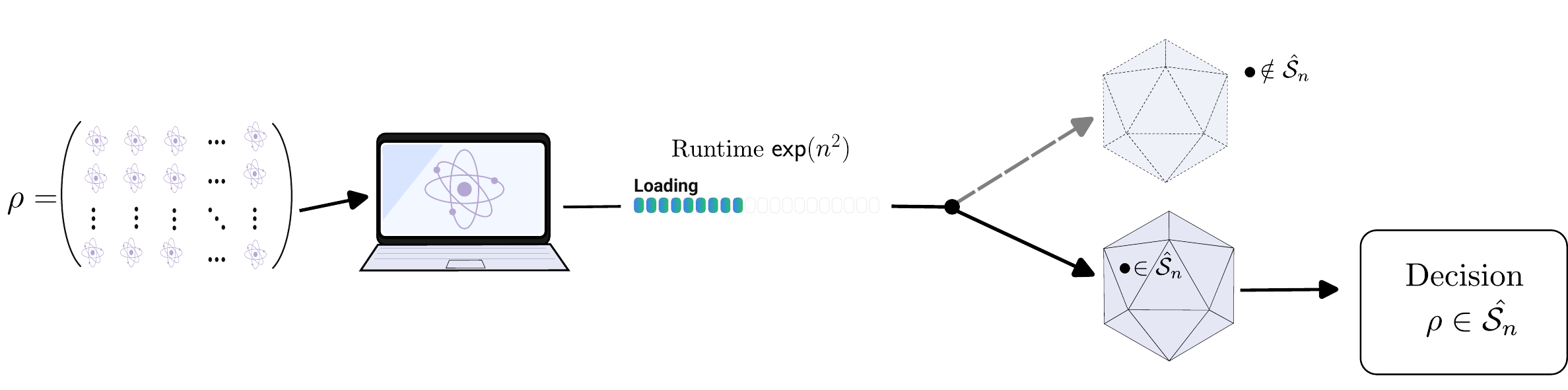}
  \caption{In this work, we show that any algorithm aimed at determining whether a $d \times d$ density matrix on $n$ qubits, $\rho$, belongs to the stabilizer polytope must necessarily run in super-exponential time. Consequently, any magic monotone also requires super-exponential time to compute. Concretely, we prove that the membership problem in the stabilizer polytope, formulated as a decision problem, lies in the complexity class $\class{QP}^2$. 
  Also depicted in the upper picture is a schematic representation of a witness as a hyperplane separating a quantum state from the stabilizer polytope.
  At a higher level, it demonstrates that delineating the ``classical'' setting—where efficient classical simulation is possible, inside the polytope—from the ``quantum'' setting—where one can hope for computational quantum advantages—is itself a computationally difficult task.
  }
  \label{fig:complexity}
\end{figure}

\section{Results}

We now state the main conclusion of this work and highlight its implications: deciding whether a quantum state contains magic—whether by quantifying it, detecting it, or determining its membership in the classically simulable region defined by the stabilizer formalism—is computationally hard. 

To set the stage in more formal terms, we take the Hilbert space dimension $d = 2^n$ as the input size of our problem, since we are given a $d \times d$ density matrix $\rho$ describing a (possibly mixed) quantum state. We denote by $\class{P}$ the class of decision problems that can be solved in time polynomial in $d$~\cite{aroraComputationalComplexityModern2009}. For instance, deciding the purity of $\rho$ clearly lies in $\class{P}$. We denote by $\class{NP}$ the class of problems for which there exists a polynomial-time deterministic Turing machine that can verify \class{YES} instance of the decision problem.
For our purposes, it is useful to introduce \emph{quasi-polynomial} complexity classes. For $k \in \mathbb{R}_{+}$, a decision problem with input size $d$ is said to lie in $\class{QP}^{k}$ if there exists a deterministic Turing machine solving it in time $\class{exp}({\log^{k} d)}$. Note that problems in $\class{QP}^{k}$ with $k>1$ have a strictly {\em super-exponential} runtime with respect to the number of qubits $n$; nevertheless, we will refer to them as {\em quasi-polynomial} because the scaling is considered relative to the Hilbert space dimension $d$. These classes can be separated from one another under one of the most widely believed assumptions in computational complexity: the \emph{exponential time hypothesis} (ETH)~\cite{impagliazzoComplexity$k$SAT2001} which we will assume throughout this work. ETH asserts that $\class{NP}$-complete problems cannot be solved in subexponential time. While the hierarchy $\class{QP}^{k}$ for different values of $k$ offers a fine-grained view of quasi-polynomial behaviour, the union $\bigcup_{k\in\mathbb{N}} \class{QP}^{k}$ captures exactly all problems that can be solved in strictly subexponential time with respect to the input size $d$.

\subsection{Quantifying magic}\label{sec:quantifying}

In close analogy with entanglement theory~\cite{guhneEntanglementDetection2009}, there are two complementary approaches to accessing magic: quantification and witnessing. Both aim to characterise the boundary of the stabilizer polytope. Quantifying magic is the more demanding task: it seeks to distinguish magic states exactly from stabilizer states---i.e., to identify the precise boundary of the stabilizer polytope in Hilbert space---and to assign a meaningful measure of how much magic a state contains. This is formalised through magic monotones, scalar functions that vanish exactly on states inside the stabilizer polytope~\cite{howardApplicationResourceTheory2017,regulaConvexGeometryQuantum2018,seddonQuantifyingMagicMultiqubit2019}.
The scarcity of computable monotones is the main reason for the slow progress in magic-state resource theory. Although many candidates have been proposed, they are notoriously difficult to evaluate~\cite{liuManybodyQuantumMagic2022}, typically requiring minimisation over all pure stabilizer states---an inherently super-polynomial task in the Hilbert space dimension. Stabilizer entropies~\cite{leoneStabilizerRenyiEntropy2022,leoneStabilizerEntropiesAre2024,bittelOperationalInterpretationStabilizer2025,olivieroMeasuringMagicQuantum2022} partially alleviate this by avoiding such minimisations and being computable from expectation values of Hermitian operators in polynomial time. However, they behave reliably only for pure states; extending them to mixed states again requires an optimisation procedure.
These difficulties suggest intrinsic computational hardness in magic-state resource theory and naturally raise the question of whether any efficiently computable magic monotone for mixed states can exist at all---or, more modestly, whether the boundary of the stabilizer polytope can be determined exactly and efficiently. 

The main result of this work answers both questions negatively: we show that there is no efficient method to precisely characterise the boundary of the stabilizer polytope, and consequently, any magic monotone must necessarily require super-polynomial time to compute. To formalise these questions, we define the weak membership problem for the stabilizer polytope $\hat{\mathcal{S}}$: given a state $\rho$, decide whether it is a free state, i.e., $\rho \in \hat{\mathcal{S}}$, or whether it is $\varepsilon$-far from any free state. In the following informal theorem, we show that this membership problem lies in $\class{QP}^2$. 
%

\begin{theorem}[Membership in the stabilizer polytope is hard. Informal of \cref{thm:nphardrhowmem,wmemeforstabinqp2}] \label{th1:membershiphard} Deciding whether a state 
$\rho$ lies in the stabilizer polytope $\hat{\mathcal{S}}$, or is $\varepsilon$-far from every state in $\hat{\mathcal{S}}$, belongs to the complexity class $\class{QP}^2$ for any $\varepsilon = 1/\class{poly}(d)$.
\end{theorem}
\noindent
{\em Proof sketch.} The proof has two parts: showing that the decision problem lies in $\class{QP}^{2}$, and that it does not belong to $\class{QP}^{2-\eta}$ for any $\eta>0$. Membership in $\class{QP}^{2}$ follows from computing the 
known magic monotone, the \textit{robustness of magic}, which can be evaluated in time $\class{exp}(\log^2d)$ using the argument in Ref.~\cite{howardApplicationResourceTheory2017}, rigorously refined in \cref{wmemeforstabinqp2}. The second, more technical part encodes an instance of the polytope membership problem as a $3$-\class{SAT} instance in $O(\log^2d)$ variables. By the ETH, this cannot be solved in time $\class{exp}(\log^{2-\eta}d)$ for any $\eta>0$. \qed

Therefore, even the simple task of classifying a state as free or resourceful in magic-state resource theory requires super-polynomial resources, highlighting the inherent hardness of the theory. Since the golden rule of a magic monotone is to \textit{faithfully} identify the set of free states---i.e., those inside the stabilizer polytope $\hat{\mathcal{S}}$---it follows as a simple corollary that the computational difficulty of every magic monotone valid for mixed states is not accidental, but an intrinsic feature of magic-state resource theory.

\begin{corollary}[No efficient magic monotone]\label{cor1main}
The computation of any magic monotone $\mathcal{M}$ requires time $\class{exp}(\log^2d)$.
\end{corollary}
At this point, some remarks are in order. Many magic monotones have been proposed so far, and as anticipated, some of them require minimisation over the set of pure stabilizer states.  Since the set of stabilizer states contains $\class{exp}(\log^2d)$ elements, such an optimisation---even when convex---can only be performed by an algorithm scaling as $\class{exp}(\log^2d)$. From \cref{cor1main}, we see that these measures are indeed {\em optimal}, as one cannot provably do better within magic-state resource theory. One example is the \emph{robustness of magic}~\cite{howardApplicationResourceTheory2017}, defined as the minimum negativity of the coefficients when expressing $\rho$ as a linear combination of pure stabilizer states.

The computational difficulty of magic monotones established by \cref{cor1main} applies to monotones designed to detect magic in {\em arbitrary mixed states}. Indeed, stabilizer entropies~\cite{leoneStabilizerRenyiEntropy2022} are efficiently computable, but they faithfully identify free states only when these states are pure. In Ref.~\cite{leoneStabilizerEntropiesAre2024}, it has been shown that stabilizer entropies can be extended to general mixed states via the standard convex roof construction. While this extension involves an optimisation over all convex decompositions of a mixed state $\rho$ into pure states---a task even more computationally demanding than optimising over all stabilizer states---restricting the rank of $\rho$ can substantially reduce this complexity. Specifically, Ref.~\cite{leoneStabilizerEntropiesAre2024} observes that computing the extended 
stabilizer entropy requires 
an 
$\class{exp}(\log d\operatorname{rank}(\rho)^2)$-time algorithm, placing it within $\class{P}$ for states with constant rank. This shows that, although magic-state resource theory is super-polynomially hard for general mixed states, restricting to low-rank states---common, for example, in matrix product states in many-body physics---allows stabilizer entropies to be practically useful for exploring moderately large systems, far beyond the $n \simeq 5$ limit implied by the full super-polynomial hardness of $\class{QP}^2$.

Finally, it is instructive to compare our results with other resource theories, highlighting that this intrinsic computational hardness is a distinctive feature of magic-state resource theory. In the resource theory of coherence, the structure is simple enough that monotones are efficiently computable \cite{PhysRevLett.113.140401}.
The same holds true for Gaussian continuous-variable entanglement, where even all tangent hyperplanes to the convex set of covariance matrices of bosonic Gaussian separable states can be classified \cite{CVWitness}.
Similarly, despite the non-convex nature of 
fermionic non-Gaussianity, monotones remain computable. 
By contrast, entanglement theory exhibits a much higher computational barrier, with complexity scaling exponentially in the dimension of the Hilbert space~\cite{gurvitsClassicalDeterministicComplexity2003,ioannouComputationalComplexityQuantum2007,gharibianStrongNPHardnessQuantum2009}.

\subsection{Witnessing magic}\label{sec:witnessing} 

Faced with this intrinsic barrier, it is natural to consider a weaker---but far more accessible---approach: witnessing magic. Rather than fully quantifying the resource, one may simply ask whether a given state possesses magic at all. A \emph{witness} is a 
Hermitian operator that, when yielding a positive expectation value, certifies the presence of magic 
in a quantum state. More precisely, a witness $W$ is such that $\tr(W\rho)\le0$ for all $\sigma \in \hat{\mathcal{S}}$. Therefore, if a state $\rho$ gives $\tr(W\rho)>0$, it is guaranteed to be a magic state. The converse does not hold: a negative expectation value does not imply that $\rho$ is a stabilizer state. Although this limitation may seem unsatisfactory, two points put witnessing in proper context. First, witnesses correspond to directly measurable observables, providing a conceptually and experimentally transparent way to test for magic. Second, witnesses are \emph{complete}: for every non-stabilizer state \(\rho\), there exists at least one witness \(W_\rho\) that detects it. 

By definition, evaluating a witness on a given state can be done in polynomial time. The main drawback, however, is that a single witness cannot detect all magic states. If it could, it would be faithful on \(\hat{\mathcal{S}}\), which is impossible according to \cref{th1:membershiphard}. Consequently, effectively probing the boundary of the stabilizer polytope and detecting magic across different states requires multiple witnesses. This naturally leads to the question: how can we systematically construct witnesses in magic-state resource theory? In recent years, several approaches have been proposed, ranging from witnesses derived from stabilizer entropies~\cite{haugEfficientWitnessingTesting2025} to criteria~\cite{liuMagicCriterionAlmost2025} reminiscent of the celebrated PPT test in entanglement theory~\cite{peresSeparabilityCriterionDensity1996,horodeckiSeparabilityMixedStates1996}. Yet, as we show here, deciding whether a given observable is a valid witness is itself as hard as faithfully identifying the stabilizer polytope—once again requiring super-polynomial computational resources.



\begin{theorem}[Finding witnesses is hard. Informal of \cref{thm:nphardwwd,wwdinqp2}]  
Given a Hermitian operator $W$, deciding whether $\tr(W\sigma)\le 0$ for all $\sigma \in \hat{\mathcal{S}}$ (i.e., $W$ is a valid witness) or there exists $\tau \in \hat{\mathcal{S}}$ such that $\tr(W\tau) > 0$ (i.e., $W$ it is not a witness) belongs to the complexity class $\class{QP}^{2}$.  
\end{theorem}  

Thus, magic-state resource theory presents a {no-free-lunch scenario}: it is computationally hard both to faithfully identify the stabilizer polytope and to determine whether a given operator can serve as a witness. As a result, capturing magic in states not detected by the relatively few known witnesses remains inherently difficult.

\subsection{The classical-quantum boundary}\label{sec:classicalquantumboundary} 
To come to the point and picking up the core theme of the introduction, we now discuss the implications of our results for delineating the boundary between classical and quantum computation. The classical-quantum boundary is notoriously elusive. Even if a state contains a large amount of magic and cannot be simulated by the stabilizer formalism, it might still be efficiently simulable using other classical methods, such as tensor networks or covariance matrix techniques~\cite{Orus-AnnPhys-2014,RevModPhys.93.045003,CorbozPEPSFermions,PhysRevA.80.042333,AreaReview}. What is well-defined, however, is the classical-quantum boundary \textit{as seen by the stabilizer formalism}. This boundary separates states that can be efficiently simulated by stabilizer methods from those that cannot. 

Although \cref{th1:membershiphard} establishes a fundamental bottleneck in tracing the boundary between free and resourceful states, we can go one step further and consider the classical-quantum boundary from the perspective of magic: given a quantum state, can we determine whether it is classically simulable using the stabilizer formalism? 
This question goes beyond the simple dichotomy between free and resourceful states, as even states containing a small amount of magic can sometimes be classically simulated~\cite{aaronsonImprovedSimulationStabilizer2004,bravyiImprovedClassicalSimulation2016,bravyiSimulationQuantumCircuits2019}. For example, consider a state prepared by a Clifford circuit doped with a finite number $t$ of local non-Clifford gates~\cite{beverlandLowerBoundsNonClifford2020,leoneLearningTdopedStabilizer2024,guDopedStabilizerStates2024,guMagicinducedComputationalSeparation2024}. It is well known 
that such circuits can be simulated in time scaling exponentially in $t$~\cite{aaronsonImprovedSimulationStabilizer2004,bravyiImprovedClassicalSimulation2016,bravyiSimulationQuantumCircuits2019}, which implies polynomial-time simulation in the number of qubits when $t = O(\log n)$. 
Hence, the classical-quantum boundary, from the perspective of magic-state resource theory, is broader than the stabilizer polytope; it includes \emph{$t$-doped stabilizer states} and their convex combinations. Therefore, we ask: how hard is it to identify the boundary of the polytope formed by $t$-doped stabilizer states hereby denoted by $\hat{\mathcal{S}}_t$? Identifying this boundary is a crucial step toward understanding the complete picture: states inside the polytope of $t$-doped states are classically simulable, while states outside---and even very close to the boundary---are, in general, hard to simulate. This sharp transition in computational power makes the polytope of $t$-doped states the natural classical--quantum boundary within the stabilizer framework. 

\begin{theorem}[Tracing the classical-quantum boundary is hard. Informal of \cref{thm:suptdop}]\label{th3:hardnesstdoped}
    For any $\varepsilon=1/\class{poly}(d)$, the problem of deciding whether a state $\rho$ lies in $\hat{\mathcal{S}}_{t}$ or is $\varepsilon$-far from every state in $\hat{\mathcal{S}}_t$ satisfies:
    \begin{itemize}
        \item It belongs to the complexity class $\class{QP}^2$ for any $t < \log n$.
        \item For any $t = O(\log n)$ lies at least in $\class{QP}^2$ and at most in $\class{QP}=\bigcup_k\class{QP}^k$.
    \end{itemize}
\end{theorem}

Unsurprisingly, at this stage, the above theorem shows that this task turns out to be super-polynomially hard. Interestingly, however, as long as $t$ scales logarithmically with the number of qubits, the problem never becomes exponentially hard in the Hilbert space dimension, but instead remains quasi-polynomial. 


\subsection{Summary of conceptual results} 

Altogether, we reveal an ``unbearable hardness'' of magic: not only its quantification and reliable certification are intractable in principle, but even deciding membership in the set of classically simulable states is computationally hard in magic-state resource theory. Moreover, this intractability manifests in a way that is significant in practice. Indeed, from the very beginning, nobody expected these problems to be polynomial-time tractable in the number of qubits; exponential complexity already seemed like the most we could reasonably hope for. Yet even this hope would have been enough: exponential-time tools still allow meaningful large-scale simulations. For example, suppose an optimised classical solver can execute about $2^{40} \sim 10^{12}$ operations. A $2^n$-time algorithm can then handle instances of size $n \approx 25$, large enough to capture genuine asymptotic behaviour.
Super-exponential scaling completely changes this picture. Under the same computational budget, a $2^{n^2}$-time algorithm becomes infeasible already around $n \approx 5$, restricting us to toy models dominated by finite-size effects. In this sense, replacing $2^n$ by $2^{n^2}$ does not merely worsen constant factors; it fundamentally changes what we can hope to analyse in practice.

To conclude, we extend the discussion to highlight the practical consequences of our results. In particular, we show that for noisy quantum circuits, determining whether noise renders them classically tractable is even harder---super-exponential compared to their brute-force simulation. Similarly, in the context of magic-state distillation, we reveal a striking structural feature of magic-state resource theory: there exist pathological magic states whose distillation, although possible in principle, requires prohibitively super-exponential time, making it useless in practice.

\section{Implications of our results}
While our results are rigorous and framed in the language of theoretical computer science, revealing the intrinsic computational complexity of magic-state resource theory, in this section, we discuss some of their implications in more practical, application-orientated scenarios. In particular, despite its foundational goal of distinguishing classical from quantum physics, magic-state resource theory has two main practical purposes.  
First, states or quantum operations containing little magic can often be \emph{simulated efficiently by classical means}~\cite{PhysRevA.70.052328,bravyiImprovedClassicalSimulation2016,StabilizerPolytope,wangQuantifyingMagicQuantum2019}. This is highly valuable, as it provides quantum theorists with an additional powerful simulation tool: alongside tensor 
network methods, it enables the study of processes implemented 
on early quantum devices and supports 
ongoing experimental efforts. 
The second major practical application of magic-state resource theory lies 
in \emph{quantum error correction} and 
fault-tolerant quantum computing, specifically in the task of \emph{magic-state distillation}~\cite{howardApplicationResourceTheory2017}. This important subroutine, combined with transversal Clifford gates, allows for efficient fault-tolerant quantum computation~\cite{PhysRevA.71.022316, paetznickUniversalFaulttolerantQuantum2013}. 
While our results are primarily conceptual, they nevertheless have important implications for both of these practical directions.

\subsection{Classical simulation of noisy quantum circuits}

Quantum computers hold the promise of outperforming classical machines on certain tasks, but this promise remains largely theoretical until fully fault-tolerant hardware is available. In the current noisy intermediate-scale regime, a central goal has been to demonstrate a \textit{quantum advantage}---running a quantum circuit that performs a task beyond the reach of the best classical computers. A leading proposal for such demonstrations has been \emph{random circuit sampling}~\cite{morvanPhaseTransitionRandom2024}, in which a randomly chosen quantum circuit is executed on a device, and one samples from its output distribution. However, the inherent noise of current devices makes these processes partially classically simulable, and improvements in classical simulation techniques---including tensor network methods and massive GPU implementations---have repeatedly challenged the claimed advantage
in an interesting way, at times matching or even surpassing benchmarks~\cite{zhaoLeapfroggingSycamoreHarnessing2024,fuAchievingEnergeticSuperiority2024}. As a result, experimental demonstrations of quantum advantage remain in constant tension with the evolving capabilities of classical simulation.

In this context, a natural systematic question arises: which noisy quantum circuits 
are efficiently simulable by stabilizer-based classical methods? A noisy quantum circuit is, in general, a quantum channel \(\mathcal{E}\), that is, a completely positive trace-preserving linear map on quantum states. Based on our discussion of classical simulability via the stabilizer formalism in \cref{sec:classicalquantumboundary}, we say that a channel \(\mathcal{E}\) is classically simulable if it never generates states outside the tractable region defined by the \(t\)-doped stabilizer polytope $\hat{\mathcal{S}}_t$ for $t=O(\log n)$. Concretely, the noisy quantum circuit \(\mathcal{E}\) is simulable if \(\mathcal{E}(\sigma)\in\hat{\mathcal{S}}_t\) for every input stabilizer state \(\sigma\). Equivalently, \(\mathcal{E}\) injects at most \(t\) non-Clifford operations into any input state, and we call such a map a \emph{completely \(t\)-doped stabilizer preserving channel}.  

Building on the hardness of stabilizer and \(t\)-doped stabilizer membership problems established in \cref{th1:membershiphard,th3:hardnesstdoped}, we obtain the following result regarding the classical simulation of noisy quantum circuits.

\begin{corollary}[Deciding classical simulation of noisy circuits is hard. Informal of \cref{cor1app,cor1appdoped}]  
Given a quantum channel \(\mathcal{E}\), deciding whether \(\mathcal{E}\) is completely \(t\)-doped stabilizer preserving or not does not belong to the complexity class \(\class{QP}^{2-\eta}\) for any $\eta>0$. In particular, this problem is much harder than brute-force simulating the process $\mathcal{E}$ on a classical device.
\end{corollary}

The content of the above corollary implies, somewhat ironically, that determining whether a noisy quantum circuit is classically simulable in the stabilizer sense is significantly harder than the brute-force simulation of the circuit itself. This highlights the intrinsic computational hardness of deciding magic resources in quantum states and channels, even when considering weaker notions of simulability. 

 \subsection{Super-polynomial time magic-state distillation}
Magic-state distillation is a key component of universal fault-tolerant quantum computation. In most physical architectures, Clifford gates can be implemented reliably and, crucially, transversally. On their own, however, Clifford operations are not computationally universal. Universality can be achieved by supplementing them with suitable non-Clifford gates. But then the hen-and-egg problem arises: how to realise those non-Clifford gates in the first place. \emph{Magic-state distillation}~\cite{PhysRevA.71.022316,PhysRevX.2.041021,LitinskiMagic,howardApplicationResourceTheory2017,paetznickUniversalFaulttolerantQuantum2013,ComplexInstructionSet,PhysRevA.95.022316,PhysRevLett.120.050504,AdamWills} addresses this challenge by starting from multiple copies of noisy non-stabilizer states and processing them into fewer, higher-fidelity non-stabilizer states. These improved resource states are then injected into Clifford circuits via gate teleportation, enabling non-Clifford operations and, ultimately, universal quantum computation. Operationally, one typically has access to only a limited number of noisy magic states on a small number of qubits and aims to convert them into high-fidelity resources suitable for implementing non-Clifford gates. A central issue in this setting is the computational cost of carrying out such distillation procedures.

While \cref{th1:membershiphard} does not directly provide constructive guidance for designing magic-state distillation protocols, it does offer important insights. First, distillation schemes based on the canonical $\ket{T}\propto\ket{0}+e^{i\pi/4}\ket{1}$ resource state are not unique, and identifying alternative resource states that may lead to improved efficiency is a meaningful research direction. Indeed, resource states involving one-, two-, and three-qubit systems, among others, have been explored \cite{howardApplicationResourceTheory2017,paetznickUniversalFaulttolerantQuantum2013,ComplexInstructionSet,PhysRevA.95.022316,PhysRevLett.120.050504}. However, a first consequence of \cref{th1:membershiphard} is that the very search for new magic resource states is hindered by the intrinsic difficulty of deciding whether a given state is a magic state at all; no efficient optimisation procedure can accomplish this task.

A further implication of \cref{th1:membershiphard} is the existence—presumably in rather pathological cases, though this is hard to assess—of magic states for which distillation, while not ruled out in principle, cannot be achieved by any protocol running in less than super-polynomial time in the Hilbert space dimension, as the next corollary rigorously establishes.

\begin{corollary}[Super-polynomial time magic-state distillation. Informal of \cref{cor2app}]  
There exists a (possibly mixed) non-stabilizer state \(\rho\) on \(n\) qubits such that any stabilizer-based distillation protocol that can be found and executed within polynomial time in $d$ produces at most a super-polynomially small $\class{exp}(-\log^2d)$ fraction of high-fidelity magic states.  
\end{corollary}

\section{Discussion and open questions}
In this work, we have shown that the geometry of stabilizer states conceals an unexpectedly severe computational barrier: both the recognition of classicality and the quantification of genuinely quantum resources require 
$\class{exp}(n^2)$ time. This result fundamentally reframes the practical and conceptual status of the stabilizer polytope, of magic resource theory, and of the broader task of delineating classical from quantum computational power.
If one identifies the class of systems that can be efficiently simulated by Pauli propagation methods—and that therefore admit no quantum advantage—with a ``classical world,'' and its complement with a ``quantum world,'' then our results imply that the very act of separating these two regimes is itself computationally intractable. In this sense, the classical–quantum boundary is not merely subtle or poorly understood, but provably hard to determine. This adds a striking and somewhat ironic twist to an already challenging foundational problem.

Many questions remain open. Can these hardness results be circumvented for restricted families of states or channels? Are there physically natural approximations that avoid the 
$\class{exp}(n^2)$ barrier? Can the promise gap discussed in the present proofs be made larger? Similar exponentially small promise gaps are known from studies of the computational complexity of the separability problem \cite{ioannouComputationalComplexityQuantum2007}, but is there room for improvement?  
Do analogous phenomena arise in other resource theories? And, ultimately, what concrete lessons do these insights offer for the search for new quantum algorithms that achieve genuine advantages over classical computation? It is hoped that the present work invites the solution of some of these open problems.
  

\section{Acknowledgments}
We thank Scott Aaronson for pointing out a typographic error to us.
This work resulted from presentations and discussions at the Second International Workshop on Many-Body Quantum Magic, held by the InQubator for Quantum Simulation (IQuS) at the University of Washington, 18-29 August 2025.  The workshop has been  supported by the U.S. Department of Energy, Office of Science, Office of Nuclear Physics  Award Number DOE (NP) Award DE-SC0020970 via the program on Quantum Horizons: QIS Research and Innovation for Nuclear Science. L.L. acknowledges funding from the Italian Ministry of University and Research, PRIN PNRR 2022, project “Harnessing topological phases for quantum technologies”, code P202253RLY, CUP D53D23016250001, and PNRR-NQSTI project ”ECoN: End-to-end long-distance entanglement in quantum networks”, CUP J13C22000680006. This work has been supported by the BMFTR (DAQC, MuniQC-Atoms, QuSol, Hybrid++, PasQuops), Clusters of Excellence (ML4Q, MATH+), the Munich Quantum Valley, Berlin Quantum, the Quantum Flagship (Millenion, Pasquans2), the DFG (CRC 183, SPP 2541), the European Research Council (DebuQC), PraktiQOM, and the Alexander-von-Humboldt Foundation. 

\bibliographystyle{apsrev4-1}
\bibliography{bib,BigReferences68}
\newpage
\appendix

\setcounter{secnumdepth}{2}
\setcounter{equation}{0}
\setcounter{figure}{0}
\setcounter{table}{0}
\setcounter{section}{0}
\renewcommand{\thetable}{S\arabic{table}}
\renewcommand{\theequation}{S\arabic{section}.\arabic{equation}}
\renewcommand{\thefigure}{S\arabic{figure}}
\renewcommand{\thesection}{S.\arabic{section}} 
\begin{center}
\textbf{\large Supplemental Material}
\end{center}
\tableofcontents

\resumetoc

\section{Preliminaries}
This section is devoted to notation and preliminaries. In~\cref{sec:nota} we introduce basic notation and fundamental concepts from quantum information. In~\cref{sec:comp} we review standard notions from complexity theory. In~\cref{sec:epsilonnet} we recall preliminary results on $\varepsilon$-nets. In~\cref{sec:magic} we introduce the resource theory of magic states.

\subsection{Notation and basic notions}\label{sec:nota}
We fix the notation and recall basic notions from quantum information theory that will be used throughout the paper.

We denote the Hilbert space of $n$ qubits by $\mathcal{H}_n$, with dimension $d = 2^n$. The space of linear operators acting on $\mathcal{H}_n$ is denoted by $\mathcal{B}(\mathcal{H}_n)$, and the set of density operators on $\mathcal{H}_n$ by $\mathcal{D}(\mathcal{H}_n) \subset \mathcal{B}(\mathcal{H}_n)$. 
Many of the quantities considered in this work are naturally expressed in terms of operator norms. In particular, we will make extensive use of Schatten norms, which we now recall.

\begin{definition}[Schatten $p$-norms]
Let $A \in \mathcal{B}(\mathcal{H}_n)$. The Schatten $p$-norm of $A$ is defined as
\begin{align}
    \|A\|_p \coloneqq \left[\tr\!\left(|A|^p\right)\right]^{1/p},
\end{align}
where $|A| \coloneqq \sqrt{A^\dagger A}$ and $p \ge 1$. In particular, with $p=2$ we have the Frobenius norm $\|A\|_2\coloneqq\sqrt{\tr(A^{\dag}A)}$
\end{definition}

The Hilbert--Schmidt inner product on $\mathcal{B}(\mathcal{H}_n)$ is given by $\langle A,B\rangle \coloneqq \tr(A^\dagger B)$. 
Hölder's inequality, which generalises the Cauchy--Schwarz inequality, states that
\begin{align}
    |\tr(A^\dagger B)| \le \|A\|_p \, \|B\|_q ,
\end{align}
for $p,q \in [1,\infty]$ such that $p^{-1} + q^{-1} = 1$.
As a consequence, the Schatten $p$-norm can be shown to admit the  variational characterisation
\begin{align}
    \|A\|_p = \max_{\|B\|_q \le 1} \tr(A^\dagger B),
\end{align}
where again $p^{-1} + q^{-1} = 1$.

\subsection{Basics of complexity theory: $3$-SAT, quasi-polynomial classes and ETH}\label{sec:comp}

This section reviews the basic notion of complexity theory, with an emphasis on the $3$-\class{SAT} problem. For a comprehensive treatment of complexity theory, we refer the reader to standard references such as Ref.~\cite{aroraComputationalComplexityModern2009}.
Throughout this section, let $D$ denote the input size of the computational problems under consideration. For the remainder of the manuscript, we restrict attention to the decision problems defined below.

\begin{definition}[Decision problem] A decision problem is a computational problem defined by a set $L \subseteq \{0,1\}^D$ (a formal language) such that for every input string $x \in \{0,1\}^D$ the answer is either \class{YES} or \class{NO}. The task of a decision problem is to determine whether a given input $x\in L$. A Turing machine that always halts and outputs \class{YES} exactly when $x\in L$ and \class{NO} otherwise is said to solve the decision problem.
\end{definition}

We define two complexity classes of primary interest in the remainder of the manuscript.
\begin{definition}[\class{P} and \class{NP} complexity classes] The class \class{P} consists of all decision problems that can be decided in polynomial time by a deterministic Turing machine. The class \class{NP} consists of all decision problems for which \class{YES} instances can be verified in polynomial time by a deterministic Turing machine.
\end{definition}

The 3-\class{SAT} problem is a fundamental computational problem in complexity theory. It belongs to the class of problems known as \class{NP}-complete, which implies that it is both in the complexity class \class{NP} and is at least as hard as the hardest problems in \class{NP}.
In 3-SAT, we are given a Boolean formula consisting of a \emph{conjunction} (AND) of clauses. Each clause is a \emph{disjunction} (OR) of exactly three literals. A literal can be either a variable or its negation. For example, a clause may take the form \( (x_1 \vee \neg x_2 \vee x_3) \), where \(x_i\) represents a Boolean variable and \(\neg\) denotes negation.
The goal of the 3-\class{SAT} problem is to determine if there exists an assignment of truth values (true or false) to the variables such that the entire Boolean formula evaluates to true. This assignment is often referred to as a satisfying assignment.
Given a 3-\class{SAT} formula \( \phi \) with \(n\) variables and \(m\) clauses, the question is whether there exists an assignment of truth values \(x_1, x_2, \ldots, x_n\) such that \( \phi(x_1, x_2, \ldots, x_n) \) evaluates to true.
Subsequently, we provide a rigorous definition of the problem.
\begin{problem}[$3$-\class{SAT}]\label{prob:3sat} Problem:
    \begin{itemize}
        \item Input: A list of indices $\vec i\in [D]^{3,n_c}$ and a negation table $s\in \{\cdot, \lnot\}^{3,n_c}$.
        \item Output: Determine, if a Boolean variable $x\in\{0,1\}^{D}$ exists which satisfies all clauses.
        \begin{align}
            \forall c\in[n_c]:\quad s_{1,c} x_{i_1(c)}\lor s_{2,c} x_{i_2(c)}\lor s_{3,c} x_{i_3(c)}.
        \end{align}
    \end{itemize}
    \end{problem}
Here, $[D]^{3,n_c}$ means that we take $n_c$ subsets of $3$ elements belonging to $[D]$.
To establish the \class{NP}-hardness of a problem, a commonly employed approach involves demonstrating a polynomial-time reduction from $3$-\class{SAT} to the specific problem at hand. This process entails the transformation of $3$-\class{SAT} instances into instances of the target problem in a manner ensuring that a solution to the transformed instance directly implies a solution to the original $3$-\class{SAT} instance.
The key insight lies in recognising that efficient resolution of the target problem would also imply an efficient resolution of $3$-SAT, which is known to be \class{NP}-hard. To show the legitimacy of the reduction, one must prove that it operates in polynomial time and satisfies the conditions of ``completeness" and ``soundness".
``Completeness" asserts that if a solution exists for $3$-SAT, then a solution also exists for the transformed version of the problem of interest. 
``Soundness" states that if a solution does not exist for $3$-SAT, then there is no solution for the transformed version of the problem of interest.
Thus, the \class{NP}-hardness of 3-\class{SAT} serves as a reference point for proving the computational intractability of other problems, including the ones currently under consideration in our work.

We can give an analogous formulation of the $3$-\class{SAT} problem as a ground-state problem, which may be more familiar to physicists. Consider spin (Boolean) variables $x_{\alpha} \in \{0,1\}$, and a sign vector $s$ with components $s_{\alpha}\in\{0,1\}$, for $\alpha \in \{1,\dots,D\}$, and define the Hamiltonian
\begin{align}\label{hamiltonianformulation3sat}
H_s(x) = \sum_{\alpha < \beta < \gamma} (x_{\alpha}\oplus s_{\alpha}) (x_{\beta}\oplus s_{\beta}) (x_{\gamma}\oplus s_{\gamma}).
\end{align}
where $\oplus$ denotes the addition modulo $2$ such that $x_{\alpha}\oplus 1=\lnot x_{\alpha}$. The goal is to determine whether there exists an assignment $\bar{x}$ of the variables $x_{\alpha}$ such that $H(\bar{x}) = 0$, or whether $H(x) \ge 1$ for all possible assignments $x$.

Besides the canonical complexity classes of \class{P} and \class{NP}, we define the intermediate regime of quasi-polynomial complexity classes.

\begin{definition}[Quasipolynomial time and $\class{QP}^k$]\label{def:quasipolyclasses}
A decision problem is said to be solvable in \emph{quasipolynomial time} if there exists a deterministic Turing machine that decides it in time
\begin{equation}
2^{(\log D)^{O(1)}},
\end{equation}
where $D$ is the size of the input. The corresponding complexity class is
\begin{equation}
\class{QP} := \bigcup_{k \ge 1} \class{TIME}\big( 2^{(\log D)^k} \big),
\end{equation}
where $\class{TIME}(f(D))$ denotes the class of problems solvable by a deterministic Turing machine in $O(f(D))$ time. For a fixed number $\mathbb{R}\ni k \ge1$, define
\begin{equation}
\class{QP}^k := \class{TIME}\big( 2^{O((\log D)^k)} \big).
\end{equation}
In particular, $\class{QP}^1 = \class{P}$ up to polynomial factors in $D$, and $\class{QP}^2$ corresponds to problems solvable in time $2^{O((\log D)^2)}$.  
Thus, we have a \emph{hierarchy of quasipolynomial time classes}.
\begin{equation}
\class{QP}^1 \subseteq \class{QP}^2 \subseteq \class{QP}^3 \subseteq \cdots \subseteq \class{QP}
\end{equation}
\end{definition}

Let us now introduce one of the most widely accepted assumptions in complexity theory, namely the \emph{exponential time hypothesis} (ETH). We will assume ETH throughout this work in order to establish the computational hardness of problems arising in the resource theory of magic states.

\begin{definition}[Exponential time hypothesis (ETH)]\label{def:eth}
The ETH posits that there exists a constant $c > 0$ such that 3-\class{SAT} on $D$ variables cannot be solved in time $2^{c D}$.
\end{definition}
Under the ETH assumption, the quasi-polynomial complexity classes are strictly separated from one another. Indeed, consider a $3$-\class{SAT} instance on $\log^{k} D$ variables for $k\ge 2$. Since it can be solved in time $2^{\log^{k} D}$, it belongs to $\class{QP}^k$. However, under ETH, it cannot be solved in time $2^{c \log^{k} D}$ for any constant $c>0$, and therefore it cannot belong to $\class{QP}^{k-\eta}$ for any $\eta>0$. This assumption will be used to prove the main result of this work.

\subsection{Convex sets and related computational problems}\label{sec:convexsetspreliminaries}
In this section, we introduce the computational problems related to convex sets of $\mathbb{R}^D$ for $D\in\mathbb{N}$.

A subset \(K \subseteq \mathbb{R}^D\) is called convex if, for any two points \(\mathbf x,\mathbf y \in K\), the entire straight line segment \(\{\lambda \mathbf x + (1-\lambda)\mathbf y : \lambda \in [0,1]\}\) lies in \(K\); the convex hull of a set \(S \subseteq \mathbb{R}^D\) is the smallest convex set containing \(S\).

A central computational question is the {\em membership} problem: given a point \(\mathbf y \in \mathbb{R}^D\) and a description of \(K\), decide whether \(\mathbf y \in K\).  Because exact membership may be numerically delicate or computationally hard, one often considers the weak membership problem \(\mathrm{WMEM}_\varepsilon(K)\), where a tolerance parameter \(\varepsilon \ge 0\) is given, and the task is to decide whether \(\mathbf y\) lies within a distance \(\varepsilon\) of \(K\) or is at least \(\varepsilon\)-far from the interior of \(K\).  This relaxed formulation affords robustness to small perturbations or approximate descriptions and plays a fundamental role in convex-body algorithms and the complexity analysis of convex-set membership~\cite{grotschelGeometricAlgorithmsCombinatorial1988,ioannouComputationalComplexityQuantum2007}.

The weak membership framework is particularly useful in high-dimensional settings or when \(K\) is given implicitly (e.g., by a family of constraints or an oracle): it permits stable, approximate decisions even under numerical or representational imprecision. In many algorithmic contexts (volume estimation, optimisation, feasibility testing), a weak-membership oracle suffices to implement efficient procedures. To precisely formulate the weak membership problem \( \class{WMEM}_{\varepsilon}(K) \), we must first establish several key concepts that underpin its definition.
First, let us lay down some notational conventions used throughout the work. Given a vector $\mathbf{x}\in\mathbb{R}^D$, we define 
\begin{equation}\label{l2norm}
\|\mathbf{x}\|_2\coloneqq\sqrt{\sum_{i=1}^{D}x_i^2} 
\end{equation}
the usual Euclidean  $\ell_2$-norm. Given two vectors $\mathbf{x},\mathbf{y}$, we denote their scalar product as $\mathbf{y}^T\mathbf{x}\in\mathbb{R}$. 

\begin{definition}[Inner and outer core of a convex set]\label{def:set1} Let \( K \subseteq \mathbb{R}^D \) be a subset of $\mathbb{R}^D$. We define the following.  
\begin{itemize}
    \item We define  
    $
        S(K, \varepsilon) \coloneqq \{\mathbf{x} \in \mathbb{R}^D \mid \exists\, \mathbf{y} \in K : \|\mathbf{x} - \mathbf{y}\|_2 \le \varepsilon\}.
    $
    \item We define  
    $
        S(K, -\varepsilon) \coloneqq \{\mathbf{x} \in \mathbb{R}^D \mid S(\{\mathbf{x}\}, \varepsilon) \subseteq K\},
    $
    where \(\{\mathbf{x}\}\) denotes the 
    singleton set containing \(\mathbf{x}\).
\end{itemize}

\end{definition}
In words: given a subset \(K \subset \mathbb{R}^D\), \(S(K,\varepsilon)\) is the \(\varepsilon\)-\emph{outer neighborhood} of \(K\), i.e.\ the set obtained by expanding \(K\) outward by \(\varepsilon\); while \(S(K,-\varepsilon)\) is the \emph{inner core} of \(K\), i.e.\ the set of points whose \(\varepsilon\)-ball lies entirely within \(K\).

For our purposes, it will also be useful to formalise the following concept.  

\begin{definition}[Well-bounded centered set]\label{def:set2}
Let $K$ be a subset of $\mathbb{R}^D$. Then, $K$ is called \emph{well-bounded and \(p\)-centered} if there exist rationals \(R > 0\), \(r > 0\), and a point \(\mathbf{p} \in K\) such that, with the origin \(\mathbf{0} \in \mathbb{R}^D\),
\begin{align}
    K &\subseteq S(\{\mathbf{0}\}, R), \\
    K &\supseteq S(\{\mathbf{p}\}, r) .
\end{align}
\end{definition}
Intuitively, the condition $K \subseteq S(\{\mathbf{0}\},R)$ ensures that $K$ is contained within some large ball of radius $R$ around the origin, i.e., $K$ has a bounded outer size. Meanwhile, the condition $K \supseteq S(\{\mathbf{p}\},r)$ asserts that $K$ contains a full ball of radius $r$ around some point $\mathbf{p} \in K$, i.e., $K$ has a non-trivial \emph{inner core}. Thus, the definition rules out both unboundedness and degeneracy: it ensures that $K$ is a reasonably sized set with some guaranteed volume (or thickness) from the inside, as well as a global bound from the outside.

Given \cref{def:set1,def:set2}, for the subset $K$, we can define weak membership as a rigorous decision problem.

\begin{problem}[Weak membership (\class{WMEM})]\label{def:weakmem}
	Let $K\subseteq \mathbb{R}^{D}$ be a compact, convex, and well-bounded $p$-centered set. Then, given $\mathbf{y}\in\mathbb{Q}^{D}$ and an error parameter $\varepsilon\in\mathbb{Q}$, s.t.,  $\varepsilon>0$, decide: 
\begin{itemize}
    \item If $\mathbf{y}\in S(K,-\varepsilon)$, then output \class{YES}.
    \item If $\mathbf{y}\notin S(K,\varepsilon)$, then output \class{NO}.

\end{itemize}
An instance of \class{WMEM} is denoted as $\class{WMEM}_{\varepsilon}(K)$.
	\end{problem}
Similarly, it is possible to define a (polynomially) equivalent notion of weak membership called $1$-side weak membership problem.

\begin{problem}[$1$-sided weak membership problem ($\class{WMEM}^{1}$)]\label{def:weak1mem}
	Let $K\subseteq \mathbb{R}^{D}$ be a compact, convex, and well-bounded $p$-centered set. Then, given $\mathbf{y}\in\mathbb{Q}^{D}$ and the error parameter $\varepsilon\in\mathbb{Q}$, s.t.\  $\varepsilon>0$, decide: 
\begin{itemize}
    \item If $\mathbf{y}\in K$, then output \class{YES}.
    \item If $\mathbf{y}\notin S(K,\varepsilon)$, then output \class{NO}.

\end{itemize}
An instance of $\class{WMEM}^{1}$ is denoted as $\class{WMEM}_{\varepsilon}^{1}(K)$.
	\end{problem}

In what follows, we restate a theorem—previously established in Refs.~\cite{liuComplexityConsistencyNrepresentability2007,grotschelGeometricAlgorithmsCombinatorial1988}—that demonstrates the polynomial equivalence of these problems. For the purposes of this section, we focus on one direction of the equivalence: specifically, we show that oracle access to $\class{WMEM}^1$ suffices to solve $\class{WMEM}$.

\begin{theorem}[Lemma 2.5 and Lemma 2.6 of Refs.~\cite{liuComplexityConsistencyNrepresentability2007,grotschelGeometricAlgorithmsCombinatorial1988}]\label{thm:wmem1wmem}Let $K\subseteq \mathbb{R}^{D}$ be a compact, convex and well-bounded $p$-centered set with an associated radii $(R,r)$. Given an instance of $\class{WMEM}_{\delta}(K)$, there exists an algorithm that runs in time $\class{poly}(D,R,\lceil 1/\delta \rceil)$ and solves $\class{WMEM}_{\delta}(K)$ using an oracle for $\class{WMEM}_{\varepsilon}^{1}(K)$ with 
\begin{equation}
\varepsilon \coloneqq \frac{\delta^3r^2}{3\cdot 2^{12} D^5 R^4}.
\end{equation}
\end{theorem}
    
Let us now introduce a notion closely related to that of membership. For a convex set $K$, deciding whether a point belongs to $K$ can be framed in terms of the existence of a linear witness that certifies non-membership. In this viewpoint, one seeks a hyperplane that supports $K$ and separates it from a given query point. This motivates the following definition, which formalises the notion of such a witness in terms of linear functionals bounded over $K$.

\begin{definition}[Witness of a set $K$] Let $K$ be a compact, convex and well-bounded centered set. Let $\mathbf{w}\in\mathbb{R}^{D}$ and $\gamma\in\mathbb{Q}$. Then $\mathbf{w}$ is a weak-witness for $K$ if and only if
\begin{align}
    \mathbf{w}^{T}\mathbf{x}\le \gamma,\quad\forall \mathbf{x}\in K. 
\end{align}
Equivalently, $\max_{\mathbf{x}\in K} \mathbf{w}^{T}\mathbf{x}\le \gamma$.
\end{definition}

For a given convex set \( K \), an alternative to the membership problem is to consider its dual: rather than determining whether a point lies in \( K \), we seek a separating hyperplane that certifies membership by distinguishing \( K \) from its exterior, i.e., a witness for $K$. The precise formulation follows below, building on 
the \emph{weak optimisation} problem (\class{WOPT}).
 
\begin{problem}[Weak optimisation (\class{WOPT})]\label{def:wopt}
    Let $K\subseteq \mathbb{R}^{D}$ be a compact, convex and well-bounded $p$-centered set. Then given $\mathbf{w}\in\mathbb{Q}^{D}$ with $\left\|\mathbf{w}\right\|_{2}\le1$ and $\gamma,\delta\in\mathbb{Q}$ such that $\delta>0$ decide: 
\begin{itemize}
  \item  If there exists $\mathbf{y}\in S(K,-\delta)$ with $\mathbf{w}^{T}\mathbf{y}\ge\gamma+\delta$, then output \class{YES}.
  \item If for all $\mathbf{y}\in S(K,\delta)$, $\mathbf{w}^{T}\mathbf{y}\le \gamma -\delta $, then output \class{NO}.
\end{itemize}
An instance of \class{WOPT} is denoted as $\class{WOPT}_{\delta}(K)$.
\end{problem}
Let us explain in words what \cref{def:wopt} is about: given a vector $\mathbf{w}$, \class{WOPT} is the decision problem that determines whether  $\mathbf{w}$ is a witness for the set $K$ or not. 
As shown in Refs.~\cite{liuComplexityConsistencyNrepresentability2007,grotschelGeometricAlgorithmsCombinatorial1988} and restated below, the weak optimisation problem can be solved 
efficiently
given oracle access to $\class{WMEM}$. This implies the existence of a polynomial-time reduction in $D$ from $\class{WMEM}$ to $\class{WOPT}$.
\begin{theorem}[Proposition 2.8 of Refs.~\cite{liuComplexityConsistencyNrepresentability2007,grotschelGeometricAlgorithmsCombinatorial1988}]\label{thm:wmemimpwopt}
	Let $K\subseteq \mathbb{R}^{D}$ be a compact, convex and well-bounded $p$-centered set with an associated radii $(R,r)$. Given an instance of $\class{WOPT}_{\delta}(K)$ with parameters $(\mathbf w, \gamma)$, there exists an algorithm that runs in time $\class{poly}(D,R,\lceil 1/\delta \rceil)$ and solves $\class{WOPT}_{\delta}(K)$ using an oracle for $\class{WMEM}_{\varepsilon}(K)$ with $\varepsilon = r^{3}\delta^{3}/[2^{13}3^{3}D^{5}R^{4}(R+r)]$.
\end{theorem}
Let us provide some intuition behind the theorem. The proof that having oracle access to \(\class{WMEM}_{\varepsilon}(K)\) solves \(\class{WOPT}_{\delta}(K)\) proceeds in two conceptual steps. First, one shows that the membership oracle can be used to simulate a weak separation oracle: by querying membership on carefully chosen test points, one either confirms that a point is “inside” (approximately), or else obtains evidence that the point lies outside \(K\), from which one can infer a separating hyperplane. Secondly, starting from a known outer ball containing \(K\), one alternates between querying the separation oracle and using the returned hyperplanes to cut away regions of space that cannot contain an optimal solution, thereby gradually shrinking the feasible region. After a polynomial number of iterations (in the dimension and precision parameters), this procedure yields a point that approximately maximises the objective over \(K\), thus solving $\class{WOPT}_{\delta}(K)$.

\subsection{Preliminaries on $\varepsilon$-nets}\label{sec:epsilonnet}
This section reviews key notions related to $\varepsilon$-nets, which are part of our subsequent complexity-theoretic arguments. The material presented here follows the treatment in Ref.~\cite{vershyninHighDimensionalProbabilityIntroduction2018}; readers seeking further details are referred to that work.
We begin with the concept of a covering net.

\begin{definition}[$\varepsilon$-covering net~\cite{vershyninHighDimensionalProbabilityIntroduction2018}]
Let $(T,\|\cdot\|)$ be a normed space with the metric induced by the norm, let $K\subseteq T$, and let $\varepsilon>0$. A subset $C\subseteq K$ is called an \emph{$\varepsilon$-covering net} of $K$ if for every $x\in K$ there exists $x_0\in C$ such that $\|x-x_0\|\le \varepsilon$. The \emph{$\varepsilon$-covering number} of $K$, denoted by $\mathscr{C}(K,\|\cdot\|,\varepsilon)$, is defined as the minimal cardinality among all $\varepsilon$-covering nets of $K$. An $\varepsilon$-covering net $C$ is said to be \emph{optimal} if $|C|=\mathscr{C}(K,\|\cdot\|,\varepsilon)$.
\end{definition}

Intuitively, an $\varepsilon$-covering net provides a finite approximation of the set $K$, in the sense that every point in $K$ lies within distance $\varepsilon$ of some net point. A dual notion is that of packing, which captures the maximum number of mutually distinguishable points that can be placed in $K$.

\begin{definition}[($\varepsilon$-packing net)~\cite{vershyninHighDimensionalProbabilityIntroduction2018}]\label{def_pack}
Let $(T,\|\cdot\|)$ be a normed space with the metric induced by the norm, let $K\subseteq T$, and let $\varepsilon>0$. A subset $P\subseteq K$ is called an \emph{$\varepsilon$-packing net} of $K$ if $\|x-y\|>\varepsilon$ for every pair of distinct points $x,y\in P$. The \emph{$\varepsilon$-packing number} of $K$, denoted by $\mathscr{P}(K,\|\cdot\|,\varepsilon)$, is defined as the maximal cardinality among all $\varepsilon$-packing nets of $K$. An $\varepsilon$-packing net $P$ is said to be \emph{optimal} if $|P|=\mathscr{P}(K,\|\cdot\|,\varepsilon)$.
\end{definition}

While a covering net ensures that all points are \emph{close} to a net point, a packing net ensures that all net points are \emph{far} from each other. This duality is formalised by the following fundamental lemma.

\begin{lemma}[Relation between covering and packing~\cite{vershyninHighDimensionalProbabilityIntroduction2018}]\label{lemma:cover_pack_relation}
Let $(T,\|\cdot\|)$ be a normed space, $K\subseteq T$, and let $\varepsilon>0$. Then
\begin{align}
\mathscr{P}(K,\|\cdot\|,\varepsilon)
\;\ge\;
\mathscr{C}(K,\|\cdot\|,\varepsilon)
\;\ge\;
\mathscr{P}(K,\|\cdot\|,2\varepsilon).
\label{eq:cover_pack_inequality}
\end{align}
\end{lemma}

In what follows, we will primarily bound covering numbers via packing numbers using~\cref{lemma:cover_pack_relation}, since packing bounds are often easier to obtain through volumetric arguments. Building upon these foundational definitions, we now review a series of standard bounds on covering and packing numbers, progressing from classical geometric objects to the quantum state spaces relevant to our analysis. For conciseness, we state these results without proof.

Let us specialise to the Euclidean space $\mathbb{R}^D$ equipped with the $\ell_2$-norm (see \cref{l2norm}). We define the ball of radius $r$ and its boundary, the sphere, as
\begin{align}
B_r(D) &\coloneqq \{x\in\mathbb{R}^D : \|x\|_2 \le r\}, \\
\partial B_r(D) &\coloneqq \{x\in\mathbb{R}^D : \|x\|_2 = r\}.
\end{align}

The following proposition provides volumetric bounds on covering and packing numbers.

\begin{proposition}[Proposition 4.2.10 in Ref.~\cite{vershyninHighDimensionalProbabilityIntroduction2018}]\label{prop:packingupplow}
Let $K\subset \mathbb{R}^D$ and let $\varepsilon>0$. Then
\begin{align}
\frac{\mathrm{Vol}(K)}{\mathrm{Vol}(B_{\varepsilon}(D))}
\;\le\;
\mathscr{C}(K,\|\cdot\|_2,\varepsilon)
\;\le\;
\mathscr{P}(K,\|\cdot\|_2,\varepsilon)
\;\le\;
\frac{\mathrm{Vol}(K + B_{\varepsilon/2}(D))}{\mathrm{Vol}(B_{\varepsilon/2}(D))}.
\end{align}
\end{proposition}

Since $\partial B_1(D) \subseteq B_1(D)$, any $\varepsilon$-packing of the sphere is also an $\varepsilon$-packing of the ball. Consequently, the packing number of the ball upper bounds that of its boundary as
\begin{align}\label{eq:vol}
\mathscr{P}(\partial B_1(D),\|\cdot\|_2,\varepsilon)
\;\le\;
\mathscr{P}(B_1(D),\|\cdot\|_2,\varepsilon)
\;\le\;
\frac{\mathrm{Vol}(B_1(D)) + \mathrm{Vol}(B_{\varepsilon/2}(D))}{\mathrm{Vol}(B_{\varepsilon/2}(D))}.
\end{align}
We now turn to packing numbers for the set of pure quantum states, which will allow us to discretise continuous families of quantum states.

\begin{definition}[Pure-state manifold~\cite{meleLearningQuantumStates2025}]\label{def:pure_states}
Let $D\in\mathbb{N}$ and $\mathcal{H}=\mathbb{C}^D$ equipped with the norm $\|\cdot\|_2$. Define the set of pure states as
\begin{align}
\mathcal{V}_{D}
\coloneqq
\{\ket{\psi}\in\mathcal{H} : \langle\psi|\psi\rangle = 1\}.
\end{align}
\end{definition}

Note that $\mathcal{V}_{D}$ can be identified with the unit sphere in $\mathbb{R}^{2D}$ via the real and imaginary parts of the complex amplitudes. Explicitly, the map $\ket{\psi}\mapsto (\Re\psi,\Im\psi)$ is an isometry between
$(\mathbb{C}^D,\|\cdot\|_2)$ and $(\mathbb{R}^{2D},\|\cdot\|_2)$. This observation allows us to transfer volumetric bounds from the Euclidean setting to the quantum state space.

\begin{lemma}[Packing bound for pure states]\label{lem:packing}
Let $D\in\mathbb{N}$, $\mathcal{H}=\mathbb{C}^D$, and $\varepsilon\in[0,1)$. Then
\begin{align}
\mathscr{P}(\mathcal{V}_{D},\|\cdot\|_2,\varepsilon)
\;\le\;
\left(\frac{3}{\varepsilon}\right)^{2D}.
\end{align}
\end{lemma}

\begin{proof}
The set $\mathcal{V}_{D}$ can be identified with the boundary $\partial B_1(2D)$ of the unit ball in $\mathbb{R}^{2D}$. Applying the volumetric upper bound in Proposition~\ref{prop:packingupplow} to $B_1(2D)$ and using~\eqref{eq:vol} yields
\begin{align}
\mathscr{P}(\partial B_1(2D),\|\cdot\|_2,\varepsilon)
\;\le\;
\frac{\mathrm{Vol}(B_1(2D)) + \mathrm{Vol}(B_{\varepsilon/2}(2D))}{\mathrm{Vol}(B_{\varepsilon/2}(2D))}
\;\le\;
\left(\frac{3}{\varepsilon}\right)^{2D},
\end{align}
where the last inequality follows from the standard scaling $\mathrm{Vol}(B_r(2D)) = r^{2D}\mathrm{Vol}(B_1(2D))$. Identifying $\mathcal{V}_{D}$ with $\partial B_1(2D)$ completes the proof.
\end{proof}

\subsection{Magic-state resource theory}\label{sec:magic}
In this section, we review the main notions of magic-state resource theory, which form the core of our work.

Magic-state resource theory provides a rigorous and operational framework for quantifying the non-stabilizer features of quantum states and channels that are necessary for universal quantum computation. In this formalism, one identifies a distinguished set of operations and states that are considered \emph{free} because they can be implemented efficiently and simulated classically. Specifically, the stabilizer formalism—comprising Clifford unitaries, preparations of stabilizer states, Pauli measurements, classical feed-forward, and discarding of subsystems—defines a subclass of quantum computations that admit efficient classical simulation via the Gottesman–Knill theorem~\cite{gottesmanHeisenbergRepresentationQuantum1998}.

Within this resource theory, the \emph{free states} are those confined to the stabilizer polytope $\hat{\mathcal{S}}_n$, i.e.\ the convex hull of all stabilizer states, which can be prepared using only free operations and therefore do not contribute any computational advantage beyond classical simulation.  Free operations are taken to be stabilizer-preserving maps that, by definition, cannot generate non-stabilizer features from free states. A state lying outside this polytope is said to possess \emph{magic}~\cite{bravyiUniversalQuantumComputation2005}—a resource quantified by magic monotones that capture the extent to which the state enables computational tasks beyond the stabilizer regime.

In what follows, we define $\mathcal{S}_n$ the set of pure stabilizer states and $\hat{\mathcal{S}}_n$ its convex hull, i.e., the stabilizer polytope. We denote by $\mathbb{P}_n$ the Pauli group on $n$ qubits and $P\in\mathbb{P}_n$ Pauli strings.  
Let us formally define magic monotones below. 

\begin{definition}[Magic monotones]\label{def:stabilizermonotone} A stabilizer monotone $\mathcal{M}$ is a real-valued function for all $n$ qubit systems (or collections thereof) such that (i) $\mathcal{M}(\rho)=0$ if and only if $\rho\in\hat{\mathcal{S}}_n$; (ii) $\mathcal{M}$ is non-increasing under free operations $\mathcal{M}(\mathcal{E}(\rho))\le\mathcal{M}(\rho)$ for any free operation $\mathcal{E}$. A function $\mathcal{M}$ is pure-state stabilizer monotone if condition (i) is satisfied for pure states and (ii) holds for any pair $(\ket{\psi},\mathcal{E})$, obeying $\mathcal{E}(\ketbra{\psi})=\ketbra{\phi}$. 
\end{definition}

The arguably most important  pure-state magic monotone is the stabilizer entropy, defined below.

\begin{definition}[Stabilizer R\'enyi entropy~\cite{leoneStabilizerRenyiEntropy2022}] Given a pure state $\psi$, the $\alpha$-stabilizer entropy is defined as
\begin{align}\label{eq:stabentropydef}
    M_{\alpha}(\ket{\psi})=\frac{1}{1-\alpha}\log\frac{1}{d}\sum_{P}\tr^{2\alpha}(P\psi)\,.
\end{align}
$M_{\alpha}$ obeys the following properties.
\begin{enumerate}[label=(\roman*)]
    \item It is a pure-state magic monotone~\cite{leoneStabilizerEntropiesAre2024} for any $\alpha\ge 2$.
    \item Its computation lies in $\class{TIME}(2^{O(n)})$ and thus in $\class{P}$.
    \item It provides a family of efficiently computable witnesses~\cite{haugEfficientWitnessingTesting2025}: let $\rho$ be a mixed state and $M_{\alpha}(\rho)$ computed through \cref{eq:stabentropydef}, then
    \begin{align}
        M_{\alpha}(\rho)-\frac{1-2\alpha}{1-\alpha}S_{2}(\rho)\,,
    \end{align}
    where $S_{2}(\rho)\coloneqq-\log\tr\rho^2$, is a magic witness for any $\alpha\ge \frac{1}{2}$.
    \item It can be extended to mixed states through the ``convex roof'' or convex hull construction~\cite{leoneStabilizerEntropiesAre2024}
    \begin{align}
        \tilde{M}_{\alpha}(\rho)=\frac{1}{1-\alpha}\log \sup_{(q_j,\phi_j)}\sum_{j}q_{j}\left(\frac{1}{d}\sum_{P\in\mathbb{P}_n}\tr^{2\alpha}(P\psi)\right)
    \end{align}
which can be computed in time $\class{TIME}(2^{O(nr^2)})$ with $r=\operatorname{rank}(\rho)$. 
\end{enumerate}
\end{definition}

The key fact that makes the stabilizer entropy relevant in the context of this work is that it can be computed in polynomial-time in the Hilbert space dimension. Indeed, given a quantum state $\rho=\ketbra{\psi}{\psi}$ corresponding to a pure state, decision problems associated with the computation of $M_{\alpha}(\rho)$ lie in $\class{QP}^1=\class{P}$ in the input parameter $d=2^n$. In sharp contrast, the main result of this paper demonstrates that the computation of any magic resource monotone working on mixed states is fundamentally super-polynomial; in particular (according to \cref{cor1main}), it lies in $\class{QP}^{k}$ for $k\ge 2$ (see \cref{def:quasipolyclasses}). Below, we introduce a magic monotone whose computation requires $\class{TIME}(2^{O(\log^2d)})$ thus lying in $\class{QP}^2$, thus matching the optimal computational time for magic monotones. 

\begin{definition}[Robustness of magic~\cite{howardApplicationResourceTheory2017}]\label{def:robustnessofmagic} Given a $n$-qubit density matrix $\rho$. The robustness of magic is defined 
\begin{align}
    \mathcal{R}(\rho)\coloneqq\min_{x}\left\{\sum_i|x_i|\,:\, \rho=\sum_ix_i\sigma_i,\,\,\sigma_i\in\mathcal{S}_{n}\right\}\,.
\end{align}
$\mathcal{R}$ is a magic monotone.
\end{definition}

Similarly to magic monotones, one can define witnesses within the magic-state resource theory.
\begin{definition}[Magic witness]
Let $\gamma\in\mathbb{Q}$. A Hermitian operator $W$ is called a \emph{magic witness} if $\tr(W\sigma) \le \gamma$ for every stabilizer state $\sigma \in \hat{\mathcal{S}}_n$.
\end{definition}

\medskip

As anticipated in the main text, a clarifying remark is in order. For general mixed states, the robustness of magic—despite its optimality as a resource monotone—incurs a computational cost scaling as $2^{\Theta(n^2)}$, rendering it intractable for systems beyond $n\simeq 5$, as also noted in Ref.~\cite{howardApplicationResourceTheory2017}. Remarkably, however, the extended stabiliser entropy $\tilde{M}_{\alpha}$ admits a polynomial-time computation when the rank satisfies $r = O(1)$. This feature is crucial, as it enables the systematic study of magic-state resource theory beyond pure states, extending its applicability to low-rank density matrices. Low rank states are ubiquitous in many-body physics. A prominent example is provided by \emph{matrix product states} (MPS) 
with constant bond dimension, which naturally arise as ground states of one-dimensional gapped many-body Hamiltonians. For these systems, the magic content of reduced density matrices can be faithfully quantified using the extended stabilizer entropy, with a computational cost that is polynomial in the Hilbert space dimension. As emphasised in the main text, this represents a substantial practical advantage over the super-exponential scaling $2^{\Theta(n^2)}$ associated with robustness-based measures. An intriguing open problem is whether the computational cost of the extended stabilizer entropy for reduced density matrices of MPS states can be further reduced from $2^{\Theta(n)}$ to $\class{poly}(n)$, as is the case for the stabilizer entropy of the full MPS state~\cite{olivieroMagicstateResourceTheory2022,Haug_2023}.

Another magic monotone that plays a crucial role in magic-state resource theory due to its intuitive operational meaning is the {\em stabilizer fidelity}~\cite{bravyiSimulationQuantumCircuits2019}, i.e., the fidelity of a given pure state with the closest state in the stabilizer polytope.

\begin{definition}[Stabilizer fidelity~\cite{bravyiSimulationQuantumCircuits2019}] Let $\rho$ be a (possibly mixed) state. Its stabilizer fidelity $F_{\text{stab}}$ is defined as
\begin{align}
    F_{\text{stab}}(\rho)=\max_{\sigma\in\hat{\mathcal{S}}_n}\|\sqrt{\sigma}\rho\sqrt{\sigma}\|_1.
\end{align}
\end{definition}
For pure states, the stabilizer fidelity is closely related to the stabilizer entropy. The following lemma will be useful in the remainder of the manuscript.
\begin{lemma}[Relation between stabilizer entropy and stabilizer fidelity~\cite{arunachalamPolynomialTimeTolerantTesting2025,baoTolerantTestingStabilizer2025,bittelOperationalInterpretationStabilizer2025}]\label{stabentstabfid} Let $\psi$ be a pure quantum state. Let $P_6(\psi)\coloneqq2^{-2M_{3}(\psi)}$ with $M_3$ the $(\alpha=3)$ stabilizer entropy. Then it holds that
\begin{equation}\label{eq:stabentropystabfidelityrelation}
(P_{6}(\psi))^{C} \le F_{\text{stab}}(\psi) \le (P_{6}(\psi))^{1/6}
\end{equation}
for some constant $C>0$.

\end{lemma}

Resource monotones are useful in the central tasks of quantum resource theory, which is resource distillation~\cite{bravyiUniversalQuantumComputation2005,howardApplicationResourceTheory2017,beverlandLowerBoundsNonClifford2020}.
Indeed, 
magic-state distillation~\cite{bravyiUniversalQuantumComputation2005} is a cornerstone of fault-tolerant quantum computing within the magic resource theory framework. While Clifford gates and stabilizer states---the free resources in this theory---can be made fault-tolerant, they are insufficient for universal quantum computation. Universality requires non-Clifford gates, such as the $T$-gate or the Toffoli gate, which are enabled by high-fidelity resource states like the single-qubit $\ket{T}$ or the multi-qubit $\ket{CCZ}$ state vector. Since these magic states cannot be prepared fault-tolerantly directly, distillation protocols are essential: they consume multiple noisy copies to produce fewer copies of a purified magic state.

A primary historical bottleneck has been the resource overhead, defined as the number of noisy input states required per clean output. Early protocols~\cite{bravyiUniversalQuantumComputation2005} exhibited overhead scaling as $O(\log^{\gamma}(1/\varepsilon))$ for a target error $\varepsilon$, with $\gamma > 1$. Recent theoretical breakthroughs have dramatically improved this scaling. Notably, protocols achieving constant overhead~\cite{willsConstantoverheadMagicState2025} ($\gamma = 0$) have been developed by constructing asymptotically good quantum codes with transversal non-Clifford gates and efficient decoders, making them suitable for distilling magic state vectors like $\ket{T}$ and $\ket{CCZ}$. This progress is being matched by experimental advances, including the first logical-level demonstrations~\cite{salesrodriguezExperimentalDemonstrationLogical2025}. Beyond engineering efficient protocols, the formal framework of magic resource theory~\cite{veitchResourceTheoryStabilizer2014}---extended to quantum channels~\cite{seddonQuantifyingMagicMultiqubit2019}---provides the necessary tools to quantify magic and establish fundamental lower bounds on the resource overhead of any distillation protocol, delineating what is possible from what is provably hard.

Besides the task of resource distillation through stabilizer protocols, the magic-state resource theory is also important because it aims to quantify the resources a quantum algorithm requires to achieve some form of advantage over classical computation. In particular, stabilizer states do not yield any computational advantage, as their outcome distributions—when measured in the computational basis—can be efficiently simulated classically in polynomial time~\cite{aaronsonImprovedSimulationStabilizer2004}. From this fact, within the resource theory of magic, one identifies additional classes of operations that strictly extend the stabilizer operations, known as \textit{completely stabilizer operations}. These are quantum channels that, when applied to stabilizer states, return only stabilizer states. We define them precisely below.

\begin{definition}[Completely stabilizer-preserving channel~\cite{seddonQuantifyingMagicMultiqubit2019}]
Let $\mathcal{E}$ be a quantum channel acting on an $n$-qubit input state. We say that $\mathcal{E}$ is completely stabilizer-preserving iff $[\mathcal{E}\otimes I](\sigma)\in\hat{\mathcal{S}}_{n+m}$ for all $\sigma\in\hat{\mathcal{S}}_{n+m}\,\,\forall m$.
\end{definition}

These channels play a pivotal role because, as they cannot create non-stabilizer states from scratch, their action can be simulated classically in full, and no magic state can be distilled from their output. Moreover, as it has been shown in Ref.~\cite{heimendahlAxiomaticOperationalApproaches2022}, the set of completely stabilizer operations is strictly larger than the one obtained by arbitrary combinations of the basic stabilizer protocols introduced above. 

Ref.~\cite{seddonQuantifyingMagicMultiqubit2019} provides a characterisation of completely stabilizer-preserving channels in terms of the Choi state $\rho_{\mathcal{E}}$ uniquely associated with the channel, which we recall below for completeness.

\begin{lemma}[Characterisation of completely stabilizer preserving channels~\cite{seddonQuantifyingMagicMultiqubit2019}]\label{lem:characterisationstabilizerpreservingchannel} Given a channel $\mathcal{E}$, it is stabilizer preserving if and only if the Choi state $\rho_{\mathcal{E}}\in\hat{\mathcal{S}}_{2n}$. 
\end{lemma}
\subsection{Summary of results on $t$-doped stabilizer states}\label{sec:dopedstates}

Within the stabilizer formalism, stabilizer states are not the only states that can be efficiently simulated classically. Indeed, small perturbations of stabilizer states can also remain classically simulable. In particular, adding $t$ local non-Clifford gates produces states that can be simulated in time $\class{poly}(n,2^{O(t)})$~\cite{aaronsonImprovedSimulationStabilizer2004,bravyiImprovedClassicalSimulation2016}. Below, we define a class of states that is strictly larger than this: the $t$-doped states.

\begin{definition}[$t$-doped states~\cite{guMagicInducedComputationalSeparation2025,guDopedStabilizerStates2024}]\label{def:tcomp}
A state vector $\ket{\psi}$ on $n$ qubits is called \emph{$t$-doped} if there exists a Clifford unitary $C \in \mathcal{C}_n$ and a $t$-qubit state vector $\ket{\phi}_t$ such that
\begin{align}
\ket{\psi} = C\bigl(\ket{0}^{\otimes (n-t)} \otimes \ket{\phi}_t\bigr).
\end{align}
We denote by $\mathcal{S}_{n,t}$ the set of all $t$-doped pure states, and by $\hat{\mathcal{S}}_{n,t}$ its convex hull.
\end{definition}

This definition induces a natural hierarchy of quantum states indexed by the parameter $t$. When $t=0$, the class $\mathcal{S}_{n,0}$ coincides exactly with the set of stabilizer states. More generally, the definition imposes a strong structural constraint on the state. When $C$ is the identity, the reference state vector $\ket{0}^{\otimes (n-t)}$ is stabilised by the operators $Z_1,\dots,Z_{n-t}$, implying that it admits a stabilizer group of size at least $2^{n-t}$. Since Clifford unitaries preserve stabilizer structure, every $t$-doped state is stabilised by at least $2^{n-t}$ Pauli operators. Consequently, for small values of $t$, such as $t=O(\log n)$, the state differs from a stabilizer state only on a logarithmic number of qubits, up to a Clifford transformation, and all non-stabilizer structure can be localised to a small subsystem. At the opposite extreme, when $t=\Theta(n)$, the class becomes universal, as any $n$-qubit state is trivially $n$-doped. In this sense, $t$ quantifies the effective ``non-stabilizer dimension'' of the state, namely the number of qubits required to store all its magic once stabilizer correlations have been disentangled.
From a computational perspective, $t$-compressibility guarantees efficient classical evaluation of measurement statistics whenever $t$ is small~\cite{bravyiImprovedClassicalSimulation2016}. The parameter $t$ also marks a sharp computational boundary. For $t=O(\log n)$, the associated classical simulation algorithms run in polynomial time. By contrast, when $t$ grows superlogarithmically, there exist $t$-doped stabilizer states that are believed to be classically intractable~\cite{guMagicInducedComputationalSeparation2025}: 
i.e., there must exist $t$-doped stabilizer states with $t=\omega(\log n)$ that cannot be efficiently simulated classically, in line with standard expectations for the emergence of quantum advantage.

The class of $t$-doped stabilizer states naturally contains many previously studied families of efficiently simulable states. In particular, any state obtained from a stabilizer state by applying at most $t$ local non-Clifford gates is a $O(t)$-doped stabilizer state~\cite{Leone_2024,leoneLearningTdopedStabilizer2024}. In the following lemma, we bound the number of $t$-doped stabilizer states given a tolerance $\varepsilon$ in trace distance. 
\begin{lemma}[Packing number of $t$-doped stabilizer states]\label{lem:covering-t-doped}
Let $\varepsilon > 0$, and let $\mathcal{S}_{n,t}$ denote the set of all pure $t$-doped stabilizer states on $n$ qubits. The $\varepsilon$-packing number of $\mathcal{S}_{n,t}$ with respect to the trace distance satisfies
\begin{align}
\mathscr{P}(\mathcal{S}_{n,t}, \varepsilon) \le 2^{2n^2+3n}\left(\frac{3}{\varepsilon}\right)^{2^{2t}}.
\end{align}
\begin{proof}
By definition, any $t$-doped stabilizer state can be written, using~\cref{def:tcomp}, in the form
\begin{align}
\ket{\psi} = C\ket{0}^{\otimes (n-t)} \otimes \ket{\phi}_t,
\end{align}
where $\ket{\phi}_t$ is a state vector on $t$ qubits, up to the action of a Clifford unitary on $n$ qubits. Consequently, an $\varepsilon$-packing of $\mathcal{S}_{n,t}$ can be constructed by taking the orbit of an $\varepsilon$-packing of the set $\mathcal{V}_{2^t}$ of pure $t$-qubit states under the action of the $n$-qubit Clifford group. This yields the upper bound $\mathscr{P}(\mathcal{S}_{n,t}, \varepsilon) 
\le \lvert \mathcal{C}_n \rvert \, \mathscr{P}(\mathcal{V}_{2^t}, \varepsilon)$,
where $\mathcal{C}_n$ denotes the $n$-qubit Clifford group.
Using the standard upper bound $\lvert \mathcal{C}_n \rvert \le 2^{n^2+2n}\prod_{j=1}^{n}(4^j-1)$ and applying~\cref{prop:packingupplow}, which gives $\mathscr{P}(\mathcal{V}_{2^t}, \varepsilon) \le \left(\frac{3}{\varepsilon}\right)^{2^{2t}}$, we obtain
\begin{align}
\mathscr{P}(\mathcal{S}_{n,t}, \varepsilon) 
\le 2^{n^2+2n}\prod_{j=1}^{n}(4^j -1)\left(\frac{3}{\varepsilon}\right)^{2^{2t}}.
\end{align}
Finally, using the bound $\prod_{j=1}^{n}(4^j-1)\le 2^{n(n+1)}$, we can conclude that
\begin{align}
\mathscr{P}(\mathcal{S}_{n,t}, \varepsilon) 
\le 2^{2n^2+3n}\left(\frac{3}{\varepsilon}\right)^{2^{2t}}.
\end{align}
\end{proof}
\end{lemma}

\section{Computational problems in magic-state resource theory}
In this section, we define the computational problems addressed in this paper within the framework of magic state resource theory. We begin by restating the definition introduced in \cref{sec:convexsetspreliminaries}, adapting it to the case of convex sets of quantum states. This specialisation is motivated by the fact that the stabilizer polytope, denoted by $\hat{\mathcal{S}}_n$, is a convex set by definition. However, we keep the discussion as general as possible by referring to an arbitrary convex set of states, which we denote by $\mathcal{K}$ throughout.

\begin{definition}\label{def:set1s} Let \(\mathcal K \subseteq \mathcal{H}_n \) be a convex subset of $\mathcal{H}_n$. We define the following.  
\begin{itemize}
    \item We define  
    $
        S(\mathcal K, \varepsilon) \coloneqq \{\rho \in \mathcal{H}_n \mid \exists\, \sigma \in \mathcal K : \|\rho - \sigma\|_2 \le \varepsilon\}.
    $
    \item We define  
    $
        S(\mathcal K, -\varepsilon) \coloneqq \{\rho \in \mathcal K \mid S(\{\rho\}, \varepsilon) \subseteq\mathcal K\},
    $
    where \(\{\rho\}\) denotes the singleton set containing the state \(\rho\).
\end{itemize}

\end{definition}
This definition mimics the one given for convex sets of $\mathbb{R}^{D}$, introducing a notion of the inner core and outer neighbourhood for convex sets of states. The careful reader will notice that in \cref{def:set1s} we use the Frobenius norm instead of the trace distance, which is the metric typically used to characterise the distance between quantum states. The two norms are related by the inequalities 
$\|\rho-\sigma\|_2 \le \|\rho-\sigma\|_1 \le \sqrt{d}\,\|\rho-\sigma\|_2$, 
that is, they differ by at most a factor polynomial in the dimension of the Hilbert space. Since the access model in the computational problems studied here consists of a $d \times d$ density matrix, the two distances are therefore equivalent up to polynomial factors. We thus work with the Frobenius norm because of its direct connection to the Euclidean $\ell_2$-norm on vectors, as we explain below. We can also generalise the notion of a well-bounded centered set, and we obtain the following. 

\begin{definition}[Well-bounded centered set for states]\label{def:set2states}
Let $\mathcal K$ be a convex subset of $\mathcal{H}_n$. Then, $\mathcal K$ is called \emph{well-bounded and \(\overline \rho\)-centered} if there exist rationals \(R > 0\), \(r > 0\), and a state \(\overline\rho \in \mathcal K\) such that, with the 
state \(I/d \in \mathcal{H}_n\),
\begin{align}
  \mathcal  K &\subseteq S(I/d, R), \\
   \mathcal K &\supseteq S(\overline \rho, r) .
\end{align}
\end{definition}
Again, this notion provides a generalisation of the concept of a \emph{well-bounded centered} set, but for states.
We can now proceed to introduce the notion of the weak membership problem for a set of quantum states. We define it for a general convex set $\mathcal{K}$, since in what follows we will consider not only the stabilizer polytope but also the convex hull of $t$-doped states introduced in \cref{sec:dopedstates}.

\begin{problem}[$\rho$-weak membership problem ($\rho$-\class{WMEM})]\label{prob:rwmem} Let $\mathcal K\subseteq \mathcal{H}_n$ be a convex and well-bounded $\overline{\rho}$-centered set for states on $n$ qubits. Let $\rho\in\mathcal{D}(\mathcal{H}_n)$ be a quantum state and $\varepsilon>0$ a error parameter. Decide:
\begin{itemize}
    \item If  $\rho\in S(\mathcal K, -\varepsilon)$  output \class{YES}.
    \item If $\rho\notin S(\mathcal K,\varepsilon)$  output \class{NO}. 
\end{itemize}
An instance of  $\rho$-\class{WMEM} is denoted as $\rho\text{-}\class{WMEM}_\varepsilon(\mathcal K)$.
\end{problem}
Let us now define a (polynomially) equivalent version of the membership problem, namely the $1$-sided weak membership problem.
\begin{problem}[$1$-sided $\rho$-weak membership problem $\rho$-$\class{WMEM}^1$]~\label{prob:1rweak} Let $\mathcal K\subseteq \mathcal{H}_n$ be a convex and well-bounded $\overline{\rho}$-centered set for $n$ qubits. Let $\rho\in\mathcal{D}(\mathcal{H}_n)$ be a quantum state and $\varepsilon>0$ a error parameter. Decide:
\begin{itemize}
    \item If  $\rho\in\mathcal{K}$  output \class{YES};
    \item If $\rho\notin S(\mathcal K,\varepsilon)$  output \class{NO}. 
\end{itemize}
An instance of  $\rho$-$\class{WMEM}^1$ is denoted as $\rho\text{-}\class{WMEM}^1_\varepsilon(\mathcal K)$.
\end{problem}


Here, we have presented a generalisation of the membership problem to an arbitrary convex set of quantum states. Since quantum states can be represented as vectors in $\mathbb{R}^{D}$ with $D = d^2 - 1$~\cite{bengtsson2006geometry},  $\rho\text{-}\class{WMEM}_{\varepsilon}(\mathcal{K})$  can be naturally related to $\class{WMEM}_{\varepsilon}(K)$. The following lemma formalise this relationship by establishing the equivalence of the two problems.

\begin{lemma}[Equivalence between $\rho$-$\class{WMEM}^{(1)}$ and $\class{WMEM}^{(1)}$]\label{lemma:rmemwmem}
Let $\mathcal{K}$ be a convex and well-bounded $\overline{\rho}$-centered set for quantum states on $\mathcal{H}_n$, and let $K\subseteq\mathbb{R}^{D}$ with $D=d^2-1$ denote its image under the Pauli-coordinate map $\mathbf{x}:\rho\mapsto \mathbf{x}(\rho)$, where $\mathbf{x}(\rho)=(\mathrm{tr}(\rho P)/\sqrt{d}:P\neq I)$. The set $K$ is well-bounded and centered in the origin with
\begin{align}
    R&=\sqrt{1-1/d},\\
    r&=\frac{1}{\sqrt{d(d-1)}}.
\end{align}
Moreover an instance of $\rho$-\class{WMEM}$_\varepsilon(\mathcal{K})$ is equivalent to an instance of \class{WMEM}$_\varepsilon(K)$. Similarly, an instance  of $\rho$-$\class{WMEM}^1_\varepsilon(\mathcal{K})$ is equivalent to an instance of $\class{WMEM}^1_\varepsilon(K)$.
\begin{proof}
Every state has a Pauli expansion $\rho=\frac{1}{d}I+\frac{1}{d}\sum_{P\neq I}x_P P$ with $x_P=\mathrm{tr}(\rho P)$. For pure states, the coefficients satisfy $\sum_{P\neq I}x_P^{2}=1-1/d$, which shows that the entire state space embeds into the Euclidean ball $B(0,\sqrt{1-1/d})\subset\mathbb{R}^D$. Thus $K=\mathbf{x}(\mathcal{K})$ is compact and convex; moreover, by Ref.\ \cite[Section 1.2]{bengtsson2006geometry},  the full state space is a well-bounded centered set with radii $R=\sqrt{(1-1/d)}$ and $r=\sqrt{1/(d^2-d)}$.
The map $\mathbf{x}$ maps the Hilbert-Schmidt distance 
\begin{align}
\|\rho-\rho'\|_2=\|\mathbf{x}(\rho)-\mathbf{x}(\rho')\|_2 
\end{align}
to the Euclidean distance with no prefactors. Hence, an $\varepsilon$-membership test in state space coincides exactly with an $\varepsilon$-membership test in the image set $K$. This proves the claim.
\end{proof}
\end{lemma}

After defining the membership problem, one can extend the notion of a separating hyperplane. In this sense, we can discuss \emph{witnesses}, operators that separate the states belonging to the convex set \(\mathcal{K}\) from those that do not belong to it. This leads to the following definition:

\begin{definition}[Witness for a convex set of states]\label{def:witK}
Let $\mathcal{K}$ be a convex set of states on $\mathcal{H}_n$, let $W$ be a Hermitian operator, and let $\gamma\in\mathbb{Q}$. We say that $W$ is a \emph{witness} for $\mathcal{K}$ if
\begin{align}\label{eq:wit}
    \max_{\sigma\in\mathcal{K}} \,\mathrm{tr}(W\sigma)\le \gamma .
\end{align}
\end{definition}
Some authors define a witness for a convex set by additionally requiring the existence of at least one state $\rho \notin \mathcal{K}$ such that $\tr(W\rho) > \gamma$, thereby ensuring that $W$ is a ``useful'' witness. We instead adopt the above definition without this extra condition in order to guarantee that the set of witnesses remains convex, which is crucial for the proof of \cref{wmemeforstabinqp2}, and centered at the identity, which is, of course, a ``useless'' witness.

Note that if $\mathcal{V}$ is a set of pure states and $\mathcal{K}=\mathrm{conv}(\mathcal{V})$ is their convex hull, then the condition of \cref{eq:wit} reduces to
\begin{align}
    \max_{\psi\in\mathcal{V}} \langle \psi \vert W \vert \psi\rangle \le \gamma.
\end{align}
This follows because the map $\sigma\mapsto\mathrm{tr}(W\sigma)$ is linear, and a linear functional attains its maximum over a convex set at one of its extreme points. Since the extreme points of $\mathcal{K}$ are precisely the pure states in $\mathcal{V}$, it suffices to verify the inequality only on $\mathcal{V}$.

Related to the notion of a witness, we state the following lemma, which explores the relationship between membership and witnesses.

\begin{lemma}[Witness vs. weak membership]\label{witweak} 
Given a state \(\rho\) such that \(\|\rho - \sigma\|_2 > 0\) for all \(\sigma \in \mathcal{K}\), there exists a witness operator \(W\) that detects \(\rho\) in the sense that
\end{lemma}

\begin{align}
    \tr(W\rho)> \max_{\sigma\in \mathcal{K}}\tr(W\sigma)+\min_{\sigma\in\mathcal{K}}\|\rho-\sigma\|^2_2\,.
\end{align} 
Conversely, if there exists a witness that detects $\rho$, i.e., $\tr(W\rho)>\max_{\sigma\in\mathcal{K}}\tr(W\sigma)+\varepsilon$, then it holds that $\min_{\sigma\in\mathcal{K}}\|\rho-\sigma\|_2> \frac{\varepsilon}{\|W\|_2}$.

\begin{proof}
    First of all, let us define $\tau=\arg \min_{\sigma\in \mathcal K}\|\rho-\sigma\|_2$. We have $\|\rho-\tau\|_2>\varepsilon$. The state $\tau+t(\sigma-\tau)\in\mathcal K$ for $\sigma\in\mathcal{S}_n$ since it 
    is a convex combination of 
    $\tau$ and $\sigma$. 
    The function $f(t)=\|\rho-\tau-t(\sigma-\tau)\|_2^2=\tr[(\rho-\tau-t(\sigma-\tau))^2]$ achieves its minimum at $t=0$ by construction. Differentiating with respect to $t$, we have
    \begin{align}
        f'(t)=2\tr[(\rho-\tau-t(\sigma-\tau))(\tau-\sigma)]\implies f'(0)=\tr((\rho-\tau)(\tau-\sigma)).
    \end{align}
Since $t=0$ is a minimum, we have $f'(0)\ge 0$. We define the witness as $W=\rho-\tau$. Indeed, we have
\begin{align}
    \tr(W(\rho-\sigma))&=\tr((\rho-\tau)(\rho-\sigma))\\
    \nonumber
    &=\tr((\rho-\tau)(\rho-\tau))+\tr((\rho-\tau)(\tau-\sigma))\\
    &\ge\|\rho-\tau\|_2^2,
    \nonumber
\end{align}
proving the first statement. Conversely, if there exists $W$ such that $\tr(W\rho)> \max_{\sigma} \tr(W\sigma)+\varepsilon$. It follows that for any $\tau\in \mathcal{K}$, it holds that $\tr(\tau W)\le \max_{\sigma}\tr(W\sigma)< \tr(W\rho)-\varepsilon$. Hence, we have $\varepsilon<\tr(W[\rho-\tau])\le \|W\|_2\|\rho-\tau\|_2$, proving the statement.
\end{proof}

Similarly to what was done for convex sets of $\mathbb{R}^D$, we can introduce the problem of deciding whether an operator is a witness or not.

\begin{problem}[Weak witnessing detection (\class{WWD})]\label{problem:WWD}
Let $\mathcal{K}\subseteq\mathcal{B}(\mathcal{H}_n)$ be a convex set and well-bounded $\overline{\rho}$-centered set for $n$-qubit states, let $W\in\mathcal{B}(\mathcal{H}_n)$ satisfy $\|W\|_2\le 1$, and let $\gamma,\delta\in\mathbb{Q}$ with $\delta>0$. Decide:
\begin{itemize}
    \item If some $\rho\in S(\mathcal{K},-\delta)$ satisfies $\tr(W\rho)\ge \gamma+\delta$, output \class{YES}.
    \item If all $\rho\in S(\mathcal{K},\delta)$ satisfy $\tr(W\rho)\le \gamma-\delta$, output \class{NO}.
\end{itemize}
An instance 
is denoted $\class{WWD}_{\delta}(\mathcal{K})$.
\end{problem}
The careful reader may have noticed that, in the promise of \cref{problem:WWD}, a witness $W$ is required to satisfy $\|W\|_2 \le 1$, whereas in the definition of a witness given in \cref{def:witK} no such condition is imposed. This discrepancy merely reflects a rescaling of the constant $\gamma \in \mathbb{Q}$ and ultimately only affects the scaling in the reduction used in \cref{lem:reductionswmemwwd}.

\begin{lemma}\label{lem:wwdwopt}
An instance of $\class{WWD}_{\delta}(\mathcal{K})$ with parameters $(W,\gamma)$ corresponds to $\class{WOPT}_{\delta}(K)$ with parameters $(\mathbf{w},\gamma-\tr(W)/d)$, where $K=\mathbf{x}(\mathcal{K})\subset\mathbb{R}^{D}$ with $D=d^2-1$ and $\mathbf{x}(\rho)$ is the Pauli-coordinate map, and $\mathbf{w}=\mathbf{x}(W)$.
\begin{proof}
Mapping $W$ to $\mathbf{w}=(\tr(WP):P\neq I)$ gives $\tr(W\rho)=\mathbf{w}^T\mathbf{x}+\tr(W)/d$ for any $\rho$ with coordinates $\mathbf{x}$. optimising $\tr(W\rho)$ over $\mathcal{K}$ is therefore equivalent to optimising $\mathbf{w}^T\mathbf{x}$ over $K$, up to the constant shift $\tr(W)/d$. For a \class{YES} instance of $\class{WOPT}_{\delta}(K)$, the condition $\mathbf{w}^T\mathbf{x}\ge\gamma-\tr(W)/d+\delta$ implies $\tr(W\rho)\ge\gamma+\delta$, matching the \class{YES} case of $\class{WWD}_{\delta}(\mathcal{K})$. For a \class{NO} instance, $\mathbf{w}^T\mathbf{x}\le\gamma-\tr(W)/d-\delta$ yields $\tr(W\rho)\le\gamma-\delta$, matching the \class{NO} case.
\end{proof}
\end{lemma}
Given the connection between $\class{WMEM}_{\varepsilon}$ and $\class{WOPT}_{\delta}$, one can expect an analogous relation between $\rho$-$\class{WMEM}_{\varepsilon}$ and $\class{WWD}_{\delta}$. This correspondence follows directly from~\cref{thm:wmemimpwopt} and from the observation that $\rho$-$\class{WMEM}_{\varepsilon}$ can be embedded into $\class{WMEM}_{\varepsilon}$, while $\class{WWD}_{\delta}$ can be embedded into $\class{WOPT}_{\delta}$. For completeness, we state the following lemma to formalise this connection.

\begin{lemma}[Reduction from $\rho$-$\class{WMEM}_{\varepsilon}$ to $\class{WWD}_\delta$]\label{lem:reductionswmemwwd}
Let $\mathcal{K}$ be a convex and well-bounded $\overline{\rho}$-centered set for quantum states. Given an instance of $\class{WWD}_\delta(\mathcal{K})$ with parameters $(W, \gamma)$, where $W$ is an operator on $\mathcal{H}_n$ with $\|W\|_2 \leq 1$, $\gamma \in \mathbb{Q}$, and $\delta \in (0,1)$, there exists an algorithm that runs in time $\class{poly}(d, \delta^{-1})$ and solves $\class{WWD}_{\delta}(\mathcal{K})$ using an oracle for $\rho$-$\class{WMEM}_{\varepsilon}(\mathcal{K})$ with $\varepsilon = \frac{\delta^3}{2^{13} \cdot 3^3 \cdot d^{13}}$.
\end{lemma}

\begin{proof}The proof follows directly from~\cref{lemma:rmemwmem,lem:wwdwopt}. Given the equivalence between the protocols $\class{WMEM}_{\varepsilon}$ and $\rho$-$\class{WMEM}_{\varepsilon}$, as well as between $\class{WOPT}_\delta$ and $\class{WWD}_\delta$, we can construct a polynomial-time algorithm for $\class{WWD}_\delta$ with oracle access to $\rho$-$\class{WMEM}_{\delta}$. The approach proceeds by first reducing the weak witness problem to its optimisation counterpart $\class{WOPT}_\delta$, then applying~\cref{thm:wmemimpwopt} to complete the reduction. 
\end{proof}
To conclude this preliminary section, we recall that the two membership problems defined above, namely $\rho$-$\class{WMEM}$ and $\rho$-$\class{WMEM}^1$, are polynomially equivalent. This correspondence follows directly from~\cref{thm:wmem1wmem} and from the observation that $\rho$-$\class{WMEM}_{\varepsilon}$ can be embedded into $\class{WMEM}_{\varepsilon}$, while $\rho$-$\class{WMEM}^1$ can be embedded into $\class{WMEM}^1$. For completeness, we state the following lemma to formalise this connection.

\begin{lemma}[Reduction from $\rho$-$\class{WMEM}^{1}$ to $\rho$-$\class{WMEM}$]\label{lem:onesideequivalentbothss} Let $\mathcal{K}$ be a convex and well-bounded $\overline{\rho}$-centered set for quantum states. Given an instance of $\rho$-$\class{WMEM}_\varepsilon(\mathcal{K})$, and $\varepsilon \in (0,1)$, there exists an algorithm that runs in time $\class{poly}(d, \delta^{-1})$ and solves $\rho$-$\class{WMEM}_\varepsilon(\mathcal{K})$ using an oracle for $\rho$-$\class{WMEM}_{\delta}^{1}(\mathcal{K})$ with $\delta = \frac{\varepsilon^3}{2^{12} \cdot 3 \cdot d^{12}}$.
    
\end{lemma}
\begin{proof}
   The proof follows directly from~\cref{lemma:rmemwmem}. Given the equivalence between the protocols $\class{WMEM}_{\varepsilon}$ and $\rho$-$\class{WMEM}_{\varepsilon}$, as well as between $\class{WMEM}^{1}_\delta$ and $\rho$-$\class{WMEM}^{1}_\delta$, we can construct a polynomial-time algorithm for $\rho$-$\class{WMEM}_{\varepsilon}$ with oracle access to $\rho$-$\class{WMEM}^{1}_{\delta}$. The approach proceeds by first reducing $\rho$-$\class{WMEM}_{\varepsilon}$ to $\class{WMEM}_{\varepsilon}$ , then applying~\cref{thm:wmem1wmem} to complete the reduction. 
\end{proof}
\subsection{Weak membership for the stabilizer polytope}\label{sec:wmemstab}

This section investigates the computational problem of determining membership in the convex hull of pure stabilizer states.
Formally, letting \( \mathcal{S}_n \) denote the set of pure stabilizer states on $n$ qubits and \( \hat{\mathcal{S}}_n \) its convex hull, we analyse the complexity of the weak membership problem \( \rho\text-\class{WMEM}_{\varepsilon}(\hat{\mathcal{S}}_n) \). 
We begin by examining the hardness of the decision problem \(\rho\text-\class{WMEM}_{\varepsilon}(\hat{\mathcal{S}}_n) \).

\begin{theorem}[Superpolynomial complexity of $\rho$-$\class{WMEM}_{\varepsilon}(\hat{\mathcal{S}_n})$]\label{thm:nphardrhowmem}
    Let $\hat{\mathcal{S}_n}$ denote the stabilizer polytope on $n$ qubits. Under the \emph{exponential time hypothesis}  (ETH) (see Definition~\ref{def:eth}), for any $\varepsilon \le \frac{1}{2^{13} 3^3 d^{53/2}\log^{18}d }$, the decision problem $\rho$-$\class{WMEM}_{\varepsilon}(\hat{\mathcal{S}_n})$ is not contained in the complexity class $\class{QP}^{2-\eta}$ for any $\eta > 0$.
\end{theorem}

\begin{proof}
   We prove that $\rho$-$\class{WMEM}_{\varepsilon}(\hat{\mathcal{S}}_n)$ belongs to $\class{QP}^2$ in \cref{wmemeforstabinqp2}. To show that this inclusion is strict, we reduce the membership problem to the weak witness detection problem $\class{WWD}_{\delta}(\hat{\mathcal{S}}_n)$. In \cref{thm:nphardwwd}, we show that $\class{WWD}_{\delta}(\hat{\mathcal{S}}_n)$ does not belong to $\class{QP}^{2-\eta}$ for any $\eta>0$, assuming the exponential time hypothesis stated in \cref{def:eth}, provided that $\delta \le \frac{1}{d^{9/2} \log^6d}$. The reduction follows from the existence of a polynomial-time reduction between $\rho$-$\class{WMEM}_{\varepsilon}(\hat{\mathcal{S}}_n)$ and $\class{WWD}_{\delta}$, established in \cref{lem:reductionswmemwwd}. As a consequence, $\rho$-$\class{WMEM}_{\varepsilon}(\hat{\mathcal{S}}_n)$ with $\varepsilon \le \frac{1}{2^{13} 3^3 d^{53/2}\log^{18}d }$ requires time $2^{\Omega(\log^2d)}$, which implies the desired lower bound. This concludes the proof. 
\end{proof}
    It should be remarked that while the promised gap $\varepsilon$ in \cref{thm:nphardrhowmem} is small, it is ``only'' polynomially small in the input parameter $d$ and such a situation is common in 
conceptual hardness proofs of this kind. Original proofs stating the hardness of testing membership in the set of bipartite separable quantum states featured promise 
gaps that 
were even inversely exponential in the dimension of the quantum system and hence doubly exponential in the system size \cite{gurvitsClassicalDeterministicComplexity2003,ioannouComputationalComplexityQuantum2007}. When this was brought down to an inverse polynomial scaling, a situation reminiscent of the one encountered here, this was dubbed the ``strong \class{NP}-hardness of the quantum separability problem'' \cite{gharibianStrongNPHardnessQuantum2009}. 
The underlying reason for this behavior lies in the structure of the reductions: rather than proving hardness directly, one relies on a chain of polynomial-time reductions between different computational problems. In our case, following the same strategy as in the proof of strong \class{NP}-hardness for quantum separability, we exploit the polynomial equivalence between the classes $\class{WMEM}_{\varepsilon}$ and $\class{WOPT}_{\delta}$. This equivalence, however, comes at a quantitative cost, as made explicit in \cref{thm:wmemimpwopt}, where one must fix $\varepsilon = O(\delta^3 / d^{13})$. In principle, these constants could be improved in two ways: either by strengthening the polynomial-time reduction between $\class{WMEM}_{\varepsilon}$ and $\class{WOPT}_{\delta}$, or by shortening the overall reduction chain. For instance, a direct reduction showing that oracle access to $\rho$-$\class{WMEM}_{\varepsilon}$ suffices to solve $3$-\class{SAT} would lead to significantly improved parameter dependencies.


\begin{lemma}[$\rho$-$\class{WMEM}_{\varepsilon}(\hat{\mathcal{S}}_{n})\in \class{QP}^2$]\label{wmemeforstabinqp2}
    For any $0<\varepsilon<1$, there exists an algorithm running in time $\class{poly}(2^{\log^2d},\log\frac{1}{\varepsilon})$ that solves $\rho$-$\class{WMEM}_{\varepsilon}(\hat{\mathcal{S}}_n)$.
    \begin{proof}
         As noted in Ref.~\cite{howardApplicationResourceTheory2017}, the optimisation for the robustness of magic (see \cref{def:robustnessofmagic}) can be written as the following linear program:
\begin{align}\label{robustnessoptimisation}
    \mathcal{R}(\rho)=\max_{ \mathbf{w}\,:\,\|A^{T}\mathbf{w}\|_\infty\le 1}(-\mathbf{x}^T\cdot \mathbf{w})
\end{align}
where $A_{i,j}=\frac{1}{\sqrt{d}}\tr(P_i\sigma_j)$ and $\mathbf{x}_i=\frac{1}{\sqrt{d}}\tr(P_i\rho)$. \cref{robustnessoptimisation} states that the dual problem associated with computing the robustness of magic is precisely the maximisation of a magic witness. Indeed, it can be recast as
\begin{align}\label{eq:linearprogramrobustness}
    \mathcal{R}(\rho)=\max_{\substack{W\\\tr(W\sigma)\le1\,\,\forall\sigma\in\mathcal{S}_n}}\tr(W\rho).
\end{align}
Indeed, the condition $\|A^T\mathbf{w}\|_{\infty}\le 1$ corresponds to the condition $\max_{\sigma\in\mathcal{S}_n}\tr(W\sigma)\le 1$, which is exactly the definition of a witness for a convex set (see \cref{def:witK} with $\gamma=1$). We define the set of such witnesses $\mathcal{W}\coloneqq\{W\,:\, \max_{\sigma\in\mathcal{S}_n}\tr(\sigma W)\le 1\}$. So far we have implicitly assumed an ideal setting in which the full space of witnesses can be accessed exactly. We now relax this assumption and consider the more realistic case in which only a restricted subset of witnesses is available, due to finite numerical precision. Given that the set of witnesses is a convex set, we can define its inner and outer boundary according to \cref{def:set1}. It is immediate to see that we can write the inner boundary of $\mathcal{W}$ as
\begin{align}
    S(\mathcal{W},-\nu) = \{\, W \in \mathcal{W} \;:\; W + H \in \mathcal{W} \ \text{for all } H \text{ with } \|H\|_2 \le \nu \,\}.
\end{align}
That is, $S(\mathcal{W},-\nu)$ consists of witnesses that remain feasible under any Hilbert--Schmidt perturbation of norm at most $\nu$. For the reminder of the proof, we adopt the lighter notation $\mathcal{W}_{\nu}\equiv S(\mathcal{W},-\nu)$.
The corresponding optimisation problem is
\begin{align}
    \mathcal R_{\nu}(\rho) \coloneqq \max_{W \in \mathcal{W}_{\nu}} \tr(\rho W),
\end{align}
which is a linear program over $\mathcal{W}_{\nu}$. This problem can be solved up to additive error $\nu$ using the ellipsoid algorithm~\cite{grotschelGeometricAlgorithmsCombinatorial1988} in time $\poly\bigl(2^{\log^2d}, \log\nu^{-1}\bigr)$ provided that $\mathcal{W}_{\nu}$ is nonempty.
Our goal is to show that the linear program associated with the robustness of magic can be used to solve the one-sided problem $\rho$-\class{WMEM}$^{1}_{\delta}(\hat{\mathcal{S}}_{n})$ (see~\cref{prob:1rweak}). The full weak membership problem \class{WMEM}$_{\varepsilon}(\hat{\mathcal{S}}_{n})$ then follows by invoking~\cref{lem:onesideequivalentbothss}.  

We begin with the \class{YES} case. Let $\rho \in \hat{\mathcal{S}}_{n}$. By construction, any admissible witness $W \in \mathcal{W}_{\nu}$ satisfies $\tr(W\rho)\le 1$. Since the linear program is solved only up to additive error $\nu$, the algorithm may return a value of at most $1+\nu$. Therefore, no \class{YES} instance can be misclassified as \class{NO}, provided that for every $\rho \notin S(\hat{\mathcal{S}}_{n},\delta)$ the value returned by the algorithm is strictly larger than $1+\nu$.

We now turn to the \class{NO} case. Let $\rho$ be a state such that $    \min_{\sigma \in \hat{\mathcal{S}}_{n}} \|\rho-\sigma\|_2 \ge \delta$ . Let us show that there exists a witness $W$ such that $\tr(W\rho)>1+\nu$. By \cref{witweak}, there exists a Hermitian operator $\tilde W$ satisfying
\begin{align}
    \tr(\tilde W\rho) > \max_{\sigma\in\hat{\mathcal{S}}_{n}}\tr(\tilde W\sigma) + \delta^2 .
\end{align}

We now show how to rescale $\tilde W$ into a valid element of $\mathcal{W}_{\nu}$. Recall that for $W\in\mathcal{W}_{\nu}$ the constraint $\tr[(W+H)\sigma]\le 1$ for all $\sigma\in\hat{\mathcal{S}}_{n}$ implies, by linearity and Hölder's inequality, that $\tr(W\sigma)\le 1-\nu$. Define
\begin{align}
    W \coloneqq \frac{\tilde W}{\max_{\sigma\in\hat{\mathcal{S}}_{n}}\tr(\tilde W\sigma)} - \nu I .
\end{align}
By construction, $W\in\mathcal{W}_{\nu}$.  Evaluating $W$ on $\rho$ yields
\begin{align}
    \tr(W\rho)
    &= \frac{\tr(\tilde W\rho)}{\max_{\sigma\in\hat{\mathcal{S}}_{n}}\tr(\tilde W\sigma)} - \nu \\
    &> 1 + \frac{\delta^2}{\max_{\sigma\in\hat{\mathcal{S}}_{n}}\tr(\tilde W\sigma)} - \nu .
    \nonumber
\end{align}
In order to distinguish \class{YES} from \class{NO} instances, we require $\tr(W\rho)>1+\nu$, which is ensured if
\begin{align}
    \frac{\delta^2}{\max_{\sigma\in\hat{\mathcal{S}}_{n}}\tr(\tilde W\sigma)} > 3\nu ,
\end{align}
where we are accounting for the additive error $\nu$. Using the bound
\begin{align}
    \max_{\sigma\in\hat{\mathcal{S}}_{n}}\tr(\tilde W\sigma)
    \le \|\tilde W\|_2
    = \|\rho-\tau\|_2
    \le 2 ,
\end{align}
we obtain the sufficient condition $\nu \le \delta^2/6$.

Finally, note that the explicit construction above also shows that the feasible set $\mathcal{W}_{\nu}$ is nonempty, completing the implications that computing the robustness offers a way to perform $1$-sided $\rho$-$\class{WMEM}_{\delta}^{1}(\hat{\mathcal{S}}_{n})$. By~\cref{lemma:rmemwmem}, we know that $\rho$-$\class{WMEM}_{\delta}^{1}(\hat{\mathcal{S}}_{n})$ and $\rho$-$\class{WMEM}_{\varepsilon}(\hat{\mathcal{S}}_{n})$ are polynomially equivalent, as long $\delta=\frac{\varepsilon^3}{2^{12} \cdot 3 \cdot d^{12}}$, meaning that in order to solve $\rho$-$\class{WMEM}_{\varepsilon}$, $\nu\le \frac{\varepsilon^6}{2^{25} \cdot 3^3 \cdot d^{24}}$, implying the desired runtime for the algorithm.

    \end{proof}
\end{lemma}

\subsection{Weak witness detection}

In this section, we prove one of the main results of the paper, namely that determining whether a Hermitian operator $W$ functions as a witness for the stabilizer polytope does not belong to $\class{QP}^{2-\eta}$ for any $\eta>0$ under the ETH assumption \cref{def:eth}.

\begin{theorem}[Superpolynomial complexity of $\class{WWD}_{\delta}(\hat{\mathcal{S}_n})$]\label{thm:nphardwwd}
    Let $\hat{\mathcal{S}_n}$ denote the stabilizer polytope on $n$ qubits. Under the \emph{exponential time hypothesis} (ETH) (see Definition~\ref{def:eth}), for any $\delta \leq \frac{1}{d^{9/2}\log^6d}$, the decision problem $\class{WWD}_{\delta}(\hat{\mathcal{S}_n})$ is not contained in the complexity class $\class{QP}^{2-\eta}$ for any $\eta > 0$.
\end{theorem}

\begin{proof} The overall strategy for proving the theorem is to reduce the weak-witness detection problem to another decision problem, whose solution requires resolving a general $3$-\class{SAT} instance on $ O(\log^2 d)$ variables. Under the ETH (see \cref{def:eth}), this requires time $2^{\Omega(\log^2 d)}$. Since the input of the problem is a $d \times d$ density matrix, this implies that $\class{WWD}_{\delta}(\hat{\mathcal{S}}_n)$ does not lie in $\class{QP}^{2-\eta}$ for any $\eta > 0$, see \cref{def:quasipolyclasses}. In particular, we construct a Hamiltonian $H$ acting on $\log^2 d$ qubits of the type in \cref{hamiltonianformulation3sat}.

The proof is technically involved. Its key ingredient is the immediate observation, established below, that optimising a Hamiltonian over two copies of graph states already requires solving a $3$-\class{SAT} problem with $O(\log^2d)$ variables. The remainder of the proof introduces suitable penalty terms that encode this optimisation into the problem of optimising a Hamiltonian over a single stabilizer state, thereby completing the argument.\smallskip

\paragraph{\bf Local Hamiltonian.}
This part of the proof constructs an instance of the local Hamiltonian problem. The key insight exploits the structure of stabilizer states, specifically focusing on graph states \cite{Graphs,Hein06}.
Recall that a graph state vector \(\ket{G(A)}\) on \(n\) qubits, associated with an adjacency matrix \(A\), is defined as
\begin{equation}\label{eq:graphstates}
	\ket{G(A)} \coloneqq \prod_{(i,j): i < j} CZ_{i,j}^{A_{i,j}} \ket{+}^{\otimes n}
    = \frac{1}{\sqrt{d}} \sum_{\mathbf{x}\in \{0,1\}^n}
    (-1)^{\tfrac{1}{2}\mathbf{x}^T A \mathbf{x}} \ket{\mathbf{x}}.
\end{equation}
For each pair \((i,j)\), let $s_{ij}\in\{0,1\}$ and define the operators  
\begin{equation}\label{eq:operators graph}
\begin{aligned}
    X^{s_{ij} }(i,j) &\coloneqq d/4 \, \bigl(\ket{\mathbf{0}}+(-1)^{s_{ij}}\ket{i,j}\bigr) \bigl(\bra{\mathbf{0}}+(-1)^{s_{ij}}\bra{i,j}\bigr), \\
    Y^{s_{ij},s_{kl}}(i,j,k,l) &\coloneqq d/4 \, \bigl(\ket{\mathbf{0}}+(-1)^{s_{ij}}\ket{i,j}\bigr)\bigl(\bra{\mathbf{0}}+(-1)^{s_{kl}}\bra{k,l}\bigr),
\end{aligned}
\end{equation}
where \(\ket{i,j}\) denotes the computational basis state of Hamming weight \(2\) with ones at positions \(i\) and \(j\). 
Their expectation values on the graph state 
state vector \(\ket{G(A)}\) are given by  
\begin{equation}\label{eq:expectationoperator}
\begin{aligned}
    \matrixel{G(A)}{X^{s_{ij}}(i,j)}{G(A)} &= \frac{\bigl(1+(-1)^{s_{ij}}(-1)^{A_{i,j}}\bigr)^2}{4},\\
    \matrixel{G(A)}{Y^{s_{ij},s_{kl}}(i,j,k,l)}{G(A)} &= \frac{\bigl(1+(-1)^{s_{ij}} (-1)^{A_{i,j}}\bigr)\bigl(1+(-1)^{s_{kl}} (-1)^{A_{k,l}}\bigr)}{4}.
\end{aligned}
\end{equation}
Since there is a one-to-one correspondence between the operators \(X(i,j)\) and the matrix entries \(A_{i,j}\), it suffices to consider \(X(i,j)\) with \(j > i\). Similarly, for \(Y(i,j,k,l)\), we can restrict to \(i < j < k < l\). We, therefore, relabel them as
\begin{equation}
\{X_{\alpha}^{s_{\alpha}}\}_{\alpha=1}^{n(n-1)/2} \quad \text{and} \quad \{Y_{\beta\gamma}^{s_{\beta},s_{\gamma}}\}_{\beta<\gamma}^{n(n-1)/2},
\end{equation}
with each \(\alpha,\beta,\gamma\) corresponding to a specific pair \((i,j)\) and $s_{\alpha},s_{\beta},s_{\gamma}\in\{0,1\}$ denote the sign of the operator.
Now consider the  Hamiltonian 
\begin{equation}\label{eq:Hamiltonian}
    H = \sum_{\alpha<\beta<\gamma} X_{\alpha}^{s_{\alpha}} \otimes \frac{(Y_{\beta,\gamma}^{s_{\beta},s_{\gamma}}+Y_{\beta,\gamma}^{s_{\beta},s_{\gamma}\dag})}{2}
\end{equation}
acting on two copies of the Hilbert space \(\mathcal{H}^{\otimes 2}\). 
Its expectation value on two copies of the graph state is
\begin{equation}\label{eq:exphamiltonian}
\begin{aligned}
    \langle{G(A)}\vert^{\otimes 2}H\vert G(A)\rangle^{\otimes 2}
    &= \sum_{\alpha<\beta<\gamma}
        \left(\frac{1+(-1)^{A_{\alpha}+s_{\alpha}}}{2}\right)^{2}
        \left(\frac{1+(-1)^{A_{\beta}+s_{\beta}}}{2}\right)
        \left(\frac{1+(-1)^{A_{\gamma}+s_{\gamma}}}{2}\right) \\
    &= \sum_{\alpha<\beta<\gamma} (x_{\alpha}\oplus{s_{\alpha}})(x_{\beta}\oplus{s_{\beta}})(x_{\gamma}\oplus {s_{\gamma}}),\nonumber
\end{aligned}
\end{equation}
where $x_{\alpha}^2=x_{\alpha}$ and
\begin{equation}
x_{\alpha}\coloneqq  \left(\frac{1+(-1)^{A_{\alpha}}}{2}\right)\in \{0,1\}\,.
\end{equation}
Moreover $\oplus$ denotes the binary sum modulo $2$. We have successfully encoded a generic $3$-\class{SAT} instance (see \cref{hamiltonianformulation3sat}) in the expectation values of graph states.
Let us now define $\mathcal{G}_n^{(2)} \coloneqq \{ \ket{G(A)}^{\otimes 2} : A \in \mathbb{F}_2^{\binom{n}{2} \times \binom{n}{2}} \}$ as the set of identical bipartite graph states on $2n$ qubits. Consider a Hamiltonian $H$ as in \cref{eq:Hamiltonian}, with arbitrary $\alpha, \beta, \gamma \in [n(n-1)/2]$. We consider the following decision problem:
\begin{itemize}
    \item Output \class{YES}, if $\max_{\ket{\psi} \in \mathcal{G}_n^{(2)}} \langle \psi | H | \psi \rangle = 0$.
    \item Output \class{NO}, if $\max_{\ket{\psi} \in \mathcal{G}_n^{(2)}} \langle \psi | H | \psi \rangle \ge 1$.
\end{itemize}
This problem is at least as hard as solving a $3$-\class{SAT} instance with ${n(n-1)}/{2}=\Omega(\log^2d)$ variables. Assuming the ETH  (\cref{def:eth})
again, such an instance requires time $2^{\Omega(\log^2d)}$, and therefore the above decision problem cannot belong to the class $\class{QP}^{2-\eta}$ for any $\eta > 0$.
\smallskip

\paragraph{\bf Extension to single copy graph state.}
The previous step alone does not suffice to establish the hardness of $\class{WWD}(\hat{\mathcal{S}}_n)$, as it only demonstrates hardness for 2-copy graph states. The next step, therefore, shows that the problem remains hard even for single-copy graph states by reducing the 2-copy case to the single-copy case. The key idea is to extend the original Hamiltonian \(H\) to a new operator \(H^{(1)}\) such that minimising its expectation value over single-copy graph states is at least as hard as minimising $H$ over the set $\mathcal{G}_n^{(2)}$. 
To proceed, let us notice that the set 
\begin{equation}
    \mathcal{G}_{n}^{(2)}\subset \{\ket{G(A)}\,:\, A\in\mathbb{F}_2^{\binom{2n}{2}\times \binom{2n}{2}}\}\eqqcolon\mathcal{G}_{2n}.
    \end{equation}
    Hence, our strategy is to force, through a penalty Hamiltonian, the low energy states to lie in $\mathcal{G}_{n}^{(2)}$. Let us label the entries of the adjacency matrix as \(i_A\) for the first \(n\) columns and \(i_B\) for the last \(n\) columns, with the same convention applied to rows. At this point, we define the operator \begin{align}
    \bar{X}(i,j) \coloneqq \frac{d^2}{4}(\ket{\mathbf 0}-\ket{i,j})(\bra{\mathbf 0}-\bra{i,j}). 
    \end{align}
    which coincides with $X^{1}$ in \cref{eq:operators graph}. 
    Its expectation value over \(\ket{G(A)} \in \mathcal{G}_{2n}\) is given by
\begin{align}
    \langle G(A)|\bar{X}(i_A,j_B)|G(A)\rangle = \left(\frac{1-(-1)^{A_{i_A,j_B}}}{2}\right),
\end{align}
where we observe that this expectation value vanishes if and only if \(A_{i_A,j_B} = 0\).
We further introduce the operator
\begin{align}
W(i_A,j_A,k_B,l_B) \coloneqq d^2/4(\ket{\mathbf 0}-\ket{i_A,j_A} \otimes \ket{k_B,l_B})(\bra{\mathbf 0}-\bra{i_A,j_A} \otimes \bra{k_B,l_B}),
\end{align}
whose expectation value over \(\ket{G(A)} \in \mathcal{G}_{2n}\) reads
\begin{align}
   \langle G(A)|W(i_A,j_A,k_B,l_B)|G(A)\rangle = \frac{1-(-1)^{A_{i_A, j_A}+A_{k_B, l_B}}}{2},
\end{align}
which vanishes if and only if \(A_{i_A, j_A} = A_{k_B, l_B}\).

It follows that incorporating these terms into the Hamiltonian preserves the minimum obtained on two copies of the graph states while penalising all other states in the minimisation problem. We, therefore, consider the following Hamiltonian
\begin{align}
    H^{(1)}=H+(\|H\|_{\infty}+1)\left[\sum_{i,j}\bar{X}(i_A,j_B)+\sum_{i\neq j}W(i_A,j_A,i_B,j_B)\right].
\end{align}
By construction, we have 
\begin{align}
    \langle G(A)|H^{(1)}|G(A)\rangle=\langle G(A)|H|G(A)\rangle\,\quad\forall \ket{G(A)}\in\mathcal{G}_{n}^{(2)}
\end{align}
since the added positive operator vanishes if and only if the graph on \(2n\) vertices satisfies two conditions: first, there are no edges between the first \(n\) vertices (labeled by \(A\)) and the last \(n\) vertices (labeled by \(B\)); second, the connection pattern among the \(A\) vertices exactly matches that among the \(B\) vertices.

Let us now show that optimising the energy of $H^{(1)}$ over $\mathcal{G}_{2n}$  is at least as hard as optimising $H$ over $\mathcal{G}_{n}^{(2)}$. 
\begin{itemize}
    \item \class{YES.} Assume there exists a graph state vector $\ket{G(A)}\in\mathcal{G}_{n}^{(2)}$ such that $\langle G(A)|H|G(A)\rangle=0$, then, by the argument above, we also have $\langle G(A)|H^{(1)}|G(A)\rangle=0$.
    \item \class{NO.} Assume that $\min_{\ket{G(A)}\in \mathcal{G}_{n}^{(2)}}\langle G(A)|H|G(A)\rangle\ge 1$ (i.e., the second-to-last lowest value for $H$ over $\mathcal{G}_{n}^{(2)}$), then for any $\ket{G(A)}\in\mathcal{G}_{2n}\setminus \mathcal{G}_{n}^{(2)}$, we have
    \begin{align}
        \langle G(A)|H^{(1)}|G(A)\rangle\ge (\|H\|_{\infty}+1)-|\langle G(A)|H|G(A)\rangle|\ge 1\,,
    \end{align}
    where we have exploited the fact that the expectation value of $H$ is strictly positive (and greater than $1$) for graph states $\not\in\mathcal{G}_n^{(2)}$. This proves the \class{NO} instance.
\end{itemize}
\smallskip

\paragraph{\bf Extension to maximally coherent stabilizer states.} At this point, we have established that optimising a Hamiltonian over the set of graph states on $2n$ qubits reduces to solve a $3$-\class{SAT} with ${n(n-1)}/{2}$ variables. We now demonstrate that Hamiltonian optimisation over the set of \textit{maximally coherent stabilizer states} in the computational basis is at least as hard as the corresponding problem for graph states.
The set of maximally coherent stabilizer states in the computational basis is formally defined as
\begin{align}
    \mathcal{C}_{n} \coloneqq \left\{\bigotimes_{i=1}^{n} U_i\ket{G(A)} : \ket{G(A)} \in \mathcal{G}_{n}, \, U_i \in \{Z, S, SZ\}\right\},
\end{align}
meaning that each maximally coherent state is locally unitarily equivalent to a graph state under diagonal local unitaries. These unitaries \(U_i\) are restricted to the Pauli-\(Z\) basis and include the phase gate \(S\) and its composition with \(Z\).
The action of such a diagonal unitary \(U_i = \operatorname{diag}(1, \phi_i)\) with \(|\phi_i| = 1\) transforms a graph state as follows
\begin{align}
    \bigotimes_{i} U_i \ket{G(A)} = \frac{1}{\sqrt{d}} \sum_{x} (-1)^{\sum_{i,j} A_{i,j} x_i x_j} \prod_{i=1}^n \phi_i^{x_i} \ket{x}.
\end{align}
Here, each basis vector  \(\ket{x}\) acquires an additional phase factor \(\prod_i \phi_i^{x_i}\) that depends on the unitary operations applied.

To analyse the effect of these unitaries, consider the operator
\begin{align}
S(i) \coloneqq \frac{d^2}{4} (\ket{\mathbf 0} - \ket{i})(\bra{\mathbf 0} - \bra{i}),  
\end{align}
where \(\ket{i}\) denotes the basis state with a single \(1\) at position \(i\). The expectation value of \(S(i)\) over a transformed state \(\bigotimes_{i} U_i \ket{G(A)}\) evaluates to \({|1 - \phi_i|^2}/{4}\). Crucially, this expectation value attains its minimum value of \(0\) if and only if \(\phi_i = 1\), which corresponds to the case where no nontrivial diagonal unitary is applied to the \(i\)-th qubit. This property will be essential for our hardness reduction, as it allows us to detect and penalise deviations from the original graph state structure.
We, therefore, 
extend the Hamiltonian as
\begin{align}
    H^{(2)} = H^{(1)} + 2(\|H\|_{\infty} + 1) \sum_{i=1}^{2n} S(i).
\end{align}
By construction, we have \(\langle\sigma|H^{(2)}|\sigma\rangle = \langle\sigma|H^{(1)}|\sigma\rangle\) for all \(\ket{\sigma} \in \mathcal{G}_{2n}\), since \(\mathcal{G}_{2n} \subseteq \mathcal{C}_{2n}\) and the operators \(S(i)\) vanish when acting on states \(\mathcal{G}_{2n}\). We now demonstrate that optimising \(H^{(2)}\) over the set \(\mathcal{C}_{2n}\) is at least as hard as optimising \(H^{(1)}\) over \(\mathcal{G}_{2n}\), thereby establishing that the former problem requires time $2^{\Omega(n^2)}$.

\begin{itemize}
    \item \class{YES.} Assume that $\min_{\ket{\sigma}\in\mathcal{G}_{2n}}\langle\sigma|H^{(1)}|\sigma\rangle=0$; then, by the argument above, we also have $\min_{\ket{\sigma}\in\mathcal{C}_{2n}}\langle\sigma|H^{(2)}|\sigma\rangle=0$.
    \item \class{NO.} Assume that $\min_{\ket{\sigma}\in\mathcal{G}_{2n}}\langle\sigma|H^{(1)}|\sigma\rangle\ge 1$, then for any $\ket{\sigma}\in\mathcal{C}_{2n}\setminus\mathcal{G}_{2n}$, we have 
    \begin{align}
        \langle\sigma|H^{(2)}|\sigma\rangle&\ge 2(\|H\|_{\infty}+1)\frac{1}{2}+\langle\sigma|H|\sigma\rangle+\langle\sigma|(H^{(1)}-H)|\sigma\rangle\\&
        \ge (\|H\|_{\infty}+1)-|\langle\sigma|H|\sigma\rangle|\\&\ge (\|H\|_{\infty}+1)-\|H\|_{\infty}\\&= 1,
    \end{align}
where we have exploited the fact that $\phi\in\{\pm i,\pm i\}$ for diagonal Clifford operators. Moreover, we noticed that the operator $H^{(1)}-H$ is a positive operator. 
\end{itemize}

\paragraph{\bf Extension to stabilizer states with exponentially small overlap with $\ket{\mathbf 0}$.}
In the previous step, we showed that minimising \(H\) over the set of maximally coherent stabilizer states requires exponential time. We now define the following set:
\begin{align}
    \mathcal{T}_{n} \coloneqq \left\{\ket{\sigma} \in \mathcal{S}_n : |\langle \sigma | \mathbf 0 \rangle|^2 \geq \frac{1}{d}\right\},
\end{align}
which consists of stabilizer states whose overlap with the \(\ket{\mathbf 0}\) state is at least $1/d$. Since \(\mathcal{C}_{2n} \subseteq \mathcal{T}_{2n}\), we aim—as in previous steps—to penalise states outside this set, thereby preserving the hardness of our solution.
 Let us define the  Hamiltonian
\begin{align}
    H^{(3)}\coloneqq H^{(2)}+d^2(\|H\|_{\infty}+1)\left(\ketbra{\mathbf 0}{\mathbf 0}-\frac{I}{d^2}\right),
\end{align}
where $\ket{\mathbf 0}$ is defined on $2n$ qubits, and show that minimising $H^{(3)}$ over $\mathcal{T}_{2n}$ is at least as hard as optimising $H^{(2)}$ over $\mathcal{C}_{2n}$.
The strategy is the same; analysing the $\class{YES}$ and $\class{NO}$ instances, we have the following:
\begin{itemize}
    \item \class{YES}. Assume that $\min_{\ket{\sigma}\in\mathcal{C}_{2n}}\langle\sigma|H^{(2)}|\sigma\rangle=0$. Since the  $\operatorname{argmin}\ket{\sigma}$ is maximally coherent, it holds that $|\langle\sigma|\mathbf 0\rangle|^{2}=1/d^2$ and therefore $\langle\sigma|H^{(3)}|\sigma\rangle=0$. 
    \item \class{NO.} Assume that $\min_{\sigma\in\mathcal{C}_{2n}}\langle\sigma|H^{(2)}|\sigma\rangle\ge 1$. Then, for any $\ket{\sigma}\in \mathcal{T}_{2n}\setminus\mathcal{C}_{2n}$, we have 
\begin{align}
   \langle\sigma|H^{(3)}|\sigma\rangle&\ge d^2(\|H\|_{\infty}+1)\left(|\langle\sigma|\mathbf 0\rangle|^2-\frac{1}{d^2}\right)+\langle\sigma|H^{(3)}-H|\sigma\rangle-|\langle\sigma|H|\sigma\rangle|\\&\ge \frac{d^2(\|H\|_{\infty}+1)}{d^2}-\|H\|_{\infty}\ge 1.
\end{align}
Above, we have used the fact that if $\ket{\sigma}$ is not maximally coherent, then the overlap with $\ket{\mathbf 0}$ must be $\ge 2/d^2$. We also used the fact that, similarly to above, the operator $H^{(2)}-H$ is a positive operator.

Hence, optimising $H^{(3)}$ over the set $\mathcal{T}_{2n}$ is at least as hard as optimising $H^{(2)}$ on $\mathcal{C}_{2n}$ and therefore requires solving a $3$-\class{SAT} with $n(n-1)/2$ variables. 
\end{itemize}

\paragraph{\bf Extension to stabilizer states.}
We now demonstrate that determining the eigenenergy difference of a Hamiltonian over the set of stabilizer states \(\mathcal{S}_{2n}\) is at least as hard as optimising it over \(\mathcal{T}_{2n}\). The key distinction between these sets is that states in \(\mathcal{T}_{2n}\) maintain a non-vanishing overlap with \(\ket{\mathbf 0}\), whereas \(\mathcal{S}_{2n}\) includes states orthogonal to \(\ket{\mathbf 0}\). To establish this relationship, we consider the modified Hamiltonian
\begin{align}
    H^{(4)} = H^{(3)} + (\|H\|_{\infty}+\|H^{(3)}-H^{(2)}\|_{\infty} + 1)\left[I - d^2 \ketbra{\mathbf 0}{\mathbf 0}\right],
\end{align}
and define the following decision problem for $H^{(4)}$:
\begin{itemize}
    \item $\min_{\ket{\sigma}\in\mathcal{S}_{2n}}\langle\sigma|H|\sigma\rangle\le 0$ output \class{YES}.
    \item $\min_{\ket{\sigma}\in\mathcal{S}_{2n}}\langle\sigma|H|\sigma\rangle\ge 1$ output \class{NO}.
\end{itemize}
Let us show, in a similar fashion to what was done before, that solving the above decision problem for $H^{(4)}$ is at least as hard as optimising $H^{(3)}$ on $\mathcal{T}_{2n}$. 
\begin{itemize}
    \item \class{YES.} Assume $\min_{\ket{\sigma}\in\mathcal{T}_{2n}}\langle\sigma|H^{(3)}|\sigma\rangle=0$. Let $\ket{\sigma}$ be the $\operatorname{argmin}$. Then we have
    \begin{align}
    \langle\sigma| H^{(4)}|\sigma\rangle=(\|H\|_{\infty}+\|H^{(3)}-H^{(2)}\|_{\infty} + 1)(1-d^2|\langle\mathbf 0|\tau\rangle|^2)\le 0,
\end{align}
as an immediate and 
trivial consequence of the definition of $\mathcal{T}_{2n}$. 

\item \class{NO.} Assume $\min_{\ket{\sigma}\in\mathcal{T}_{2n}}\langle\sigma|H^{(3)}|\sigma\rangle\ge1$. For any $\ket{\sigma}\in\mathcal{S}_{2n}\setminus\mathcal{T}_{2n}$ we have 
\begin{align}
    \langle\sigma|H^{(4)}|\sigma\rangle\ge (\|H\|_{\infty}+\|H^{(3)}-H^{(2)}\|_{\infty} + 1)+\langle \sigma|H^{(2)}-H|\sigma\rangle-|\langle\sigma|H+(H^{(3)}-H^{(2)})|\sigma\rangle|\ge 1,
\end{align}
which shows the claim. Notice that we decomposed $H^{(3)}=H+(H^{(2)}-H)+(H^{(3)}-H^{(2)})$ and noticed that $(H^{(2)}-H)\ge 0$.
\end{itemize}
\smallskip

\paragraph{\bf Reduction to $\class{WWD}_{\delta}$.}
 To complete the proof, we reduce our problem to an instance of \(\class{WWD}\) as defined in \cref{problem:WWD}. Let \(\hat{\mathcal{S}}_n\) denote the convex hull of stabilizer states \(\mathcal{S}_n\), with $S(\hat{\mathcal{S}}_{2n}, -\delta)$ and \(S(\hat{\mathcal{S}}_{2n}, \delta)\) representing its inner and outer approximations, respectively. Our \(\class{WWD}\) instance involves a Hermitian operator \(W\) with \(\|W\|_2 \leq 1\), where the decision problem is the following:
\begin{itemize}
    \item If there exists $\rho\in S(\hat{\mathcal{S}}_{2n}, -\delta)$ s.t. $\tr(W\rho)\ge \gamma+\delta$, then output \class{YES}.
    \item If $\forall \rho\in S(\hat{\mathcal{S}}_{2n}, \delta)$ holds that $\tr(W\rho)\le \gamma-\delta$, then output \class{NO}.
\end{itemize}
We now establish the connection to our Hamiltonian problem by defining 
\begin{equation}
W := -\frac{H^{(4)}}{\|H^{(4)}\|_2}.
\end{equation}
For the \(\class{YES}\) case, when \(\max_{\sigma \in \hat{\mathcal{S}}_{2n}} \tr(W\sigma) \geq 0\), there exists \(\rho \in S(\hat{\mathcal{S}}_{2n}, -\delta)\) satisfying \(\tr(W\rho) \geq \gamma + \delta\). Conversely, if \(\max_{\sigma \in \hat{\mathcal{S}}_{2n}} \tr(W\sigma) \leq -\frac{1}{\|H^{(4)}\|_2}\), then all \(\rho \in S(\hat{\mathcal{S}}_{2n},- \delta)\) satisfy \(\tr(W\rho) \leq \gamma - \delta\). Crucially, optimisation over the convex hull \(\hat{\mathcal{S}}_{2n}\) reduces to optimisation over its extreme points \(\mathcal{S}_{2n}\), meaning \(\max_{\sigma \in \hat{\mathcal{S}}_{2n}} \tr(W\sigma) = \max_{\sigma \in \mathcal{S}_{2n}} \langle \sigma | W | \sigma \rangle\). This reduction demonstrates that the Hamiltonian decision problem maps directly to \(\class{WWD}_\delta(\hat{\mathcal{S}}_n)\), thereby establishing its hardness.
To show it, set $\gamma
:=-\frac{1}{2\|H^{(4)}\|_2}$.
\begin{itemize}
    \item \class{YES.} Assume $\max_{\sigma\in\mathcal{S}_{2n}}\tr(W\sigma)\ge 0$. Since $\hat{\mathcal{S}}_{2n}$ is a convex set, there exists $\sigma_{\varepsilon}\in S(\hat{\mathcal S}_{2n},-\delta)$ such that $\|\sigma-\sigma_{\delta}\|_{2}\le 2\delta R/r$~\cite{bengtsson2006geometry}. We recall that $R\le 1$, $r\ge 1/d^2$ , because $H^{(4)}$ is defined on $2n$ qubits. This means that
    \begin{align}
        |\tr(W\sigma_{\delta})-\tr(W\sigma)|\le \|\sigma_{\delta}-\sigma\|_2\le \delta d^2,
    \end{align}
    implying that $\tr(W\sigma_{\delta})\ge \tr(W\sigma)-\delta d^2\ge -\delta d^2$, using that $\tr(W\sigma)\ge 0$. Imposing the latter, being $\ge\gamma+\delta$, we find $\delta\le \frac{1}{(1+d^2)\|H^{(4)}\|_2}$.

    \item \class{NO.} Assume that $\max_{\sigma\in\mathcal{S}_{2n}}\tr(W\sigma)\le -\frac{1}{\|H^{(4)}\|_2}$. By the definition of $S(\hat{\mathcal S}_{2n},\delta)$, for any $\rho\in S(\hat{\mathcal S}_{2n},\delta)$, there exists $\sigma\in\mathcal{S}_{2n}$ such that $\|\rho-\sigma\|_2\le \delta$, implying that $|\tr(W\rho)-\tr(W\sigma)|\le \delta$. 
    It 
    follows that, for any $\rho\in S(\hat{\mathcal{S}}_{2n},\delta)$
    \begin{align}
        \tr(W\rho)\le \tr(W\sigma)+\delta\le \max_{\sigma\in\mathcal{S}_{2n}}\langle\sigma|W|\sigma\rangle+\delta\le -\frac{1}{\|H^{(4)}\|_2}+\delta=2\gamma+\delta.
    \end{align}
Imposing the latter $\le \gamma-\delta$, we find $\delta\le -\gamma/2=\frac{1}{4\|H^{(4)}\|_2}$.
\end{itemize}
To accommodate both \class{YES} and \class{NO} instances, we take $\delta\le \frac{1}{4d^2\|H^{(4)}\|_2}$. Therefore, we have shown that with oracle access to $\class{WWD}_{\delta}(\hat{\mathcal{S}}_{2n})$ with parameters $\left(W\equiv H^{(4)}/\|H^{(4)}\|_{2},\gamma\equiv-\frac{1}{2\|H^{(2)}\|_2}\right)$ and any $\delta\le \frac{1}{4d^2\|H^{(4)}\|_2}$, it is possible to solve 
the optimisation of a Hamiltonian over the set of pure stabilizer states $\mathcal{S}_{2n}$, which we have demonstrated to be at least as hard as a $3$-\class{SAT} with input size $n(n-1)/2=\Omega(\log^2d)$. We are just left to find an upper bound to $\|H^{(4)}\|_2$, which 
is what we will turn to now.\smallskip
\paragraph{\bf Upper bound on $\|H^{(4)}\|_2$.} We then have the 
following chain of inequalities

\begin{equation}
    \begin{aligned}
        \|H^{(4)}\|_{2}& \le d\|H^{(4)}\|_{\infty} \\ 
        & \le d\qty(\|H^{(3)}\|_{\infty} + d^2(\|H\|_{\infty}+\|H^{(3)}-H^{(2)}\|_{\infty} + 1))\\
        & \le d\qty[\|H^{(3)}\|_{\infty} + d^2(d^2+1)(\|H\|_{\infty}+ 1)] \\
        & \le d\qty[\|H^{(2)}\|_{\infty} + d^2(d^2+2)(\|H\|_{\infty}+ 1) ] \\
        & \le d\qty[\|H^{(1)}\|_{\infty}+2d^2(\|H\|_{\infty}+1)\log d + d^2(d^2+2)(\|H\|_{\infty}+ 1)] \\
        & \le d\qty[\|H\|_{\infty}+ 2d^2 (\|H\|_{\infty}+1)\log^{2}d  + d^2(d^2+4\log d )(\|H\|_{\infty}+ 1)] \\
        & \le  d(\|H\|_{\infty}+1) (d^2 (d^2+6\log^2 d)) \le 2^4 d^{7}\log^6 d \,,
    \end{aligned}
\end{equation}
where we have used the definitions of $H^{(4)}, H^{(3)}, H^{(2)}, H^{(1)}$ and $H$, and the triangle inequality. Consequently, we obtain that $\|H^{(4)}\|_2\le 2^4 d^{7}\log^6 d $. Therefore, any $\delta\le \frac{1}{2^6d^{9}\log^6 d}$ suffices to establish the computational hardness of $\class{WWD}_{\delta}(\hat{\mathcal{S}}_{2n})$. To conclude, we are just left to map $d^2\mapsto d$. 
\end{proof}

Moreover, it is simple to see that $\class{WWD}(\hat{\mathcal{S}}_n)$ is contained in the complexity class $\class{QP}^2$, as the next lemma points out.

\begin{lemma}[$\class{WWD}_{\delta}(\hat{\mathcal{S}}_n)\in\class{QP}^2$] \label{wwdinqp2}The decision problem $\class{WWD}_{\delta}(\hat{\mathcal{S}}_n)$ belongs to $\class{QP}^2$.
\begin{proof}
    Given an operator $W$ with the threshold parameter $\gamma$ and the precision $\delta$, to solve $\class{WWD}_{\delta}(\hat{\mathcal{S}}_n)$ is sufficient to check the value of $\tr(W\sigma)$ for any $\sigma\in\mathcal{S}_n$ (see \cref{problem:WWD}) up to precision $\delta$. This algorithm takes $\class{poly}(2^{\log^2d},\log\delta^{-1})$. Therefore for any $\delta$ reasonably large the problem can be solved within the complexity class $\class{QP}^2$.
\end{proof}
    
\end{lemma}

\subsection{Weak membership in the convex hull of $t$-doped stabilizer states}
In this section, we extend the results obtained regarding the membership of stabilizer states to the more general notion of $t$-doped stabilizer states. Let us first start by examining the hardness of the decision problem $\rho$-$\class{WMEM}_{\varepsilon}(\hat{\mathcal{S}}_{n,t})$. To do so, we first prove that determining whether a Hermitian operator $W$ functions as a witness for the $t$-doped stabilizer polytope does not belong to $\class{QP}^{2-\eta}$ for any $\eta>0$ under the ETH assumption (see~\cref{def:eth}), and for any $t\le \kappa n$ with $\kappa<1$.
\begin{lemma}[Superpolynomial complexity of $\class{WWD}_{\delta}(\hat{\mathcal{S}}_{n,t})$]\label{wwdtdoped}
Let $\hat{\mathcal{S}}_{n,t}$ denote the $t$-doped stabilizer polytope on $n$ qubits and let $t\le \kappa n$ with $\kappa<1$. Under the ETH (Definition~\ref{def:eth}), for any
\begin{align}
 \delta \le \frac{1}{ 2 d^{31/2} \log^{30}d},
\end{align}
the decision problem $\class{WWD}_{\delta}(\hat{\mathcal{S}}_{n,t})$ is not contained in the complexity class $\class{QP}^{2-\eta}$ for any $\eta>0$.
\end{lemma}

\begin{proof}
The overall strategy is to reduce the problem of optimising a Hamiltonian over the $t$-doped stabilizer polytope $\hat{\mathcal{S}}_{n,t}$ to an instance of weak-witness detection. The reduction allows us to embed a decision problem that, under the ETH, requires time $2^{\Omega(\log^2d)}$ into $\class{WWD}_{\delta}(\hat{\mathcal{S}}_{n,t})$ while preserving a non-vanishing promise gap. Since the input is a $d\times d$ density matrix with $d=2^n$, this implies that the problem cannot lie in $\class{QP}^{2-\eta}$ for any $\eta>0$, see Definition~\ref{def:quasipolyclasses}.

By \cref{thm:nphardwwd}, there exists a family of Hamiltonians $H$ acting on $n-t$ qubits such that the following promise problem is at least as hard as solving a $3$-\class{SAT} instance on $\Omega(\log^22^{n-t})$ variables:
\begin{itemize}
\item \class{YES}. $\max_{\sigma \in \hat{\mathcal{S}}_{n-t}} \tr(H\sigma) \ge 1$.
\item \class{NO}. $\max_{\sigma \in \hat{\mathcal{S}}_{n-t}} \tr(H\sigma) \le 0$.
\end{itemize}
To symmetrise the promise around a tunable threshold, define the shifted Hamiltonian
\begin{align}
H' \coloneqq \alpha I - H,
\end{align}
for a parameter $\alpha>0$ to be fixed later. The problem becomes
the following.
\begin{itemize}
\item \class{YES}. $\max_{\sigma}\tr(H'\sigma)\ge \alpha$.
\item \class{NO}. $\max_{\sigma}\tr(H'\sigma)\le \alpha-1$.
\end{itemize}

\paragraph{\bf Embedding into $\hat{\mathcal{S}}_{n,t}$.}
Recall that every state vector in $\hat{\mathcal{S}}_{n,t}$ admits a decomposition of the form (\cref{def:tcomp}):
\begin{align}
\ket{\psi} = C\bigl(\ket{\phi}\otimes\ket{0}^{\otimes (n-t)}\bigr),
\end{align}
where $\ket{\phi}$ is a $t$-qubit state vector and $C$ is a Clifford unitary.
Let $\ket{\phi}$ 
be a $t$-qubit state vector to be fixed later, and define the lifted Hamiltonian acting on $n$ qubits as
\begin{align}
H''_n \coloneqq \ketbra{\phi}{\phi} \otimes H'_{n-t}.
\end{align}
By construction, $H''_n$ acts as $H'_{n-t}$ on the subspace where the first $t$ qubits are aligned with $\ket{\phi}$ and suppresses contributions from states with small overlap on this subsystem. Let us show that optimising $H''_n$ on $\mathcal{S}_{n,t}$ is at least as hard as optimising $H'_{n-t}$ on $\mathcal{S}_{n-t}$ thus proving the claim for any $(n-t)=\Omega(n)$.
\begin{itemize}
\item \class{YES}. Assume there exists a state vector $\ket{\sigma}_{n-t}$ such that
\begin{align}
\tr(H'_{n-t}\ketbra{\sigma}{\sigma}) \ge \alpha.
\end{align}
Then the product state vector $\ket{\phi}\otimes\ket{\sigma}$ satisfies
\begin{align}
\tr\bigl(H''_n(\ketbra{\phi}{\phi}\otimes\ketbra{\sigma}{\sigma})\bigr) \ge \alpha.
\end{align}
We note that the $n$-qubit state $\ket{\phi}\otimes \ket{\sigma}$ lies in $\hat{\mathcal{S}}_{n,t}$. Hence, the lifted instance is a \class{YES} instance.

\item \class{NO}. Assume instead that
\begin{align}\label{conditiononsigmanmenot}
\max_{\ket{\sigma_{n-t}}\in\mathcal{S}_{n-t}}\tr(H'_{n-t}\ketbra{\sigma}{\sigma}) \le \alpha-1.
\end{align}
Let $\ket{\psi}\in\hat{\mathcal{S}}_{n,t}$ be an arbitrary $t$-doped state vector. For this, we define the overlap
\begin{align}
F \coloneqq \tr\bigl(\ketbra{\psi}{\psi}(\ketbra{\phi}{\phi}\otimes I)\bigr).
\end{align}

If $F \ge 1-\vartheta$, then projecting $\ket{\psi}$ onto the $\ket{\phi}$ subspace produces a normalised $(n-t)$-qubit state vector $\ket{\tilde{\sigma}}$ such that, by the gentle measurement lemma~\cite{MarkWilde}, we have that
\begin{align}
\|\ketbra{\psi}{\psi}-\ketbra{\phi}{\phi}\otimes\ketbra{\tilde{\sigma}}{\tilde{\sigma}}\|_1 \le 2\sqrt{\vartheta}.
\end{align}
We now aim to show \(\ket{\tilde{\sigma}}\) to be a stabilizer state (see \cref{eq:stabentropydef}). Let $P_{6}$ be the stabilizer purity introduced in \cref{stabentstabfid}. We then choose \(\ket{\phi} \coloneqq \operatorname{argmin}_{\tau} P_{6}(\ket{\tau})\), which minimises the stabilizer purity over \(t\)-qubit states. However, by construction of $t$-doped states (\cref{def:tcomp}), for any $\ket{\psi}\in\mathcal{S}_{n,t}$ we have that $P_{6}(\ket{\phi})\le P_{6}(\ket{\psi})$. Let us note that we can bound
\begin{align}
     \|\ketbra{\psi}{\psi}^{\otimes 6} - (\ketbra{\phi}{\phi}\otimes \ketbra{\tilde\sigma}{\tilde\sigma})^{\otimes 6}\|_1
     \le 6\,\|\ketbra{\psi}{\psi}-\ketbra{\phi}{\phi}\otimes \ketbra{\tilde{\sigma}}{\tilde{\sigma}}\|_1
     \le 12\sqrt{\vartheta}.
\end{align}
Then, we find that
\begin{align}
    P_{6}(\ket{\phi})&\le P_{6}(\ket{\psi})
    \le P_{6}(\ket{\phi})P_6({\ket{\tilde\sigma}})+12\sqrt{\vartheta}\left\|\frac{1}{d}\sum_{P} P^{\otimes 6}\right\|_{\infty} \\ 
    \nonumber
    &\le P_{6}(\ket{\phi})P_6({\ket{\tilde\sigma}})+12\sqrt{\vartheta}.
\end{align}

Thus we obtain the lower bound
\begin{align}
    P_{6}(\ket{\tilde{\sigma}})\ge 1-\frac{12\sqrt{\vartheta}}{P_{6}(\ket{\phi})}.
\end{align}
Let us now use the relation between stabilizer purity and stabilizer fidelity $F_{\text{stab}}$ in \cref{stabentstabfid} to lower bound the stabilizer fidelity of $\ket{\tilde{\sigma}}$ as
\begin{align}
    F_{\text{stab}}(\ket{\tilde{\sigma}})
    \ge \left(1-\frac{12\sqrt{\vartheta}}{P_{6}(\ket{\phi})}\right)^{C}
    \ge 1-\frac{12C\sqrt{\vartheta}}{P_{6}(\ket{\phi})}
    \ge 1-\frac{12C\sqrt{\vartheta}}{2^{t}}.
\end{align}
where we have upper bounded $P_{6}(\ket{\phi})\le 2^t$ being $\ket{\phi}$ a $t$-qubit state. Hence, we deduce that there exists a stabilizer state vector $\ket{\sigma}$ on $n-t$ qubits such that
\begin{align}
    \|\ketbra{\sigma}-\ketbra{\tilde{\sigma}}\|_1=2\sqrt{1-F_{\text{stab}}(\ket{\tilde{\sigma}})}\le 2\sqrt{12C2^{-t}\sqrt{\theta}}\,,
\end{align}
from which we have
\begin{align}
    \tr(\ketbra{\tilde{\sigma}}H_{n-t}')\le \alpha-1+2\sqrt{12C2^{-t}\sqrt{\theta}}\|H_{n-t}\|_{\infty}
\end{align}
where we have used the condition in \cref{conditiononsigmanmenot}. Using $F\ge 1-\theta$,
we derive
\begin{align}
    \tr(\ketbra{\psi}{\psi} (\ketbra{\phi}{\phi}\otimes H'_{n-t}))
    &\le \alpha - \tr(\ketbra{\psi}{\psi} (\ketbra{\phi}{\phi}\otimes H_{n-t})) \\ 
    \nonumber
    & \le \alpha - \tr(\ketbra{\phi}{\phi} \otimes \ketbra{\tilde \sigma}{\tilde \sigma} (\ketbra{\phi}{\phi}\otimes H_{n-t})) + 2\sqrt{\vartheta}\|H_{n-t}\|_{\infty} 
    \\
     \nonumber
    &\le \alpha-\tr(\ketbra{\tilde{\sigma}}H_{n-t})+2\sqrt{\theta}\|H_{n-t}\|_{\infty}\\
     \nonumber
    &\le \tr(\ketbra{\tilde{\sigma}}H'_{n-t})+2\sqrt{\theta}\|H_{n-t}\|_{\infty}\\
     \nonumber
    & \le  \alpha-1 + \qty(2\sqrt{\theta}+2\sqrt{12C2^{-t}\sqrt{\theta}})\|H_{n-t}\|_{\infty}. 
     \nonumber
\end{align}
In summary, this yields
\begin{align}
\tr(\ketbra{\psi}{\psi}H''_n)
\le  \alpha-1 + \qty(2\sqrt{\theta}+2\sqrt{12C2^{-t}\sqrt{\theta}})\|H_{n-t}\|_{\infty} .
\end{align}

In the converse case, if $F < 1-\vartheta$, then $\ket{\psi}$ has significant support outside the $\ket{\phi}$ subspace, and by definition of $H''_n$,
\begin{align}
    \tr(\ketbra{\psi}(\ketbra{\phi}\otimes H'_{n-t}))
    \le \|\bra{\phi}\psi\ket{\phi}\|_1\|H'_{n-t}\|_{\infty}
    =\|\bra{\phi}\psi\ket{\phi}\|_2\|H'_{n-t}\|_{\infty}
    <(1-\theta)\|H'_{n-t}\|_{\infty}\,,
\end{align}
where we have used the fact that $\|\langle\phi|\psi|\phi\rangle\|_2=\sqrt{\tr[(\ketbra{\phi}\otimes I)\ketbra{\psi}(\ketbra{\phi}\otimes I)\ketbra{\psi}]}=F$.
We now impose for both cases, i.e., $F\ge 1-\theta$ and $F<1-\theta$, that $\tr(\ketbra{\psi}H''_{n-t})\le \alpha-\frac{1}{2}$ to separate \class{YES} and \class{NO} instance. For the latter case, we simply impose
\begin{align}
    (1-\theta)\|H_{n-t}'\|_{\infty}\le \alpha-\frac{1}{2} .
\end{align}
Crudely upper bounding $\|H'_{n-t}\|_{\infty}=\alpha+\|H_{n-t}\|_{\infty}$, we find a condition on the value of $\alpha$ as $\alpha\ge \theta^{-1}[\frac{1}{2}+\|H_{n-t}\|_{\infty}(1-\theta)]$. In the former case ($F\ge 1-\theta$), we impose
\begin{align}
 \alpha-1 + \qty(2\sqrt{\theta}+2\sqrt{12C2^{-t}\sqrt{\theta}})\|H_{n-t}\|_{\infty} \le \alpha-\frac{1}{2}   
\end{align}
from which we get the condition for $\theta$ given by 
\begin{align} \theta^{1/4}\le  \frac{1}{2(2+2\sqrt{12C2^{-t}})\|H_{n-t}\|_{\infty}}  .
\end{align}
By choosing $\alpha$ and $\vartheta$ such that both bounds lie strictly below $\alpha-\tfrac12$, no state in $\hat{\mathcal{S}}_{n,t}$ can achieve expectation value at least $\alpha$ in the \class{NO} case. Combining the two, we get
\begin{align}
    \alpha & \ge (1/2+\|H_{n-t}\|_{\infty})\|H_{n-t}\|_{\infty}^{4}[2(2+2\sqrt{12C2^{-t}})]^4\\
    \nonumber
    &\ge 2^{-10}d^{15} \log^{30}d  .
\end{align}

\end{itemize}
\smallskip
\paragraph{\bf Reduction to $\class{WWD}_{\delta}$.}
Define the witness
\begin{align}
W \coloneqq -\frac{H''_n}{\|H''_n\|_2}, \qquad \gamma \coloneqq -\frac{1}{2\|H''_n\|_2}.
\end{align}
Since optimisation over the convex set $\hat{\mathcal{S}}_{n,t}$ is attained on its extreme points, the Hamiltonian decision problem reduces to an instance of $\class{WWD}_{\delta}(\hat{\mathcal{S}}_{n,t})$ with parameters $(W,\gamma)$, provided $\delta$ is chosen sufficiently small.
Using submultiplicativity and the definition of $H'$, we obtain
\begin{align}
\|H''_n\|_2
&\le \|H'_{n-t}\|_2 \\
\nonumber
&\le \alpha \sqrt{d} 2^{-t/2} + \|H_{n-t}\|_2 \\
&\le \frac{1}{2}d^{31/2}\log^{30}d,
\nonumber
\end{align}
where we have set $\alpha=2^{-10}d^{15}\log^{30}d$, and the last inequality follows from the norm bounds established in \cref{thm:nphardwwd}.
Combining this with the conditions in the \class{YES} and \class{NO} cases of the weak-witness formulation yields the constraint
\begin{align}
\delta \le \frac{1}{ 2 d^{31/2} \log^{30}d}.
\end{align}
With this choice of $\delta$, oracle access to $\class{WWD}_{\delta}(\hat{\mathcal{S}}_{n,t})$ suffices to solve a Hamiltonian optimisation problem on $n-t\ge (1-\kappa)n$ qubits. Therefore, under the ETH, the problem $\class{WWD}_{\delta}(\hat{\mathcal{S}}_{n,t})$ cannot belong to $\class{QP}^{2-\eta}$ for any $\eta>0$.
\end{proof}

\begin{theorem}[Superpolynomial complexity of $\rho$-$\class{WMEM}_{\varepsilon}(\hat{\mathcal{S}}_{n,t})$]\label{thm:suptdop}
Let $\hat{\mathcal{S}}_{n,t}$ be the $t$-doped stabilizer polytope on $n$-qubits and $t\le\kappa n$ with $\kappa<1$. Under the ETH (\cref{def:eth}), and for any $\varepsilon \le (2^{16} 3^3   d^{119/2}\log^{90} d)^{-1}$, with $C$ the constant introduced in~\cref{stabentstabfid}, the decision problem of $\rho$-$\class{WMEM}_{\varepsilon}(\hat{\mathcal{S}}_{n,t})\notin \class{QP}^{2-\eta}$ for any $\eta>0$. 
\begin{proof}
To prove the result, we reduce the membership problem to the weak witness detection problem $\class{WWD}_{\delta}(\hat{\mathcal{S}}_{n,t})$. In~\cref{wwdtdoped} we show that $\class{WWD}_{\delta}(\hat{\mathcal{S}}_{n,t})\notin \class{QP}^{2-\eta}$ for any $\eta>0$, assuming the exponential time hypothesis stated in~\cref{def:eth}, provided 
\begin{equation}
\delta \le \frac{1}{ 2 d^{31/2} \log^{30}d}. 
\end{equation}
The reduction follows from~\cref{lem:reductionswmemwwd}. Consequently, as long as $\varepsilon \le (2^{16} 3^3   d^{119/2}\log^{90} d)^{-1}$, there exists a polynomial-time reduction between $\rho$-$\class{WMEM}_{\varepsilon}(\hat{\mathcal{S}}_{n,t})$ and $\class{WWD}_{\delta}\hat{\mathcal{S}}_{n,t})$. This concludes the proof.
\end{proof}
\end{theorem}

Before introducing the next lemma, we introduce two notions that are relevant to our analysis. The following definition extends the notion of robustness of magic by allowing, as free resources, convex combinations of $t$-doped states, as also done in Ref.~\cite{nakagawa2025applicationresourcetheorybased}. As in the case of stabilizer states (see \cref{def:robustnessofmagic}), this quantity serves a dual purpose: (i) it yields an algorithm to decide membership in the convex hull of $t$-doped states $\hat{\mathcal{S}}_{n,t}$, and (ii) it quantifies how far a given state lies outside $\hat{\mathcal{S}}_{n,t}$.

\begin{definition}[$t$-extended robustness of magic]
Given an $n$-qubit density matrix $\rho$, the $t$-extended robustness of magic is defined as
\begin{align}
    \mathcal{R}_t(\rho)\coloneqq\min_{\{x_i,\psi_i\}}\left\{\sum_i |x_i| \;:\; \rho=\sum_i x_i \psi_i,\ \psi_i\in\mathcal{S}_{n,t}\right\}.
\end{align}
Moreover, $\mathcal{R}_t(\rho)=0$ if and only if $\rho\in\hat{\mathcal{S}}_{n,t}$.
\end{definition}

We now establish a dual characterisation of the $t$-extended robustness and derive a uniform upper bound on its value.

\begin{lemma}[Dual formulation of $t$-extended robustness]\label{lemmatextendedrobusteness}
The $t$-extended robustness admits the dual formulation
\begin{align}
    \mathcal{R}_t(\rho)
    =\max_{\substack{W\\ \tr(W\psi)\le 1\ \forall\,\psi\in\mathcal{S}_{n,t}}}\tr(W\rho).
    \label{eq:optimisationextendedrobustness}
\end{align}
Moreover, for any $n$-qubit state $\rho$, it holds that $\mathcal{R}_t(\rho)\le \sqrt{d(d+1)}$, and any feasible witness $W$ in Eq.~\eqref{eq:optimisationextendedrobustness} satisfies $\|W\|_2\le \sqrt{d(d+1)}$.
\begin{proof}
The dual formulation~\eqref{eq:optimisationextendedrobustness} follows from the standard dual program associated with $\mathcal{R}_t$, in complete analogy with the robustness of magic (\cref{def:robustnessofmagic}). To upper bound both the $t$-extended robustness and the Hilbert–Schmidt norm of feasible witnesses, we proceed as in Ref.~\cite{liuManybodyQuantumMagic2022}.
First, we observe that
\begin{align}
    \mathcal{R}_t(\rho)
    \le \max_{\substack{W\\ \tr(W\psi)\le 1\ \forall\,\psi\in\mathcal{S}_{n,t}}}\|W\|_{\infty}
    \le \max_{\substack{W\\ \tr(W\psi)\le 1\ \forall\,\psi\in\mathcal{S}_{n,t}}}\|W\|_2.
\end{align}
Expanding $W$ in the Pauli basis, its Hilbert–Schmidt norm can be written as
\begin{align}
    \|W\|_2=\sqrt{\frac{1}{d}\sum_{P}\tr^2(WP)}.
\end{align}
The constraint $\tr(W\psi)\le 1$ holds in particular for all stabilizer states $\sigma\in\mathcal{S}_n$. For such states, $\tr(\sigma P)\in\{\pm1\}$ for all Pauli operators $P$, and hence
\begin{align}
    \tr(W\sigma)=\frac{1}{d}\sum_{P}\tr(WP)\tr(\sigma P)\le 1.
\end{align}
It follows that imposing the constraint only on stabilizer states is already sufficient to bound $\mathcal{R}_t$, so the remainder of the proof coincides with the argument in Ref.~\cite{liuManybodyQuantumMagic2022}. We include it here for completeness.

Fix an unsigned stabilizer group $G$. Define the matrix $V_{P,\sigma}=\frac{1}{\sqrt{d}}\tr(\sigma P)$, where $P\in G$ and $\sigma$ ranges over the $d$ stabilizer states with stabilizer group $G$. This matrix is unitary: it is square since $|G|=d$, and its rows and columns are orthonormal. Indeed,
\begin{align}
    \sum_{\sigma}V_{P,\sigma}V_{Q,\sigma}
    &=\frac{1}{d}\sum_{\sigma}\tr(\sigma P)\tr(\sigma Q)
    =\frac{1}{d}\sum_{\sigma}\tr(\sigma PQ)
    =\delta_{P,Q},
\end{align}
and similarly,
\begin{align}
    \sum_{P}V_{P,\sigma}V_{P,\sigma'}
    =\frac{1}{d}\sum_{P}\tr(\sigma P)\tr(\sigma' P)
    =\tr(\sigma\sigma')
    =\delta_{\sigma,\sigma'},
\end{align}
where we have used the fact that two stabilizer states associated with the same stabilizer group are either identical or orthogonal.
Using this unitary change of basis, we obtain
\begin{align}
    \sum_{P\in G}\tr^2(WP)
    &=\sum_{\sigma}\left(\frac{1}{\sqrt{d}}\sum_{P\in G}\tr(P\sigma)\tr(WP)\right)^2
    =d\sum_{\sigma}\tr^2(W\sigma)
    \le d^2.
\end{align}
Finally, since there exist $d+1$ stabilizer groups $\{G_j\}_{j=1}^{d+1}$ that intersect only at the identity, we can write
\begin{align}
    \sum_{P}\tr^2(WP)
    \le \sum_{j=1}^{d+1}\sum_{P\in G_j}\tr^2(WP)
    \le (d+1)d^2.
\end{align}
Therefore, $\|W\|_2\le \sqrt{d(d+1)}$, which concludes the proof.
\end{proof}
\end{lemma}
Given the definition of $t$-extended robustness, it is possible to prove the following result.

\begin{lemma}[$\rho$-$\class{WMEM}_{\varepsilon}(\hat{\mathcal{S}}_{n,t})\in \class{QP}$ for any $t=O(\log \log d)$]\label{lmm:wmemtdoped}
For any $0<\varepsilon<1$, there exists an algorithm running in time $\class{poly}(2^{\log^2 d+2^{2t}(\log d+\log\varepsilon^{-1}) })$ that, by computing the $t$-extended robustness of magic up to additive precision $\nu\le\frac{\varepsilon^6}{2^{26}\cdot 3^3\cdot d^{24}}$, solves $\rho$-$\class{WMEM}_{\varepsilon}(\hat{\mathcal{S}}_{n,t})$. In particular, the problem lies in $\class{QP}^2$ for any $t<\log \log d$, and in $\class{QP}$ for any $t=O(\log \log d)$.
\begin{proof}
The proof employs a strategy analogous to that presented in \cref{wmemeforstabinqp2}, substituting the robustness of magic with its $t$-extended counterpart. The primary challenge arises from the continuity of the set $\mathcal{S}_{n,t}$, which would otherwise result in an infinite runtime for the algorithm. To address this, we construct a $\vartheta$-packing net $\mathcal{S}_{n,t}^{\vartheta}$ of $\mathcal{S}_{n,t}$ such that $\mathcal{S}_{n,t}^{\vartheta} \subseteq \mathcal{S}_{n,t}$. A straightforward method for constructing this ensemble involves first creating a $\vartheta$-packing net of the set of $t$-qubit states and subsequently applying rotations using the Clifford group. By \cref{def:tcomp}, this ensures that $\mathcal{S}_{n,t}^{\vartheta} \subseteq \mathcal{S}_{n,t}$.
Let $\hat{\mathcal{S}}_{n,t}^{\vartheta}$ denote the convex hull of $\mathcal{S}_{n,t}^{\vartheta}$. For any state $\rho \in \hat{\mathcal{S}}_{n,t}$, there exists a state $\rho^{\vartheta} \in \hat{\mathcal{S}}_{n,t}^{\vartheta}$ such that $\|\rho - \rho^{\vartheta}\|_2 \leq \vartheta$.

With the ensemble of states thus defined, we proceed to prove the lemma by demonstrating that the $t$-extended robustness, when evaluated on the $\vartheta$-packing net, is sufficient to solve the one-sided problem $\rho$-$\class{WMEM}^{1}_{\delta}(\hat{\mathcal{S}}_{n,t})$.
The $t$-extended robustness admits the dual formulation~\cref{lemmatextendedrobusteness}
\begin{align}
\mathcal{R}_t(\rho) = \max_{\substack{W \\ \tr(W\psi) \le 1 \ \forall \psi \in \mathcal{S}_{n,t}}} \tr(W\rho),
\end{align}
where the feasible set of witnesses is given by
\begin{align}
\mathcal{W}_t \coloneqq \left\{ W : \max_{\psi \in \mathcal{S}_{n,t}} \tr(W\psi) \le 1 \right\}.
\end{align}
To adapt this definition to $\mathcal{S}_{n,t}^{\vartheta}$, we obtain
\begin{align}
\mathcal{R}_t^{\vartheta}(\rho) = \max_{\substack{W \\ \tr(W\psi) \le 1 \ \forall \psi \in \mathcal{S}_{n,t}^{\vartheta}}} \tr(W\rho),
\end{align}
with the corresponding feasible set defined as
\begin{align}
\mathcal{W}_t^{\vartheta} \coloneqq \left\{ W : \max_{\psi \in \mathcal{S}_{n,t}^{\vartheta}} \tr(W\psi) \le 1 \right\}.
\end{align}
We begin by considering an idealised setting in which the entire witness set $\mathcal{W}_t^{\vartheta}$ is accessible exactly. Subsequently, we introduce robustness with respect to finite precision by defining the $\nu$-robust witness set. Recall from Definition~\ref{def:set1} that for a set $K$, the inner core 
\begin{align}
S(K, -\nu) = \{ W \in K : S(\{W\}, \nu) \subseteq K \}
\end{align}
consists of all points whose $\nu$-ball is contained in $K$
where $S(\{W\}, \nu) = \{ W' : \|W' - W\|_2 \le \nu \}$.
For $K = \mathcal{W}_t^{\vartheta}$, we have
\begin{align}
S(\mathcal{W}_t^{\vartheta}, -\nu) = \{ W \in \mathcal{W}_t^{\vartheta} : \forall W' \text{ with } \|W' - W\|_2 \le \nu, \ W' \in \mathcal{W}_t^{\vartheta} \}.
\end{align}
For the reminder of the proof, we adopt the notation $\mathcal{W}_{t,\nu}^{\vartheta} \equiv S(\mathcal{W}_t^{\vartheta}, -\nu)$. 
The associated optimisation problem
\begin{align}
    R_{t,\nu}^{\vartheta}(\rho)\coloneqq\max_{W\in\mathcal{W}_{t,\nu}^{\vartheta}}\tr(W\rho)
\end{align}
is a linear program that can be solved up to additive error $\nu$ using the ellipsoid algorithm~\cite{grotschelGeometricAlgorithmsCombinatorial1988}, provided that $\mathcal{W}_{t,\nu}^{\vartheta}$ is nonempty. We now show that computing $R_{t,\nu}^{\vartheta}(\rho)$ suffices to solve the one-sided problem $\rho$-$\class{WMEM}^{1}_{\delta}(\hat{\mathcal{S}}_{n,t})$, from which the full weak membership problem follows by \cref{lem:onesideequivalentbothss}.
\begin{itemize}
\item \class{YES}. Let $\rho\in\hat{\mathcal{S}}_{n,t}$. By definition of the $\vartheta$-packing net, there exists $\rho^{\vartheta}\in\hat{\mathcal{S}}_{n,t}^{\vartheta}$ such that $\|\rho-\rho^{\vartheta}\|_2\le \vartheta$. For any witness $W$, we then have
\begin{align}
    \tr(W\rho)
    &\le \tr(W\rho^{\vartheta}) + \vartheta\|W\|_2 + \nu \\
    &\le 1+ \vartheta \sqrt{d(d+1)} + \nu,\nonumber
\end{align}
where we have used $\|W\|_{2}\le \sqrt{d(d+1)}$ (\cref{lemmatextendedrobusteness}) and the fact that the optimisation is solved up to $\nu$-additive error. Hence, no \class{YES} instance can be misclassified as \class{NO}, provided that \class{NO} instances yield values strictly larger than $1+\nu+\vartheta\sqrt{d(d+1)}$.

\item \class{NO}. Let $\rho$ satisfy $\min_{\sigma\in\hat{\mathcal{S}}_{n,t}}\|\rho-\sigma\|_2\ge\delta$. Then
\begin{align}
\min_{\sigma\in\hat{\mathcal{S}}_{n,t}^{\vartheta}}\|\rho-\sigma\|_2
\ge \min_{\sigma\in\hat{\mathcal{S}}_{n,t}}\|\rho-\sigma\|_2
\ge\delta.
\end{align}
By \cref{witweak}, since $\hat{\mathcal{S}}^{\theta}_{n,t}$ is convex, there exists a Hermitian operator $\tilde W$ such that
\begin{align}
\tr(\tilde W\rho)>\max_{\sigma\in\hat{\mathcal{S}}_{n,t}^{\vartheta}}\tr(\tilde W\sigma)+\delta^2.
\end{align}
We rescale $\tilde W$ into a valid element of $\mathcal{W}_{t,\nu}^{\vartheta}$ by defining 
\begin{align}
    W\coloneqq\frac{\tilde W}{\max_{\sigma\in\hat{\mathcal{S}}^{\vartheta}_{n,t}}\tr(\tilde W\sigma)}-\nu I.
    \end{align}
    Indeed, for any $H$ such that $\|H\|_2\le \nu$ and $\sigma\in\hat{\mathcal{S}}_{n,t}^{\theta}$, we have $\tr[(W+H)\sigma]\le \tr(W\sigma)+\nu\le 1-\nu+\nu\le 1$. Hence $W\in\mathcal{W}_{t,\nu}^{\vartheta}$. Evaluating $W$ on $\rho$ yields
\begin{align}
    \tr(W\rho)
    &=\frac{\tr(\tilde W\rho)}{\max_{\sigma\in\hat{\mathcal{S}}_{n,t}^{\vartheta}}\tr(\tilde W\sigma)}-\nu >1+\frac{\delta^2}{\max_{\sigma\in\hat{\mathcal{S}}_{n,t}^{\vartheta}}\tr(\tilde W\sigma)}-\nu.
    \nonumber
\end{align}
To tell apart \class{YES} and \class{NO} instant, we must have $\tr(W\rho)>1+\nu+\vartheta\sqrt{d(d+1)}$. It suffices that
\begin{align}
    \frac{\delta^2}{\max_{\sigma\in\hat{\mathcal{S}}_{n,t}^{\vartheta}}\tr(\tilde W\sigma)}>3\nu +\vartheta\sqrt{d(d+1)}.
\end{align}
Using the bound $\max_{\sigma\in\hat{\mathcal{S}}_{n,t}^{\vartheta}}\tr(\tilde W\sigma)\le\|\tilde W\|_2\le2$ (see \cref{witweak}), we obtain
\begin{align}
    \frac{\delta^2}{2}>3\nu +\vartheta\sqrt{d(d+1)}.
\end{align}
For simplicity, we choose $\nu\le \frac{\delta^{2}}{12}$ and $\vartheta<\frac{\delta^2}{4\sqrt{d(d+1)}}$, which ensures that the condition is satisfied. We now examine the runtime of the algorithm. Since we are solving a linear program, the runtime is $\class{poly}(|\mathcal{S}_{n,t}^{\vartheta}|,\log\nu^{-1})$ and therefore it is governed by an upper 
bound on the size of the set $\mathcal{S}_{n,t}^{\vartheta}$. As shown in \cref{lem:packing}, we have:
\begin{align}
\lvert\mathcal{S}_{n,t}^{\vartheta}\rvert
&\le 2^{2\log^2d+4\log d}\left(\frac{4}{\vartheta}\right)^{2^{2t}}  
\end{align}
Using the bound on $\vartheta$, this becomes
\begin{align}
\lvert\mathcal{S}_{n,t}^{\vartheta}\rvert
&\le 2^{2\log^2d+4\log d} \left(\frac{4 \sqrt{d(d+1)}}{\delta^2}\right)^{2^{2t}} \le 2^{2\log^2d+4\log d} \left(\frac{8 d}{\delta^2}\right)^{2^{2t}}\\
& =2^{2\log^2d+4\log d+  2^{2t}(3+\log d + 2\log\delta^{-1})}  .
\nonumber
\end{align}
 The above argument shows that computing $R_{t,\nu}^{\vartheta}(\rho)$ solves $\rho$-$\class{WMEM}^{1}_{\delta}(\hat{\mathcal{S}}_{n,t})$. By \cref{lem:onesideequivalentbothss}, the one-sided and two-sided weak membership problems are polynomially equivalent. In particular, there exists a polynomial relation
\begin{align}
\delta=\frac{\varepsilon^3}{2^{12}\cdot 3\cdot d^{12}}
\end{align}
such that solving the one-sided problem with precision $\delta$ yields a solution to $\rho$-$\class{WMEM}_{\varepsilon}(\hat{\mathcal{S}}_{n,t})$. Choosing $\nu\le \frac{\varepsilon^6}{2^{26}\cdot 3^3\cdot d^{24}}$ completes the reduction.
\end{itemize}

\end{proof}
\end{lemma}

\section{Additional results and applications}
This section first introduces the notion of completely $t$-doped stabilizer-preserving channels and explores in \cref{sec:ctdspchannels}, which generalises completely stabilizer-preserving channels (\cref{lem:characterisationstabilizerpreservingchannel}). We then present applications of these results to the computational complexity of channel classification in \cref{sec:channelclass}, and to the limitations of magic-state distillation in \cref{sec:exmsd}. In particular, we show that our hardness results imply both the intractability of certain channel classification tasks and the existence of non-stabilizer states that admit no efficient stabilizer-based distillation protocol.
\subsection{Completely $t$-doped stabilizer preserving channels and their characterisation}\label{sec:ctdspchannels}
In this section, we introduce a generalisation of completely stabilizer preserving channels. For stabilizer states, we have seen that there exist channels, called \emph{completely stabilizer-preserving channels}, that map states belonging to the stabilizer polytope into states that remain within the polytope. This notion can be generalised to \emph{$t$-doped stabilizer-preserving channels}, namely channels that map $\hat{\mathcal{S}}_{n}$ into $\hat{\mathcal{S}}_{n,t}$. Formally, we define the set of completely $t$-doped stabilizer-preserving channels as follows.

\begin{definition}[Completely $t$-doped stabilizer-preserving channels]\label{def:comptdop}
Let $\mathcal{E}$ be a quantum channel with an $n$-qubit input. We say that $\mathcal{E}$ is \emph{completely $t$-doped stabilizer-preserving} if
\begin{align}
(\mathcal{E} \otimes I)(\sigma) \in \hat{\mathcal{S}}_{n+m,t}
\end{align}
for all $\sigma \in \hat{\mathcal{S}}_{n+m}$ and for all $m \in \mathbb{N}$.
\end{definition}

In words, the action of the channel $\mathcal{E}$ on a stabilizer input of arbitrary size yields a convex combination of $t$-doped states which, for $t = O(\log n)$, can be efficiently simulated classically and, therefore, do not provide any computational advantage. 

We now provide a useful characterisation of completely $t$-doped stabilizer-preserving channels. To this end, we first introduce the following lemma.

\begin{lemma}[\cite{Beverland_2020}]\label{partialtracetdoped}
Let $\ket{\psi}\in\mathcal{S}_{n+m,t}$ be a $t$-doped stabilizer state vector on $n+m$ qubits. Then:
\begin{enumerate}
\item For any stabilizer state vector $\ket{\phi}\in\mathcal{S}_{m}$, the (unnormalised) vector $(I\otimes\bra{\phi})\ket{\psi}$ is a $t$-doped stabilizer state on $n$ qubits.
\item The reduced state $\tr_{m}(\ketbra{\psi}{\psi})$ lies in the convex hull of $t$-doped stabilizer states of $n$ qubits, i.e., $\tr_{m}(\ketbra{\psi}{\psi})\in\hat{\mathcal{S}}_{n,t}$.
\end{enumerate}
\end{lemma}

\begin{proof}
We work entirely at the level of the Pauli expansion. Let $d=2^{n+m}$ and write
\begin{align}
\ketbra{\psi}{\psi}=\frac{1}{d}\sum_{P\in\mathcal{P}_{n+m}} \alpha_P\, P,
\end{align}
where $\mathcal{P}_{n+m}$ is the $n+m$-qubit Pauli group modulo phases and $\alpha_P\coloneqq\tr(P\ketbra{\psi}{\psi})$.
By definition of a $t$-doped stabilizer state, $\ketbra{\psi}{\psi}$ can be written as
\begin{align}\label{eq:nullityfactor}
\ketbra{\psi}{\psi}
=\left(\frac{1}{2^t}\sum_{Q\in\bar{G}} \alpha_Q Q\right)
\left(\frac{1}{2^{n+m-t}}\sum_{R\in G} R\right),
\end{align}
where $G$ is the set of Pauli operators that stabilises $\ket{\psi}$, while $\bar{G}$ is a set of Pauli operators that do not stabilise $\ket{\psi}$, but have non-zero component.

Let $\ket{\phi}\in\mathcal{S}_m$ be a stabilizer state vector. Up to a local Clifford unitary on the last $m$ qubits, we may assume that $\ket{\phi}=\ket{0}^{\otimes m}$. Such a Clifford acts by permuting Pauli operators and therefore preserves the cardinality of the set of nonzero coefficients $\{\alpha_P\}$, so $t$ remains unchanged.
Now consider the unnormalised post-measurement operator
\begin{align}
\rho' \coloneqq (I\otimes\bra{0}^{\otimes m})\ketbra{\psi}{\psi}(I\otimes\ket{0}^{\otimes m}).
\end{align}
Using the Pauli expansion, we have
\begin{align}
\rho' = \frac{1}{d}\sum_{P_n\otimes P_m} \alpha_{P_n\otimes P_m}\,
P_n\, \langle 0 |^{\otimes m} P_m | 0 \rangle^{\otimes m},
\end{align}
where the sum runs over $P_n\in\mathcal{P}_n$ and $P_m\in\mathcal{P}_m$.
The scalar $\langle 0|^{\otimes m} | P_m | 0 \rangle^{\otimes m}$ vanishes unless $P_m$ is a tensor product of $I$ and $Z$ operators. Hence, only those Pauli components of $\ketbra{\psi}{\psi}$ of the form $P_n\otimes Z_m$ contribute, where $Z_m\in\{I,Z\}^{\otimes m}$.

Let $G$ be the stabilizer subgroup associated with the decomposition given in~\eqref{eq:nullityfactor}. The projection onto $\ket{0}^{\otimes m}$ restricts the Pauli expansion to a subset of the original $G$, and $\bar{G}$ but cannot generate new nonzero Pauli coefficients. Therefore, the number of distinct Pauli operators on the first $n$ qubits that appear with nonzero weight is at most $4^t$~\cite{guMagicinducedComputationalSeparation2024}. This shows that  $(I\otimes\bra{\phi})\ket{\psi}$ is an unnormalised $t$-doped stabilizer state on $n$ qubits.

Let $\{\ket{\phi_k}\}_k$ be any orthonormal stabilizer basis of the $m$-qubit Hilbert space. Then
\begin{align}
\tr_m(\ketbra{\psi}{\psi})
=\sum_k (I\otimes\bra{\phi_k})\ketbra{\psi}{\psi}(I\otimes\ket{\phi_k}).
\end{align}
Define $v_k\coloneqq(I\otimes\bra{\phi_k})\ket{\psi}$ and $p_k\coloneqq\|v_k\|_2^2$. If $p_k>0$, let $\ket{\psi_k}\coloneqq p_k^{-1/2}v_k$. By Claim 1, each nonzero $v_k$ is an non-normalised $t$-doped stabilizer state, and therefore each $\ket{\psi_k}\in\mathcal{S}_{n,t}$. This yields the convex decomposition
\begin{align}
\tr_m(\ketbra{\psi}{\psi})=\sum_k p_k\, \ketbra{\psi_k}{\psi_k},
\end{align}
which shows that $\tr_m(\ketbra{\psi}{\psi})\in\hat{\mathcal{S}}_{n,t}$.
\end{proof}

The next lemma gives a characterisation of completely $t$-doping stabilizer preserving channels in terms of their Choi states. 

\begin{lemma}[Completely $t$-doping stabilizer preserving channels in terms of their Choi states]\label{thm:choitdop}
Let $\mathcal{E}$ be a quantum channel with $n$-qubit input. Then $\mathcal{E}$ is completely $t$-doping stabilizer preserving if and only if its Choi state $\rho_{\mathcal{E}} \in \hat{\mathcal{S}}_{2n,t}$.
\begin{proof}

Following the proof in Ref.~\cite{seddonQuantifyingMagicMultiqubit2019}, we divide the proof into two parts. First, we show that in order to verify whether $\mathcal{E}$ is completely $t$-doping stabilizer preserving, it is sufficient to check the condition for $m = n$. Second, we show that this is equivalent to requiring that the Choi state lies in $\hat{\mathcal{S}}_{2n,t}$.

By linearity, it is sufficient to verify that $[\mathcal{E} \otimes I]$ maps pure stabilizer states in $\mathcal{S}_{n+m}$ to states in $\hat{\mathcal{S}}_{n+m,t}$. Let $m = n + m'$ and let $\ket{\sigma}$ be a stabilizer state vector on $n+m$ qubits. It is well known that any stabilizer state on a bipartite system $A|B$ is Clifford-unitarily equivalent to a tensor product of Bell pairs across the cut $A|B$ and local stabilizer states. Let $A$ denote the first $n$ qubits and $B$ the remaining $m$ qubits. Since there can be at most $n$ Bell pairs across the cut, there exists a local Clifford operation $C_B$ such that
\begin{equation}
\ket{\sigma}_{A,B} = (I_A \otimes C_B) \bigl( \ket{\sigma'}_{A,A'} \otimes \ket{\sigma''}_{B'} \bigr),
\end{equation}
where $A'$ denotes the first $n$ qubits of $B$ and $B'$ the remaining $m'$ qubits.

Suppose there exists a state vector  $\ket{\sigma}_{A,B}$ with $m' > 0$ such that $[\mathcal{E} \otimes I](\sigma_{A,B}) \notin \hat{\mathcal{S}}_{n+m,t}$. We show that there then exists a corresponding state vector $\ket{\sigma'}_{A,A'}$ supported on $2n$ qubits such that $[\mathcal{E} \otimes I](\sigma_{A,A'}) \notin \hat{\mathcal{S}}_{2n,t}$. Indeed, by denoting $\mathcal{C}_B(\cdot)\coloneqq C_B(\cdot)C_B^{\dag}$, we have that
\begin{align}
[\mathcal{E} \otimes I](\sigma_{A,B})
&= [\mathcal{E} \otimes I] \circ [I_A \otimes \mathcal{C}_B]\bigl(\sigma'_{A,A'} \otimes \sigma''_{B'}\bigr) \\
\nonumber
&= [I_A \otimes \mathcal{C}_B]\Big( [\mathcal{E} \otimes I](\sigma'_{A,A'}) \otimes \sigma''_{B'} \Big),
\end{align}
which implies that $[\mathcal{E} \otimes I](\sigma'_{A,A'}) \notin \hat{\mathcal{S}}_{2n,t}$, since tensoring with auxiliary stabilizer states and applying Clifford operations preserve membership in $\hat{\mathcal{S}}_{n+m,t}$.

We now show that the image of all pure stabilizer states under $[\mathcal{E} \otimes I]$ lies in $\hat{\mathcal{S}}_{2n,t}$ if and only if $\rho_{\mathcal{E}} \in \hat{\mathcal{S}}_{2n,t}$. The forward implication is immediate. For the reverse implication, let $W$ be a witness for $\hat{\mathcal{S}}_{n,t}$, that is, $\tr(W \psi) \le 0$ for all $\psi \in \mathcal{S}_{n,t}$ (see \cref{def:witK} with $\gamma=0$). We claim that $W \otimes \ketbra{\phi}{\phi}$ is a valid witness for $\hat{\mathcal{S}}_{n+m,t}$ for any $m$-qubit stabilizer state vector $\ket{\phi}$.
For any $\sigma \in \mathcal{S}_{n+m,t}$, we have
\begin{align}
\tr\bigl((W \otimes \ketbra{\phi}{\phi}) \sigma \bigr)
= \tr\bigl((W \otimes I)\tilde{\sigma}\bigr),
\end{align}
where $\tilde{\sigma} = (I \otimes \ketbra{\phi}{\phi}) \, \sigma \, (I \otimes \ketbra{\phi}{\phi})$. Since $\ket{\phi}$ is a stabilizer state, it follows from \cref{partialtracetdoped} that $\tilde{\sigma}$ is a non-normalized $t$-doped state. We further obtain
\begin{align}
\tr\bigl((W \otimes I)\tilde{\sigma}\bigr)
= \tr\bigl(W \, \tr_m(\tilde{\sigma})\bigr) \le 0,
\end{align}
where we have used that the partial trace of a $t$-doped state lies in $\hat{\mathcal{S}}_{n,t}$ by \cref{partialtracetdoped}, together with the assumption that $W$ is a witness for $\hat{\mathcal{S}}_{n,t}$.

For the final step, suppose for contradiction that $\rho_{\mathcal{E}} \in \hat{\mathcal{S}}_{2n,t}$, but there exists $\phi \in \mathcal{S}_{2n}$ such that $\rho = [\mathcal{E} \otimes I](\phi) \notin \hat{\mathcal{S}}_{2n,t}$. By \cref{witweak}, there exists a witness $W$ for $\hat{\mathcal{S}}_{2n,t}$ such that $\tr(W \rho) > 0$. Then
\begin{align}
0 < \tr(W \rho)
= 2^n \tr\bigl( \rho_{\mathcal{E} \otimes I} \, (W \otimes \phi^T) \bigr).
\end{align}
Since $W$ is a witness for $\hat{\mathcal{S}}_{2n,t}$ and the transpose of a stabilizer state is again a stabilizer state, the operator $2^n (W \otimes \phi^T)$ is a witness for $\hat{\mathcal{S}}_{4n,t}$. This implies that $\rho_{\mathcal{E} \otimes I} \notin \hat{\mathcal{S}}_{4n,t}$. Noting that $\rho_{\mathcal{E} \otimes I} = \rho_{\mathcal{E}} \otimes \phi^+_{2n}$ completes the proof.
\end{proof}
\end{lemma}

The set of $t$-doped stabilizer states thus provides an extension of the set of stabilizer states by also incorporating the use of non-Clifford operations, while still remaining classically simulable for $t = O(\log n)$. 

\subsection{Channel classification computational problem}\label{sec:channelclass}
In this section, we develop the first application of the results established in \cref{sec:wmemstab}. In particular, we present two corollaries of \cref{thm:nphardrhowmem} and \cref{lmm:wmemtdoped} in the context of channel classification. Starting from the first corollary, we consider the following computational task: given a description of a quantum channel \(\mathcal{E}\), determine whether it admits an efficient description within the stabilizer formalism. We formalise this task as a decision problem, asking whether the channel \(\mathcal{E}\) is completely stabilizer-preserving.

\begin{problem}[Completely stabilizer-preserving classification (\class{CSPC})]\label{channelproblem}
Let $\mathcal{E}$ be a quantum channel and let $\varepsilon>0$. Decide whether $\mathcal{E}$ satisfies:
\begin{itemize}
    \item For all $\sigma\in\mathcal{S}_{2n}$, $[\mathcal{E}\otimes \mathds{1}_n](\sigma)\in S(\hat{\mathcal{S}}_{2n},-\varepsilon)$; output \class{YES}.
    \item There exists $\sigma\in\mathcal{S}_{2n}$ such that $[\mathcal{E}\otimes \mathds{1}_n](\sigma)\notin S(\hat{\mathcal{S}}_{2n},\varepsilon)$; output \class{NO}.
\end{itemize}
An instance of this problem is denoted by $\class{CSPC}_{\varepsilon}$.
\end{problem}

\begin{corollary}[Superpolynomial complexity of $\class{CSPC}$]\label{cor1app}
Let $\hat{\mathcal{S}}_{2n}$ denote the stabilizer polytope on $2n$ qubits. Under the \emph{exponential time hypothesis}  (ETH) (see \cref{def:eth}), for any
\(\varepsilon \le\frac{1}{2^{30}3^3 d^{52}\log^{18} d }\),
the decision problem $\class{CSPC}_{\varepsilon}$ is not contained in the complexity class $\class{QP}^{2-\eta}$ for any $\eta>0$. 
\end{corollary}

\begin{proof}
The proof proceeds by reduction to a weak membership problem for the stabilizer polytope. By \cref{lem:characterisationstabilizerpreservingchannel}, deciding whether a quantum channel $\mathcal{E}$ is completely stabilizer-preserving can be reduced to a membership test for its Choi state $\rho_{\mathcal{E}}$. More precisely, $\mathcal{E}$ is completely stabilizer-preserving if and only if its Choi state satisfies $\rho_{\mathcal{E}}\in\hat{\mathcal{S}}_{2n}$, together with the constraint $\tr_B \rho_{\mathcal{E}} = I/d$, where $B$ denotes the second $n$-qubit subsystem.

This motivates the following decision problem. Given a state $\rho_{\mathcal{E}}$ satisfying $\tr_B \rho_{\mathcal{E}} = I/d$ and a parameter $\delta>0$, decide whether
\begin{itemize}
    \item for all states $\sigma$ such that $\|\rho_{\mathcal{E}}-\sigma\|_2 \le \delta$, we have $\sigma\in\hat{\mathcal{S}}_{2n}$; output \class{YES}.
    \item There exists a state $\sigma\in\hat{\mathcal{S}}_{2n}$ such that $\|\rho_{\mathcal{E}}-\sigma\|_2 > \delta$; output \class{NO}.
\end{itemize}
We denote this problem by $\rho_{\mathcal{E}}$-$\class{WMEM}_{\delta}(\hat{\mathcal{S}}_{2n})$.

We first show that $\rho$-$\class{WMEM}_{\delta}(\hat{\mathcal{S}}_n)$ for an arbitrary $n$-qubit 
state $\rho$ reduces to $\rho_{\mathcal{E}}$-$\class{WMEM}_{\delta}(\hat{\mathcal{S}}_{2n})$. Assume oracle access to $\rho_{\mathcal{E}}$-$\class{WMEM}_{\delta}(\hat{\mathcal{S}}_{2n})$.
Given $\rho$, define $\omega\equiv I/d$ and consider the state $\rho\otimes\omega$, which is a valid Choi state since $\tr_B(\rho\otimes\omega)= I/d$. Querying the oracle on $\rho\otimes\omega$, a \class{YES} instance implies that every state $\sigma$ satisfying $\|\rho\otimes\omega-\sigma\|_2\le\delta$ belongs to $\hat{\mathcal{S}}_{2n}$.
For any $\tilde{\sigma}$ such that $\|\rho-\tilde{\sigma}\|_2\le\delta$, we have
\begin{align}
\|\rho\otimes\omega-\tilde{\sigma}\otimes\omega\|_2
\le \|\rho-\tilde{\sigma}\|_2 \|\omega\|_2
= \delta d^{-1/2}
\le \delta,
\end{align}
and hence $\tilde{\sigma}\otimes\omega\in\hat{\mathcal{S}}_{2n}$. Since the partial trace preserves stabilizer membership, this implies $\tilde{\sigma}\in\hat{\mathcal{S}}_n$.

Conversely, for a \class{NO} instance we have $\min_{\sigma\in\hat{\mathcal{S}}_{2n}}\|\rho\otimes\omega-\sigma\|_2>\delta$. If there exists $\tilde{\sigma}\in\hat{\mathcal{S}}_n$ with $\|\rho-\tilde{\sigma}\|_2\le\delta$, then it would follow $\tilde{\sigma}\otimes\omega\in\hat{\mathcal{S}}_{2n}$, contradicting the inequality. Therefore, $\min_{\sigma\in\hat{\mathcal{S}}_n}\|\rho-\sigma\|_2>\delta$, which completes the reduction.

We now reduce $\class{CSPC}_{\varepsilon}$ to $\rho_{\mathcal{E}}$-$\class{WMEM}_{\delta}(\hat{\mathcal{S}}_{2n})$. Assume oracle access to $\class{CSPC}_{\varepsilon}$. Then:
\begin{itemize}
    \item \class{YES}. For all $\sigma\in{\mathcal{S}}_{2n}$ and all states $\tau$ such that $\|[\mathcal{E}\otimes I](\sigma)-\tau\|_2\le\varepsilon$, we have $\tau\in\hat{\mathcal{S}}_{2n}$. In particular, choosing $\sigma=\phi^+_{2n}$ shows that the Choi state $\rho_{\mathcal{E}}=[\mathcal{E}\otimes I](\phi^+_{2n})$ is a \class{YES} instance of $\rho_{\mathcal{E}}$-$\class{WMEM}_{\delta}(\hat{\mathcal{S}}_{2n})$, provided that $\delta\le\varepsilon$.
    \item \class{NO}. There exists $\sigma\in{\mathcal{S}}_{2n}$ such that $\min_{\tau\in\hat{\mathcal{S}}_{2n}}\|[\mathcal{E}\otimes I](\sigma)-\tau\|_2>\varepsilon$. By \cref{witweak}, there exists a witness $W'$ such that $\tr(W'[\mathcal{E}\otimes I](\sigma))>\max_{\tau\in\mathcal{S}_{2n}}\tr(W'\tau)+\varepsilon^2$. We redefine $W\coloneqq W'-\max_{\tau\in\mathcal{S}_{2n}}\tr(W'\tau)I$ so that $\max_{\tau\in\mathcal{S}_{2n}}\tr(W\tau)\le 0$.  Using the Choi representation, this can be written as
\begin{align}
\tr(W[\mathcal{E}\otimes I](\sigma))
= d\tr\!\big(\rho_{\mathcal{E}\otimes I}(W\otimes\sigma^T)\big)
> \varepsilon^2. 
\end{align}

As shown in 
Ref.\ \cite{seddonQuantifyingMagicMultiqubit2019}, whenever $\sigma$ is a pure stabilizer state, the operator $W\otimes\sigma$ defines a valid witness for $\hat{\mathcal{S}}_{4n}$.  Applying \cref{witweak} again yields $\|\rho_{\mathcal{E}\otimes I}-\sigma'\|_2>\varepsilon^2/d/\|W\|_2$ for all $\sigma'\in\hat{\mathcal{S}}_{4n}$. Choosing $\sigma'=\tilde{\sigma}\otimes\phi^+_{2n}$ and using $\rho_{\mathcal{E}\otimes I}=\rho_{\mathcal{E}}\otimes\phi^+_{2n}$, we obtain
\begin{align}
\|\rho_{\mathcal{E}}-\tilde{\sigma}\|_2 > \frac{\varepsilon^2}{d\|W\|_2}\ge \frac{\varepsilon^2}{d(d+2)},
\end{align}
where we have used the fact that $\|W\|_2\le \|W'\|_2+d\le d+2$ because $\max_{\tau}(W'\tau)\le 1$ (see \cref{witweak}).
\end{itemize}

This establishes the relation $\varepsilon=\sqrt{d(d+2)}\,\delta$. Combining the above reductions with \cref{thm:nphardrhowmem} concludes the proof. 
\end{proof}

A similar problem can be defined when asking the question of whether a channel $\mathcal{E}$ is completely $t$-doping stabilizer preserving~\cref{def:comptdop} or not. 
An analogous decision problem can be defined by asking whether a quantum channel $\mathcal{E}$ is completely $t$-doping stabilizer preserving (see \cref{def:comptdop}).

\begin{problem}[Completely $t$-doped stabilizer-preserving classification (\class{CTDSPC})]\label{channeldopingproblem}
Let $\mathcal{E}$ be a quantum channel and let $\varepsilon>0$. Decide whether $\mathcal{E}$ satisfies:
\begin{itemize}
    \item For all $\sigma\in\mathcal{S}_{2n}$, the state $[\mathcal{E}\otimes \mathds{1}_n](\sigma)$ belongs to $S(\hat{\mathcal{S}}_{2n,t},-\varepsilon)$; output \class{YES}.
    \item There exists $\sigma\in\mathcal{S}_{2n}$ such that $[\mathcal{E}\otimes \mathds{1}_n](\sigma)\notin S(\hat{\mathcal{S}}_{2n,t},\varepsilon)$; output \class{NO}.
\end{itemize}
An instance of this problem is denoted by $\class{CTDSPC}_{\varepsilon}$.
\end{problem}

\begin{corollary}[Superpolynomial complexity of completely $t$-doped stabilizer-preserving classification]\label{cor1appdoped}
Let $\hat{\mathcal{S}}_{2n,t}$ denote the $t$-doped stabilizer polytope on $2n$ qubits. Under the \emph{exponential time hypothesis} (ETH) (see \cref{def:eth}), for any 
\begin{equation}\varepsilon \le(2^{115} 3^3 (\log d)^{90} d^{118})^{-1}, 
\end{equation}
the decision problem $\class{CTDSPC}_\varepsilon$ is not contained in the complexity class $\class{QP}^{2-\eta}$ for any $\eta>0$.
\end{corollary}

\begin{proof}
The argument follows the same reduction strategy used for the completely stabilizer-preserving case. The only modification is that, instead of characterising completely stabilizer-preserving channels via their Choi states using \cref{lem:characterisationstabilizerpreservingchannel}, we invoke the Choi-state characterisation of completely $t$-doping-preserving channels given in \cref{thm:choitdop}.
In particular, a \class{YES} instance of $\class{CTDSPC}_{\varepsilon}$ implies that the Choi state of the channel lies within distance $\varepsilon$ of the $t$-doped stabilizer polytope $\hat{\mathcal{S}}_{2n,t}$, while a \class{NO} instance yields, via the same witness-based argument as before, a separation of order $\varepsilon/d$ from $\hat{\mathcal{S}}_{2n,t}$. This establishes a polynomial-time reduction from $\class{CTDSPC}_{\varepsilon}$ to $\rho$-$\class{WMEM}_{\varepsilon/d}(\hat{\mathcal{S}}_{2n,t})$.
\end{proof}

\subsection{Existence of super-polynomial magic-state distillation}\label{sec:exmsd}
In this section, we apply our main hardness result to the task of magic-state distillation.
Magic-state distillation protocols based on stabilizer operations constitute the standard route
to achieving universal quantum computation. A natural question is whether the mere presence of
non-stabilizer resources guarantees the existence of an efficient (in the Hilbert space dimension $d$) distillation procedure.

We show that this is not the case. More precisely, we prove that there exist non-stabilizer states
for which no stabilizer-based distillation protocol that can be found and executed within
polynomial time in $d$ is capable of producing a non-negligible amount of magic, even when allowed
polynomially many copies of the input state. This result establishes a fundamental computational
obstruction to magic-state distillation beyond information-theoretic considerations.

\begin{corollary}[Superpolynomial-time magic-state distillation]\label{cor2app}
There exists a (possibly mixed) non-stabilizer 
state $\tilde{\rho}$ on $n$ qubits such that, for every
stabilizer-based distillation protocol $\mathcal{E}$ satisfying the following conditions:
\begin{enumerate}
    \item $\mathcal{E}$ can be specified by an algorithm running in time $\class{poly}(d)$,
    \item $\mathcal{E}$ can be executed in time $\class{poly}(d)$,
\end{enumerate}
the probability that $\mathcal{E}$ outputs a magic state is at most $\frac{1}{\omega(\class{poly}(d))}$.
\end{corollary}

\begin{proof}
Assume, for the sake of contradiction, that the statement is false. By \cref{thm:nphardrhowmem}, there exists a non-stabilizer state $\tilde{\rho}$ such that
$\tilde{\rho}$ cannot be distinguished from the stabilizer polytope $\hat{\mathcal{S}}_n$
by any algorithm running in time $\class{poly}(2^{\log^{2-\eta}d})$ for any $\eta>0$.
Equivalently, for any such algorithm, the distinguishing advantage between $\tilde{\rho}$ and
any state in $\hat{\mathcal{S}}_n$ is negligible.

Under the above assumption, given $\tilde{\rho}$ one can proceed as follows. First, in time
$\class{poly}(d)$, compute the description of the corresponding distillation protocol
$\mathcal{E}_{\tilde{\rho}}$. Next, prepare $\class{poly}(d)$ copies of $\tilde{\rho}$ and apply
$\mathcal{E}_{\tilde{\rho}}$ to obtain an output state. By assumption, this procedure produces
a magic state with a probability of at least $1/\class{poly}(d)$. Repeating this entire experiment $\class{poly}(d)$ times yields an overall algorithm running in
time $\class{poly}(d)$ that outputs a magic state with a 
non-negligible probability whenever the
input state 
is $\tilde{\rho}$.
In contrast, 
if the input state is any $\sigma\in\hat{\mathcal{S}}_n$, then every
stabilizer-based operation preserves $\hat{\mathcal{S}}_n$, thereby distinguishing $\tilde{\rho}$ from the stabilizer polytope in time $\class{poly}(d)$. This contradicts
the computational indistinguishability guaranteed by \cref{thm:nphardrhowmem}, implying our initial assumption to be false.
\end{proof}

\end{document}